\documentclass[a4paper]{paper}
\usepackage[ margin=1in]{geometry}
\usepackage[utf8]{inputenc}
\usepackage[T1]{fontenc}
\usepackage{url}              
\usepackage{cite}             
\usepackage{amsfonts}
\usepackage{mathrsfs,amsthm, amssymb,stmaryrd,bm}

\usepackage{graphicx}
\newcommand{\mkv}{-\!\!\!\!\minuso\!\!\!\!-}
\usepackage{hyperref, comment, cite}
\usepackage{bbm}

\usepackage{enumitem} 

\usepackage{tikz}
\usetikzlibrary{positioning, arrows.meta}

\usetikzlibrary{fit} 
\usetikzlibrary{backgrounds} 
\newcommand{\rmr}{\mathrm{r}}

\newtheorem{theorem}{Theorem}
\newtheorem{lemma}{Lemma}
\newtheorem{proposition}{Proposition}
\newtheorem{corollary}{Corollary}

\newtheorem{definition}{Definition}
\newtheorem{remark}{Remark}
\newtheorem*{exmp*}{Example}
\newcommand*{\medcap}{\mathbin{\scalebox{1.5}{\ensuremath{\cap}}}}%

\usepackage[cmex10]{amsmath}  
\usepackage{mleftright}       
\mleftright                   

\usepackage{booktabs}         

\usepackage{titlesec}

\titleclass{\subsubsubsection}{straight}[\subsection]

\newcounter{subsubsubsection}[subsubsection]
\renewcommand\thesubsubsubsection{\thesubsubsection.\arabic{subsubsubsection}}

\titleformat{\subsubsubsection}
  {\normalfont\normalsize\bmseries}{\thesubsubsubsection}{1em}{}
\titlespacing*{\subsubsubsection}
{0pt}{3.25ex plus 1ex minus .2ex}{1.5ex plus .2ex}

\makeatletter
\renewcommand\paragraph{\@startsection{paragraph}{5}{\z@}%
  {3.25ex \@plus1ex \@minus.2ex}%
  {-1em}%
  {\normalfont\normalsize\bmseries}}
\renewcommand\subparagraph{\@startsection{subparagraph}{6}{\parindent}%
  {3.25ex \@plus1ex \@minus .2ex}%
  {-1em}%
  {\normalfont\normalsize\bmseries}}
\def\toclevel@subsubsubsection{4}
\def\toclevel@paragraph{5}
\def\toclevel@subparagraph{6}
\def\l@subsubsubsection{\@dottedtocline{4}{7em}{4em}}
\def\l@paragraph{\@dottedtocline{5}{10em}{5em}}
\def\l@subparagraph{\@dottedtocline{6}{14em}{6em}}
\makeatother

\begin{document}

\title{Dependence Balance and Capacity Bounds for Multiterminal Communication and Wiretap Channels}

\author{Amin Gohari and Gerhard Kramer}

%
\allowdisplaybreaks
\maketitle
\begin{abstract}
An information measure based on fractional partitions of a set is used to derive a general dependence balance inequality for communication. 
This inequality is used to obtain new upper bounds on reliable and secret rates for multiterminal channels.
For example, we obtain a new upper bound on the rate of shared randomness generated among terminals, a counterpart of the cut-set bound for reliable communication. 
The bounds for reliable communication use the concept of auxiliary receivers, and we show that they are optimized by Gaussian distributions for Gaussian channels. The bounds are applied to multiaccess channels with generalized feedback and relay channels, and improve the cut-set bound for scalar Gaussian channels. The improvement for Gaussian relay channels complements results obtained with other methods.\footnote{This work was presented in part at the 2023 and 2026 IEEE International Symposium on Information Theory.
The work of G.~Kramer was supported by the German Federal Ministry of Education and Research in the Program “Souverän.\ Digital.\ Vernetzt.” Joint Project 6G-Life under Project 16KISK002, and by the German Research Foundation (DFG) under project 509917421.}
\end{abstract}

\section{Introduction}
Mutual information quantifies the dependence of two random variables. One operational interpretation of mutual information is that it characterizes the maximum common randomness generated through interactive, public, and noiseless communication \cite{AhlswedeCsiszar1993,Maurer1993}, referred to as the \emph{source model}. A natural question is how to generalize mutual information to more than two random variables. For instance, one can define the \emph{shared information} as the maximum common randomness that multiple terminals can generate in the source model \cite{csiszar2008secrecy,5625626,chan2015multivariate,chan2014multiterminal}. For random variables $Y_1, Y_2, \dots, Y_k$, this leads to an information measure based on the \emph{fractional partition} $\lambda$ of the set $\{1, 2, \dots, k\}$; see \cite{csiszar2008secrecy}. We call this shared information \emph{the fractional partition multivariate information} or \emph{$\lambda$-multivariate information}.

The $\lambda$-multivariate information for $k\ge 3$ does not include the usual mutual information; hence, we define a mixed version that does. We further use $\lambda$-multivariate information to derive a new \emph{dependence balance (DB) inequality}. 
The original DB inequality was proposed for single-output two-way channels and multiaccess channels (MACs) with feedback in \cite{hekstra1989dependence} and was extended to discrete memoryless networks in \cite{kramer03,kramer2006dependence}.  Without feedback, the channel inputs are independent (conditioned on a time-sharing random variable) because they are functions of independent messages. However, feedback lets transmitters learn of each other's messages and generate statistically dependent inputs. DB constrains the mutual information of the channel inputs, i.e., each terminal ``must produce the dependence it consumes''~\cite[Sec.~IV]{hekstra1989dependence}. The new DB inequality with auxiliary receivers extends the bounds in \cite{hekstra1989dependence,
kramer03,kramer2006dependence,
gastpar06,gastpar06b,
tandon2011dependence,Sula-IT20} and is central to our proofs.

\subsection{Contributions and Organization}

This paper studies the following questions. How can $\lambda$-multivariate information be used to study common randomness generation and \emph{secrecy} for the source model? What happens for the \emph{channel model}, which replaces the noiseless public channels with a noisy network? What are the implications for reliable communication in noisy networks?

Our contributions can be summarized as follows. 
\begin{itemize} \itemsep 0pt
    \item We derive a new DB inequality with $\lambda$-multivariate information.
    \item For shared randomness generation:
    
    \begin{enumerate}[label=(\roman*)]
    \item We propose a general communication model for sharing randomness and derive an upper bound on secret key rates in terms of $\lambda$-multivariate information. The bound leverages the DB inequality and auxiliary receivers as in \cite{hekstra1989dependence,gohari2021outer}.
    \item We show the upper bound generalizes existing bounds for the source and channel models \cite{gohari2010information1,gohari2010information2}. For instance, the bound recovers the key agreement bound in \cite{Ardestanizadeh} for wiretap channels with a secure rate-limited feedback link. 
    \item The theory establishes a new upper bound on the shared randomness rates analogous to the cut-set bound for reliable communication \cite{aref80,ElGamal-NTC81,thomas2006elements,kramer03}. 
        \end{enumerate}
    \item For reliable communication over arbitrary multiterminal noisy networks:
    
    \begin{enumerate}[label=(\roman*)]
    \item We generalize the classic cut-set bound by including dependence balance constraints.
    \item For Gaussian multiterminal channels, we show that Gaussian distributions characterize the new bound. The bound thus requires optimizing only second-order statistics, as in the cut-set bound.
    \item We strengthen existing bounds for Gaussian MACs with generalized feedback and relay channels. The improvement for Gaussian relay channels complements the work in~\cite{gohari2021outer,el2022strengthened}.
    \end{enumerate}
\end{itemize}

This paper is organized as follows. Section~\ref{sec:prelim} introduces fractional partitions and $\lambda$-multivariate information and proves a general DB constraint. Section~\ref{sec:multiterminal-wiretap-channels} develops a new outer bound based on the DB constraint on the secret key rates.  Section~\ref{sec:mt-networks} similarly derives new capacity upper bounds for reliable communication. Section~\ref{sec:conclusions} concludes the paper.

\begin{remark} Prakash Narayan presented several open problems on $\lambda$-multivariate information in a plenary talk on ``Shared Information'' at the \emph{2024 IEEE Information Theory Workshop}, including the following.
\begin{itemize} \itemsep 0pt
    \item \emph{Noisy Interactive Communication:} The source-model key agreement framework assumes noiseless communication---can $\lambda$-multivariate information be utilized to study interactive communication over noisy channels? We address this question in Section~\ref{sec:multiterminal-wiretap-channels}.
    \item \emph{Network Coding Applications:} What is the operational significance of $\lambda$-multivariate information in network source and channel coding? We address this question in Section~\ref{sec:mt-networks}.
\end{itemize}
\end{remark}

\section{Preliminaries}
\label{sec:prelim}

The set $\{1,\cdots, k\}$ is denoted by $[k]$ and the cardinality of a set $\mathcal U$ is written as $|\mathcal{U}|$. Let $Y_{\mathcal{U}}$ denote $(Y_{i}: i\in\mathcal{U})$ so that $Y_{[k]}=(Y_1, Y_2, \cdots, Y_k)$. Let $Y^{i}$ denote the string $(Y_1, Y_2, \cdots, Y_{i})$. We similarly write
\begin{align}
    Y_{[u]}^{i} & = 
    \big( Y_{[u]1}, Y_{[u]2}, \cdots, Y_{[u]i} \big) \nonumber \\
    & = \big( Y_{11},\cdots,Y_{u1},\; Y_{12},\cdots,Y_{u2}, \; \cdots, Y_{1i},\cdots,Y_{ui} \big).
\end{align}
The expression $Y_{[u]}^i$ is the empty string if $i<1$.
We use the common notation $H(X)$ for discrete entropy, $I(X;Y)$ for mutual information, and $h(X)$ for differential entropy. We say $X\mkv Y\mkv Z$ forms a Markov chain if $I(X;Z|Y)=0$. Unless stated otherwise, we write ${\mathcal B}^c$ for the complement of the set $\mathcal B$, i.e., ${\mathcal B}^c=[k]-\mathcal B$.

\subsection{Fractional Partitions and Multivariate Information}

This section reviews a notion of multivariate information using fractional partitions.

\begin{definition}[Fractional Partition]\label{defFP}
Let $k\geq 2$ be a natural number. Let $\mathsf B$ be the collection of all non-empty proper subsets of $[k]$, i.e., sets $\mathcal B$ such that $\mathcal B\ne\emptyset$ and $\mathcal B\ne[k]$. A fractional partition of $[k]$ is a collection of non-negative weights $\lambda_{\mathcal B}$, $\mathcal{B} \in \mathsf{B}$, such that
\begin{align}
\sum_{{\mathcal B}\in \mathsf B:\,i\in {\mathcal B}}\lambda_{\mathcal B}&=1,\quad\forall\, i\in[k].
\label{eqnEE2}
\end{align}
\end{definition}

The $k$ constraints \eqref{eqnEE2} should not be confused with a constraint on the sum over all $\lambda_\mathcal{B}$. 
For example, for the set $[2]=\{1,2\}$ we have $\lambda_{\{1\}}=\lambda_{\{2\}}=1$. Similarly, for the set $[3]=\{1,2,3\}$ and
\begin{align} 
    \lambda_{\{1,2\}}=\lambda_{\{3\}}=1
\end{align}
we have $\lambda_{\mathcal B}=0$ for $\mathcal{B}\notin\{\{1,2\},\{3\}\}$. This fractional partition corresponds to the partition $\{1,2\}\cup\{3\}$. On the other hand, the choice
\begin{align} \label{eq:frac-example}
    \lambda_{\{1,2\}} = \lambda_{\{1,3\}} = \lambda_{\{2,3\}} = 1/2
\end{align}
is a fractional partition that does not correspond to any partition or linear combination of partitions.

Note that $\lambda_{\mathcal B}$ was defined for $\emptyset\subsetneq\mathcal{B}\subsetneq [k]$. Alternatively, one may include $\mathcal{B}=\emptyset$ and $\mathcal{B}=[k]$ by requiring $\lambda_\emptyset=\lambda_{[k]}=0$, and we use this convention below. Observe that $\sum_{{\mathcal B}}\lambda_{\mathcal B}\geq 1$ in any fractional partition. 

\begin{definition}[Multivariate Information] \label{defMI} Let $k\geq 2$ be a natural number. Let $(\lambda_\mathcal B: \mathcal B\in \mathsf{B})$ be a fractional partition of $[k]$. The $\lambda$-multivariate information of variables $X_i$, $i\in [k]$, conditioned on a variable $T$ is
\begin{align}
    I_{\lambda}(X_1;X_2;\cdots;X_k|T)
    & = H(X_{[k]}|T) - \sum\nolimits_{\mathcal B\subsetneq [k]}\lambda_{\mathcal B} H(X_{\mathcal B}|X_{{\mathcal B}^c},T) \nonumber \\
    & = \left(1-\sum\nolimits_{\mathcal B\subsetneq [k]}\lambda_{\mathcal B}\right) H(X_{[k]}|T) + \sum\nolimits_{\mathcal B\subsetneq [k]}\lambda_{\mathcal B} H(X_{{\mathcal B}^c}|T).
    \label{eq:defMI}
\end{align}
\end{definition}

For example, for $k=2$ we recover the conditional mutual information $I_\lambda(X_1;X_2|T)=I(X_1;X_2|T)$. For $k=3$ and the choice \eqref{eq:frac-example} we obtain (see Appendix~\ref{subsec:appendixA1})
\begin{align}
    I_{\lambda}(X_1;X_2;X_3)
    & = H(X_1,X_2,X_3) - \frac12 \big( H(X_1,X_2|X_3) + H(X_1,X_3|X_2) + H(X_2,X_3|X_1) \big) \nonumber \\
    & = \frac12 \big( H(X_1) + H(X_2) + H(X_3) - H(X_1,X_2,X_3) \big).
\end{align}
Further basic properties of $I_{\lambda}$ are discussed in Appendix \ref{sec:appendixA}. 

\begin{remark} Definition~\ref{defMI} can be traced to \cite[Equation 6]{csiszar2008secrecy} (that refers to \cite{Chung-etal-SIAM88,madiman-barron-IT07}) where the minimum of $I_{\lambda}$ over all fractional partitions $\lambda$ is related to the secret key rate. This minimum is called multivariate information in \cite{chan2015multivariate} and shared information in \cite[Remark 3.11]{narayan2016multiterminal}; see also \cite{5625626}. We instead consider $ I_{\lambda}$ for each fixed choice of $\lambda$ as a multivariate information. 
\end{remark}

If $T$ is independent of $X_{[k]}$, we have
\begin{align}\label{eq:Ilambda}
    I_{\lambda}(X_1;X_2;\cdots;X_k)
    = \left(1-\sum\nolimits_{\mathcal B}\lambda_{\mathcal B}\right) H(X_{[k]}) + \sum\nolimits_{\mathcal B}\lambda_{\mathcal B} H(X_{{\mathcal B}^c}).
\end{align}
Since $\lambda_{\mathcal B}\geq 0$ and $\sum_{\mathcal B}\lambda_{\mathcal B}\geq 1$, the coefficient of $H(X_{[k]})$ is non-positive while the coefficient of $H(X_{{\mathcal B}})$ for any proper subset $\mathcal B$ is non-negative. Consequently,  we cannot express
\begin{equation}
   I(X_1;X_2)=H(X_1)+H(X_2)-H(X_1,X_2)
\end{equation}
as special case of $I_{\lambda}(X_1;X_2;\cdots;X_k)$ if $k>2$, as the coefficient of $H(X_1,X_2)$ is non-negative. We are thus motivated to consider a weighted version of $I_{\lambda}$ for different subsets of the variables.

\begin{definition}\label{defnneqw2}
   Let $k\geq 2$ be a natural number. For every subset $\mathcal{U}\subseteq[k]$ of cardinality $|\mathcal{U}|\geq 2$, take a fractional partition $\lambda^{\mathcal{U}}_{\mathcal{B}}$ for indices in $\mathcal{U}$ such that
\begin{align}
    \sum_{{\mathcal B\subsetneq \mathcal{U}}:\,i\in {\mathcal B}}\lambda^{\mathcal{U}}_{\mathcal{B}}&=1,\quad\forall\, i\in\mathcal{U}.
    \label{eqnEE2d}
\end{align}
    Writing $\mathcal{U}=\{i_1, i_2,\cdots, i_u\}\subseteq[k]$, the multivariate information using the fractional partition $\lambda^{\mathcal{U}}_{\mathcal{B}}$ is
\begin{align}
    I_{\lambda^{\mathcal{U}}}(X_{i_1};X_{i_2};\cdots;X_{i_u})
\end{align}
    where now the $\mathcal{B}^c$ in \eqref{eq:Ilambda} are the complements of $\mathcal{B}$ in $\mathcal{U}$.
    Let $\omega_\mathcal{U}$ be a non-negative weight assigned to set $\mathcal{U}$ such that $\sum_\mathcal{U}\omega_\mathcal{U}=1$.
    Then the $(\omega, \lambda^{\cdot})$ multivariate information among $X_1, \cdots, X_k$ is defined as
\begin{align}
    I_{\omega, \lambda^{\cdot}}(X_1;X_2;\cdots;X_k) \triangleq \sum\nolimits_\mathcal{U}\omega_\mathcal{U}\cdot I_{\lambda^\mathcal{U}}(X_{i_1};X_{i_2};\cdots;X_{i_u}).
\end{align}
\end{definition}

Note that setting 
$\omega_\mathcal{U}=0$ for $\mathcal{U}\neq \mathcal{U}^*$, and $\omega_{\mathcal{U}^*}=1$ recovers the ordinary $\lambda$-multivariate information on the subset $\mathcal{U}^*$. Thus, the weights $\omega_\mathcal{U}$ allow defining a multivariate information that specializes to $I(X_1;X_2)$ by setting $\omega_{\{1,2\}}=1$ and $\omega_\mathcal{U}=0$ for $\mathcal{U}\neq \{1,2\}$.

\begin{remark}
We utilize the $(\omega, \lambda^{\cdot})$ multivariate information to obtain tight upper bounds for the source model with silent terminals in Section~\ref{subsubsec:kterminalssec} and Appendix~\ref{appendixC}.
\end{remark}

\section{A General Dependence Balance Inequality}
\label{subsec:DB-bound}

The following bound is key to proving our main results. 

\begin{lemma}[General DB constraint]
\label{lemma1}
Let $k\geq 2$ and $n\geq 1$ be natural numbers. Consider random variables $W_{i}, X_{ij}, Y_{ij}$ and $Z_j$ for $i\in [k], j\in[n]$ satisfying
\begin{align}
    X_{ij} &= f_{ij}(W_i, Y_{i[j-1]}), \qquad i\in[k],~j\in [n] \label{eqnconditionx}
\end{align}
for some functions $f_{ij}(\cdot)$. Consider a set $\mathcal{U}\subseteq [k]$ with $|\mathcal{U}|=u\geq 2$ and assume the Markov chains
\begin{align}
    W_{\mathcal U} Y_{\mathcal U}^{j-1}
    \mkv X_{[k]j} Z^{j-1}
    \mkv Y_{\mathcal Uj} Z_{j},
    \quad j\in [n].
    \label{eq:DB-Markov-chain2}
\end{align}
Write $\mathcal{U}=\{i_1, i_2, \cdots, i_u\}$ and let $\lambda=(\lambda_\mathcal{B}:\mathcal{B}\subsetneq \mathcal{U})$ be a fractional partition of $\mathcal{U}$. We have
\begin{align}
    &I_{\lambda}(W_{i_1}Y_{i_1}^n;W_{i_2}Y_{i_2}^n;\cdots;W_{i_u}Y_{i_u}^n|Z^n) -I_{\lambda}(W_{i_1};W_{i_2};\cdots;W_{i_u}) \nonumber \\
    & \le \sum_{j\in[n]} \bigg[I_{\lambda}(X_{i_1j}Y_{i_1j};X_{i_2j}Y_{i_2j};\cdots;X_{i_uj}Y_{i_uj}|Z^{j-1},Z_j) - I_{\lambda}(X_{i_1j};X_{i_2j};\cdots;X_{i_uj}|Z^{j-1}) \nonumber \\
    &\qquad \qquad - \left(1-\sum\nolimits_{\mathcal B\subsetneq \mathcal{U}}\lambda_{\mathcal B}\right) I(X_{[k]j} ; Z_j Y_{\mathcal{U}j} | Z^{j-1} X_{\mathcal{U}j})\Bigg]
    \label{eq:DB-bound2}
\end{align}
where we recall that $X_{\mathcal{U}j}=(X_{i_1j},\cdots,X_{i_uj})$ and similarly for $Y_{\mathcal{U}j}$. Observe that choosing $\mathcal U=[k]$ makes the last mutual information term in \eqref{eq:DB-bound2} vanish.
\end{lemma}

\begin{proof}
One may assume $\mathcal{U}=[u]$ without loss of generality. Now expand
\begin{align}
    & I_{\lambda}(W_{1}Y_{1}^n ; W_{2}Y_{2}^n; \cdots; W_{u}Y_{u}^n \big| Z^n) - I_{\lambda}(W_{1};W_{2}; \cdots; W_{u}) \nonumber \\
    & \overset{(a)}{=} \sum\nolimits_{j \in [n]} \bigg[ I_{\lambda}(W_{1}Y_{1}^j ; W_{2}Y_{2}^j ; \cdots ; W_{u}Y_{u}^j \big| Z^j ) - I_{\lambda}(W_{1}Y_{1}^{j-1} ; W_{2}Y_{2}^{j-1} ; \cdots; W_{u}Y_{u}^{j-1} \big| Z^{j-1}) \bigg] \nonumber \\
    & \overset{(b)}{=} 
    \sum\nolimits_{j \in [n]} \bigg[ I_{\lambda}(W_{1}Y_{1}^jX_{1j} ; W_{2}Y_{2}^jX_{2j} ; \cdots ; W_{u}Y_{u}^jX_{uj} \big| Z^j )\nonumber \\
    & \qquad\qquad \qquad\color{black}- I_{\lambda}(W_{1}Y_{1}^{j-1}X_{1j} ; W_{2}Y_{2}^{j-1}X_{2j} ; \cdots; W_{u}Y_{u}^{j-1}X_{uj}\big| Z^{j-1}) \bigg] \nonumber
    \\&\color{black}\overset{(c)}{=} 
    \sum_{j\in[n]} \bigg[I_{\lambda}(X_{1j}Y_{1j};X_{2j}Y_{2j};\cdots;X_{uj}Y_{uj}|Z^{j-1},Z_j) - I_{\lambda}(X_{1j};X_{2j};\cdots;X_{uj}|Z^{j-1}) \nonumber \\
    &\qquad \qquad\color{black} - \left(1-\sum\nolimits_{\mathcal B\subsetneq [u]}\lambda_{\mathcal B}\right)  I(W_{[u]} Y_{[u]}^{j-1} ; Z_jY_{[u]j} \big| Z^{j-1} X_{[u]j} )
    \nonumber \\
    &\qquad \qquad\color{black}-\sum\nolimits_{\mathcal B} \lambda_{\mathcal B}I(W_{{\mathcal B}^c}Y_{{\mathcal B}^c}^{j-1} ; Z_j Y_{{\mathcal B}^cj} \big| Z^{j-1} X_{\mathcal{B}^cj})\bigg]
    \label{eqn15just}\\
    & \color{black}\overset{(d)}{\leq} 
    \sum_{j\in[n]} \bigg[I_{\lambda}(X_{1j}Y_{1j};X_{2j}Y_{2j};\cdots;X_{uj}Y_{uj}|Z^{j-1},Z_j) - I_{\lambda}(X_{1j};X_{2j};\cdots;X_{uj}|Z^{j-1}) \nonumber \\
    & \qquad \qquad \color{black}- \left(1-\sum\nolimits_{\mathcal B\subsetneq [u]}\lambda_{\mathcal B}\right) I(X_{[k]j} ; Z_jY_{[u]j} | Z^{j-1} X_{[u]j})\Bigg]
    \label{eqn16justn}
    \end{align}\color{black}
where step $(a)$ follows by telescoping and step $(b)$ by~\eqref{eqnconditionx}. Step $(c)$ follows by writing the expression as
\begin{align}
    &\sum\nolimits_{j \in [n]} \bigg[ \big(1-\sum\nolimits_{\mathcal B} \lambda_{\mathcal B}\big) \bigg( H( W_{[u]}Y_{[u]}^{j} X_{[u]j} \big| Z^{j})
    - H(W_{[u]}Y_{[u]}^{j-1} X_{[u]j} \big| Z^{j-1} ) \bigg) \nonumber \\
    & \qquad\qquad\qquad +\sum\nolimits_{\mathcal B} \lambda_{\mathcal B} \bigg( H(W_{{\mathcal B}^c}Y_{{\mathcal B}^c}^{j} X_{{\mathcal B}^cj} \big| Z^j)
    - H(W_{{\mathcal B}^c}Y_{{\mathcal B}^c}^{j-1} X_{{\mathcal B}^cj} \big| Z^{j-1}) \bigg)\bigg]
    \label{eq:DB-proof-expand1}
\end{align}
and expanding the first and second entropy differences in \eqref{eq:DB-proof-expand1} as
\begin{align}
    H( X_{[u]j} Y_{[u]j} \big| Z^{j})
    - H( X_{[u]j} \big| Z^{j-1})
    - I(W_{[u]} Y_{[u]}^{j-1} ; Z_jY_{[u]j} \big| Z^{j-1} X_{[u]j} )
    \label{eq:DB-proof-expand2} \\
    H(X_{\mathcal{B}^cj} Y_{\mathcal{B}^cj} \big| Z^j)
    - H(X_{\mathcal{B}^cj} \big| Z^{j-1})
    - I(W_{{\mathcal B}^c}Y_{{\mathcal B}^c}^{j-1} ; Z_j Y_{{\mathcal B}^cj} \big| Z^{j-1} X_{\mathcal{B}^cj}).
    \label{eq:DB-proof-expand3}
\end{align}
Step $(d)$ follows by upper bounding the term $I(W_{[u]} Y_{[u]}^{j-1} ; Z_jY_{[u]j} \big| Z^{j-1} X_{[u]j} )$ with
\begin{align}
    I(W_{[u]} Y_{[u]}^{j-1} X_{[k]j} ; Z_jY_{[u]j} \big| Z^{j-1} X_{[u]j} ) 
    = I(X_{[k]j} ; Z_jY_{[u]j} \big| Z^{j-1} X_{[u]j} )
\end{align}
where the equality follows by~\eqref{eq:DB-Markov-chain2}, and by using the non-negativity of mutual information. 
\end{proof}
\color{black}
\subsection{Discussion}
\label{subsec:auxiliary-discussion}

\subsubsection{Auxiliary Random Variables and Receivers}
\label{subsubsec:auxiliary-rvs}

The dependence balance bound in Lemma \ref{lemma1} involves auxiliary random variables $Z_j$, $j\in[n]$. Roughly speaking, auxiliary random variables can be categorized as either ``transmitter-side'' or ``receiver-side''. The former were introduced by Cover for coding theorems and by Gallager \cite{gal74} for converse proofs, in both cases for broadcast channels. The adjective ``auxiliary'' is misleading for coding theorems because the variables usually represent concrete coded symbols, e.g., in superposition coding. In Gallager-type converse proofs, however, the auxiliary variables often involve past and/or future variables of the problem and may lack an intuitive interpretation.

Receiver-side auxiliary variables instead represent \emph{new or artificial} receivers that do not necessarily exist in the original problem. These receivers do not communicate or influence the messages, nor do they decode; they may be viewed as silent observers. For example, Ozarow found the rate-distortion region of the Gaussian two-description problem \cite{ozarow1980source} by introducing ``an artificial [random variable that] ... plays no apparent intuitive role in the encoding/decoding process, [but] provides the crucial lower bound in the proof." A notable special class of auxiliary receivers is \emph{genies} or enhanced receivers. For example, genies help to analyze the capacity of Gaussian interference channels, where treating interference as noise characterizes the sum capacity under specific weak interference conditions; see~\cite{skc09,mok09,anv09} and also~\cite{kra04,elk11}. Other examples of auxiliary receivers are given in \cite{hekstra1989dependence,gohari2010information,wagner2008improved,liu2009capacity,Yu2018}. 

\subsubsection{Capacity Region Surface}
\label{subsec:capacity-surface}

Let $\mathsf{C}\big(p(y_{[k]}|x_{[k]})\big)$ be the capacity region of a network with the channel $p(y_{[k]}|x_{[k]})$. The paper \cite{gohari2021outer} used auxiliary receivers to study the surface of $\mathsf{C}\big(p(y_{[k]}|x_{[k]})\big)$. More precisely, the curvature of $\mathsf{C}\big(p(y_{[k]}|x_{[k]})\big)$ with respect to variations in $p(y_{[k]}|x_{[k]})$ is based on comparing 
\[
   \mathsf{C}\big(p(y_{[k]}|x_{[k]})\big) \quad \text{and} \quad \mathsf{C}\big(p(z_{[k]}|x_{[k]})\big)
\]
for two distinct channels, $p(y_{[k]}|x_{[k]})$ and $p(z_{[k]}|x_{[k]})$. Treating $p(z_{[k]}|x_{[k]})$ as an \emph{auxiliary channel}, one can derive an outer bound on $\mathsf{C}\big(p(y_{[k]}|x_{[k]})\big)$ if the following conditions are met:
\begin{itemize} \itemsep 0pt
    \item The gap between $\mathsf{C}\big(p(y_{[k]}|x_{[k]})\big)$ and $\mathsf{C}\big(p(z_{[k]}|x_{[k]})\big)$ can be characterized;
    \item A suitable outer bound on $\mathsf{C}\big(p(z_{[k]}|x_{[k]})\big)$ is available.
\end{itemize}
For instance, genie-aided proofs select $p(z_{[k]}|x_{[k]})$ as an enhanced version of $p(y_{[k]}|x_{[k]})$ so that $\mathsf{C}\big(p(y_{[k]}|x_{[k]})\big)$ is a subset of $\mathsf{C}\big(p(z_{[k]}|x_{[k]})\big)$, and so $p(z_{[k]}|x_{[k]})$ belongs to a class of channels for which the capacity can be characterized. However, the auxiliary receiver $Z_{[k]}$ need not be an enhanced version of $Y_{[k]}$. This perspective, combined with additional insights (such as modified manipulations of the past or future of the auxiliary receiver variable), lets one systematically derive outer bounds for broadcast, interference, and relay channels \cite{gohari2021outer}; see \cite{wen2024new,chen2025differential} for recent developments.

\subsubsection{Two Choices}
\label{subsec:two-choices}

We consider only auxiliary receivers and make the following choices; see \cite{gohari2021outer}.
\begin{itemize} \itemsep 0pt
    \item \emph{Modify Inactive Terminals:} We modify only the output variables $Y_i$ of inactive terminals, i.e., those with input alphabets having $|\mathcal{X}_i| = 1$. Specifically, we require $Z_i = Y_i$ for all terminals $i$ where $X_i$ is constant. This ensures that any encoding strategy designed for $p(y_{[k]}|x_{[k]})$ applies to $p(z_{[k]}|x_{[k]})$. For example, in key agreement problems with a passive eavesdropper, replacing the eavesdropper’s channel output with an auxiliary variable preserves compatibility with existing encoding schemes. We refer to Section \ref{sec:multiterminal-wiretap-channels}, which introduces the auxiliary receiver $T$.
    \item \emph{Output Enhancement:} Choose $Z_i$ as an \emph{enhanced} version of $Y_i$, e.g., so that $Y_i$ is a function of $Z_i$. Encoding strategies for $p(y_{[k]}|x_{[k]})$ then remain valid for $p(z_{[k]}|x_{[k]})$ since terminals may discard the enhanced information in $Z_i$. 
    Section \ref{sec:mt-networks} generalizes this approach by using multiple auxiliary receivers, rather than relying on a single one. 
\end{itemize}

We apply Lemma \ref{lemma1} with these choices. Specifically, Section \ref{sec:mt-networks} uses output enhancement to improve the cut-set bound for scalar Gaussian relay channels, rather than modifying inactive terminals as in \cite{gohari2021outer}. Note that \cite{gohari2021outer} used both approaches to develop outer bounds for broadcast channels. One may also combine the two ideas above by selecting multiple auxiliary receivers in Sections \ref{sec:multiterminal-wiretap-channels} and \ref{sec:mt-networks}.
\begin{remark}
  An example of how a sequence of auxiliary receivers can improve bounds is given in  \cite{Yuval2024}. See also Remark \ref{remarkbeyond} below for a recent attempt to go beyond the above two types of auxiliary receivers.
\end{remark}

\subsubsection{Continuous Random Variables}
\label{subsubsec:continuous-rvs}

Definition \ref{defMI} writes multivariate information using discrete entropy, which illustrates certain symmetries of the measure. {\color{black} More generally, for continuous or mixed discrete-continuous random variables, one may define $\lambda$-multivariate information as follows:
\begin{align}
    I_\lambda(X_1;X_2;\cdots;X_k | T)
    & = \sum\nolimits_{{\mathcal B\subsetneq [k]}} \sum\nolimits_{i\in {\mathcal B}} \lambda_{\mathcal B}\,
    I(X_i;X_{\mathcal B^c} | T, X^{i-1}).
    \label{eq:defMI-2}
\end{align}
The expression \eqref{eq:defMI-2} is less intuitive than \eqref{eq:defMI}, but the advantage is that it involves mutual information terms only. For example, for discrete random variables, we recover \eqref{eq:defMI} via
\begin{align}
    H(X_{[k]})
    & =\sum\nolimits_{i} \left( \sum\nolimits_{{\mathcal B}: i\in {\mathcal B}} \lambda_{\mathcal B} \right) H(X_i|X^{i-1}) \nonumber \\
       & =\sum\nolimits_{{\mathcal B}} \sum\nolimits_{i\in {\mathcal B}} \lambda_{\mathcal B}\, H(X_i|X^{i-1}) \nonumber \\
    & = \sum\nolimits_{{\mathcal B}} \sum\nolimits_{i\in {\mathcal B}} \lambda_{\mathcal B}\, H(X_i|X_{[i-1]\cap \mathcal B}, X_{\mathcal B^c}) +\sum\nolimits_{{\mathcal B}} \sum\nolimits_{i\in {\mathcal B}} \lambda_{\mathcal B}\, I(X_i;X_{[i-1]\cap \mathcal B}, X_{\mathcal B^c}|X^{i-1})\nonumber \\
    &=\sum\nolimits_{{\mathcal B}} \lambda_{\mathcal B} \, H(X_{\mathcal B}|X_{{\mathcal B}^c})+\sum\nolimits_{{\mathcal B}} \sum\nolimits_{i\in {\mathcal B}} \lambda_{\mathcal B}\, I(X_i; X_{\mathcal B^c}|X^{i-1}).
    \label{eq:defMI-2a}
\end{align}

The generalization of Lemma~\ref{lemma1} to mixed discrete--continuous random variables follows by applying similar proof steps and establishing the identity in~\eqref{eqn15just} using the definition~\eqref{eq:defMI-2}. More precisely, to establish~\eqref{eqn15just}, we wish to show
\begin{align}
    &\sum\nolimits_{j \in [n]} \bigg[ I_{\lambda}(W_{1}Y_{1}^jX_{1j} ; W_{2}Y_{2}^jX_{2j} ; \cdots ; W_{u}Y_{u}^jX_{uj} \big| Z^j )\nonumber \\
    & \qquad\qquad \qquad\color{black}- I_{\lambda}(W_{1}Y_{1}^{j-1}X_{1j} ; W_{2}Y_{2}^{j-1}X_{2j} ; \cdots; W_{u}Y_{u}^{j-1}X_{uj}\big| Z^{j-1}) \bigg] \nonumber
    \\&=\sum_{j\in[n]} \bigg[I_{\lambda}(X_{1j}Y_{1j};X_{2j}Y_{2j};\cdots;X_{uj}Y_{uj}|Z^{j-1},Z_j) - I_{\lambda}(X_{1j};X_{2j};\cdots;X_{uj}|Z^{j-1}) \nonumber \\
    &\qquad \qquad\color{black} - \left(1-\sum\nolimits_{\mathcal B\subsetneq [u]}\lambda_{\mathcal B}\right)  I(W_{[u]} Y_{[u]}^{j-1} ; Z_jY_{[u]j} \big| Z^{j-1} X_{[u]j} )
    \nonumber \\
    &\qquad \qquad\color{black}-\sum\nolimits_{\mathcal B} \lambda_{\mathcal B}I(W_{{\mathcal B}^c}Y_{{\mathcal B}^c}^{j-1} ; Z_j Y_{{\mathcal B}^cj} \big| Z^{j-1} X_{\mathcal{B}^cj})\bigg].\label{neqtp2}
\end{align}
See Appendix \ref{appendixNew2} for details.}

\section{Multiterminal Wiretap Channels}
\label{sec:multiterminal-wiretap-channels}

Consider a memoryless network with the channel $p(y_{[k]}|x_{[k]})$ where the $X_i$ and $Y_i$ are the respective channel inputs and outputs of the $i$-th transceiver for $i\in[k]$. In this paper, we are interested in \emph{common/shared randomness} that can be generated among the terminals. Common randomness includes reliable communication since messages sent between terminals can be interpreted as producing shared randomness. Common randomness may also be generated through correlated channel noise. 

We include secrecy through a passive wiretapper with channel output $z$ and write the $(k+1)$-terminal network model as $p(y_{[k]},z|x_{[k]})$. The common randomness should be kept hidden from the wiretapper, i.e., the common randomness shared among a group of terminals can serve as a secret key. For example, the problem of generating multiple keys among different sets of terminals has been studied in~\cite{zhang2017multiple}. While capacity results are known for special cases, e.g., \cite{zhang2017multi},  no general outer bound on the trade-off of key rates is known.  We provide an upper bound that unifies several results in the literature. Some results involve channels with feedback; for example, we study the source and channel models that include noiseless public feedback links as in \cite[Chapter 22]{elk11}. To incorporate feedback, we consider a model where, in addition to the main channel $p(y_{[k]},z|x_{[k]})$, there are $L$ parallel channels $q_\ell(y_{[k]},z|x_{[k]})$ for $\ell=1,2,\cdots, L$ that the legitimate terminals can use. 

\subsection{System Model}
\label{subsec:system-model}

The main channel $p(y_{[k]},z|x_{[k]})$ has input alphabets $\mathcal{X}_i$ and output alphabets $\mathcal{Y}_i$ and $\mathcal{Z}$. The parallel channels $q_\ell(y_{[k]},z|x_{[k]})$ have input alphabets $\mathcal{X}_i^{(\ell)}$ and output alphabets $\mathcal{Y}_i^{(\ell)}$ and $\mathcal{Z}^{(\ell)}$, $\ell\in[L]$, where $x_i\in \mathcal{X}_i^{(\ell)}$, $y_i\in \mathcal{Y}_i^{(\ell)}$ and $z\in\mathcal{Z}^{(\ell)}$.\footnote{By writing $p(y_{[k]},z|x_{[k]})$ and $q_\ell(y_{[k]},z|x_{[k]})$, the input/output alphabet sets of the channels are formally the same. This restriction is unnecessary for the proofs, i.e., different channels can have different input/output alphabets.
}
For instance, a noiseless public discussion channel can be modeled by the parallel channel $Y_1=\cdots=Y_k=Z = X_{[k]}$.

A code of length $n$ is defined as follows: at time instance $j\in[n]$, the $i$-th legitimate terminal uses a local (private) random variable $W_i$ and transmits the symbol
\begin{align}
X_{ij} &= f_{ij}(W_i, Y_{i[j-1]}), \qquad i\in [k],\; j\in [n] \label{eqnDefCode}
\end{align}
over the main channel $p(y_{[k]},z|x_{[k]})$ or over one of the parallel channels  $q_\ell(y_{[k]},z|x_{[k]})$; the type of channel (main or parallel) used at time $j$ is known and fixed a priori. Here, $n$ is the number of transmissions and $f_{ij}(\cdot)$ is the encoding function of terminal $i$ at time $j$, and $Y_{ij}$ is the channel output of terminal $i$ at time $j$. The random string $Y_{i[j-1]}$, sometimes denoted by $Y_i^{j-1}$, is the string of past outputs of terminal $i$ at time $j$.
Suppose the main channel is used $m\leq n$ times during the $n$ transmissions, while the channel
$q_\ell(y_{[k]}, z|x_{[k]})$ is used $m_\ell$ times for $\ell\in[L]$. Thus, we have
$m+\sum_{\ell=1}^Lm_\ell=n$. We call \begin{align}
    \alpha_\ell = m_\ell \big/ m
    \label{def-alpha}
\end{align}
the \emph{rate of channel use} for $q_\ell(y_{[k]},z|x_{[k]})$.

After transmission, every subset $\mathcal V\subseteq [k]$ of terminals ($|\mathcal V|\geq 2$) generates a shared key of rate $R_{\mathcal V}$, i.e., the $i$-th terminal generates 
\begin{align}
    S_{i,\mathcal{V}} = g_{i,\mathcal{V}}(W_i, Y_{i[n]})\label{eqgenS}
\end{align}
for every $\mathcal{V}$ containing $i$ where  $S_{i,\mathcal{V}}\in [2^{mR_\mathcal{V}}]$. For an $(n,\epsilon)$ code, we require existence of random variables
\begin{align}\label{eqnAlphabet}
    S_{\mathcal{V}} \text{ with alphabet } [2^{mR_\mathcal{V}}], \quad \mathcal{V}\subseteq[k], |\mathcal V|\geq 2
\end{align}
that are (almost) mutually independent of each other and $Z^n$. Specifically, the following uniformity, reliability, independence, and security conditions must hold for the $S_{\mathcal{V}}$ and $S_{i,\mathcal{V}}$:
\begin{subequations}
\begin{align}
    \frac 1m H(S_{\mathcal{V}})
    & \geq R_\mathcal{V}-\epsilon
    \label{eqnCond-1} \\
    \mathbb{P}\left[\medcap_{i\in\mathcal{V}}
    \{S_{i,\mathcal{V}} = S_{\mathcal{V}}\}\right]
    & \geq 1-\epsilon
    \label{eqnCond-2} \\
    \frac1m \left( -H(\{S_{\mathcal{V}}:\mathcal{V}\subseteq [k]\}) + \sum\nolimits_{\mathcal{V}\subseteq [k]} H(S_{\mathcal{V}}) \right)
    & \leq \epsilon
    \label{eqnCond-3} \\
    \frac 1m I(\{S_{\mathcal{V}}:\mathcal{V}\subseteq [k]\};Z^n)
    & \leq \epsilon.
    \label{eqnCond-4}
\end{align}
\end{subequations}
Note the normalization factor $1/m$ rather than $1/n$. The non-negative number $R_\mathcal{V}$ is called the \emph{group secret key rate} for the subset $\mathcal{V}$. Given channel-use rates $\alpha_\ell\geq 0$ for $\ell\in [L]$, we are interested in the rates $R_\mathcal{V}$ that can be achieved for any $\epsilon>0$ as $m\rightarrow\infty$.

An important special case is when there is only one subset of terminals -- without loss of generality taken to be the first $u$ terminals -- that generate
the secret keys, i.e., $R_\mathcal{V}=0$ when $\mathcal{V}\neq [u]$. Thus, terminals $u+1, u+2, \cdots, k$ do not generate secret keys but can participate as \emph{helper terminals}. If we wish to keep the secret key private from the helper terminals, their outputs could be included in the eavesdropper's $Z$. 

Our model includes several special cases. 
\begin{itemize} \itemsep 0pt
\item \emph{Source model:} consider $k=2$ and let the main channel $X_1$ and $X_2$ be constants. The source model follows by adding a channel for public discussion with $\alpha_1\rightarrow\infty$, meaning public discussion is unrestricted. Similarly, the multiuser case studied in \cite{csiszar2008secrecy,gohari2010information1} 
is a special case of our model. The capacity of the source model is open in general; see \cite{8437913,8995629,vippathalla2021secret}.
\item \emph{Channel model:} consider $k=2$ and let the main channel $X_2$ and $Y_1$ be constants. The channel model follows by adding a channel for public discussion with $\alpha_1\rightarrow\infty$. Similarly, the multiuser case in \cite{csiszar2008secrecy,gohari2010information1} 
is a special case of our model. Also, we can include the MAC models in~\cite{csiszar2012secrecy,tyagi2013secret}, where each legitimate terminal is either a receiver or a transmitter, by setting the alphabets of $X_i$ or $Y_i$ to be constants.
\item \emph{Wiretap channels with a private feedback link:} A secure rate-limited feedback link as in \cite{Ardestanizadeh} is included by choosing $k=2$ and a parallel channel where $Y_2$ and $Z$ are constant while $p(y_1|x_{2})$ has a capacity equal to the desired feedback rate.
\item The channel model of \cite{9366098}  reduces to the model considered here if the parallel channels are public and available to all parties. 
\end{itemize}

\subsection{Special Case: Common Key with Free Public Discussion}
We begin with a special case and generalize in the next section. Consider $R_{\mathcal{V}} = 0$ for $\mathcal{V} \neq [k]$, i.e., only the entire set of terminals aims to create a common key $S_{[k]}$. The objective is to maximize the key rate $R_{[k]}$. Moreover, suppose free, noiseless public discussion is available to all terminals, modeled by a parallel channel with $Y_1 = \cdots = Y_k = Z = X_{[k]}$ and $\alpha_1\rightarrow\infty$.
Here, $X_{[k]}$ refers to the parallel channel inputs. For the main channel inputs, we consider two special cases.

\textbf{Case of $|\mathcal{X}_i| = 1$:}  
When $|\mathcal{X}_i| = 1$, i.e., the $X_i$'s are constants, the model reduces to the source model key agreement problem \cite{csiszar2008secrecy,gohari2010information1}. For $k=2$ users with one-way public communication from the first terminal, the secrecy capacity of the source model is given in \cite{AhlswedeCsiszar1993}.

\begin{definition}
Given a joint distribution $p_{A,B,C}$, the one-way secrecy capacity in the source model problem is defined as
\begin{align}
  S(A \rightarrow B \| C) = \max \left[ I(V; B | U) - I(V; C | U) \right]
\end{align}
where the maximum is over Markov chains $(U, V) \mkv A \mkv (B, C)$ satisfying cardinality bounds
$$|\mathcal{U}|\leq |\mathcal{A}|,\qquad |\mathcal{V}|\leq |\mathcal{A}|.$$
It is known that $S(A \rightarrow B \| C) \leq I(A; B | C)$ and $S(A \rightarrow B \| C) = 0$ when $B = C$.
\end{definition}

Let $S(Y_1; Y_2; \cdots; Y_k \| Z)$ be the supremum of the key rates $R_{[k]}$ using free public discussion. The current best upper bound for the source model and $k=2$ users \cite{7976410} is as follows. Let $T$ be an auxiliary receiver with conditional distribution $P_{T | Y_1, Y_2, Z}$.
 The paper \cite{7976410} showed that
\begin{align}
    S(Y_1; Y_2 \| Z) &\leq S(Y_1; Y_2 \| T) + S(Y_1, Y_2 \rightarrow T \| Z) .
    \label{eqn24}
\end{align}
Since $S(Y_1; Y_2 \| T)\leq I(Y_1;Y_2|T)$, we obtain the following bound for the source model and $k=2$ users:
\begin{align}
S(Y_1; Y_2 \| Z) &\leq I(Y_1; Y_2 | T) + S(Y_1, Y_2 \rightarrow T \| Z) \nonumber \\
&= I(Y_1; Y_2 | T) + \max_{(V, U) \mkv (Y_1, Y_2) \mkv (T, Z)} \left[ I(V; T | U) - I(V; Z | U) \right].
\label{eq:R2-upper-bound}
\end{align}
By using the arguments in \cite{7976410}, or Theorem \ref{thm1} in this paper, one can generalize \eqref{eq:R2-upper-bound} to any number of users, any conditional distribution $P_{T | Y_{[k]}, Z}$, and any fractional partition $\lambda$:
\begin{align}
S(Y_1; Y_2; \cdots; Y_k \| Z) \leq I_{\lambda}(Y_1; Y_2; \cdots; Y_k | T) + S(Y_{[k]} \rightarrow T \| Z). \label{eqnSM1}
\end{align}

Next, suppose $Z = \emptyset$ is a constant. If all terminals participate in public discussion, \cite{csiszar2008secrecy} shows that
\begin{align}
S(Y_1; Y_2; \cdots; Y_k \| \emptyset) = \min_{\lambda} I_{\lambda}(Y_1; Y_2; \cdots; Y_k). \label{eqnSZc}
\end{align}
Thus, the upper bound \eqref{eqnSM1} is tight when $T$ is chosen as a constant. The key capacity is also known if only a subset of parties participates in public discussion; see \cite[Theorem 6]{gohari2010information1} and Appendix \ref{appendixC} for the explicit expression. However, the capacity does not have the simple form given in \eqref{eqnSZc}. Nevertheless, after some manipulation (see Appendix \ref{appendixC}), we rewrite the expression from \cite[Theorem 6]{gohari2010information1} using $I_{\omega, \lambda^{\cdot}}(X_1; X_2; \cdots; X_k)$ as in Definition \ref{defnneqw2}. Our general upper bound involves $I_{\omega, \lambda^{\cdot}}(X_1; X_2; \cdots; X_k)$ rather than $I_{\lambda}(X_1; X_2; \cdots; X_k)$, as we aim to derive an upper bound that is tight for the source model with silent terminals in Section~\ref{subsubsec:kterminalssec} and Appendix~\ref{appendixC}.

\textbf{Case of arbitrary $|\mathcal{X}_i|$:} Permitting any $\mathcal{X}_i$ includes the channel model. Our main result in Theorem \ref{thm1} implies that for any fractional partition $\lambda$ and any conditional distribution $P_{T | X_{[k]}, Y_{[k]}, Z}$, the key rate is bounded from above by
\begin{align}
    \max \Big[ &I_{\lambda}(X_1 Y_1; X_2 Y_2; \cdots; X_k Y_k | T) - I_{\lambda}(X_1; X_2; \cdots; X_k) + S(X_{[k]} Y_{[k]} \rightarrow T \| Z) \Big] \label{eqnSM1Add}
\end{align}
where the maximum over all $p(x_{[k]})$. This formula generalizes \eqref{eqnSM1}.

The term $I_{\lambda}(X_1 Y_1; X_2 Y_2; \cdots; X_k Y_k | T) - I_{\lambda}(X_1; X_2; \cdots; X_k)$ can be interpreted as a DB term. The DB constraint was originally formulated for communication over MACs with feedback \cite{hekstra1989dependence}, which is a different setting from the source or channel models. Our work establishes a connection between these models.
\begin{remark}
    \label{remarkbeyond}
    It is interesting to relate \eqref{eqn24} to the discussion regarding the role of auxiliary receivers in Section \ref{subsec:capacity-surface} to characterize the surface of \( p_{Y_1,Y_2,Z} \mapsto S(Y_1; Y_2 \| Z) \).\end{remark}

\begin{remark}
    The following generalization of \eqref{eqn24} is conjectured in
    \cite[Section III]{abin2025source}: for any $p_{Y_1,Y_2,Z,Y'_1,Y'_2,T}$ we have
    \begin{align*}
    S(Y_1; Y_2 \| Z)-S(Y'_1; Y'_2 \| T) \leq S(Y_1, Y_2 \rightarrow T \| Z)+I(Y'_2T;Y_1|Y'_1)+I(Y'_1T;Y_2|Y'_2)+I(Y_1;Y_2|Y'_1Y'_2T).
\end{align*}
\end{remark}

\subsection{General Outer Bound}

Consider an auxiliary variable $T$ with alphabet $\mathcal{T}$ defined by a conditional distribution
$q(t\,|\,y_{[k]},z,x_{[k]})$. We refer to $T$ as an auxiliary receiver.

\begin{definition}
Consider a $(\omega, \lambda^{\cdot})$ in Definition \ref{defnneqw2} and a conditional distribution $q(t,y_{[k]},z\,|\,x_{[k]})$. Define
\begin{align}
    V_{\omega, \lambda^{\cdot}}(q(t,y_{[k]},z|x_{[k]}))
    = & \max\Big[ I_{\omega, \lambda^{\cdot}}
    (X_{1}Y_{1};X_{2}Y_{2};\cdots;X_{k}Y_{k}|T)- I_{\omega, \lambda^{\cdot}}(X_{1};X_{2};\cdots;X_{k}) \nonumber \\
    & \qquad\qquad - \sum\nolimits_{\mathcal{U}}\omega_\mathcal{U}\left(1-\sum\nolimits_{\mathcal B\subsetneq \mathcal{U}}\lambda^{\mathcal{U}}_{\mathcal B}\right)I(X_{[k]};Y_{\mathcal{U}},T|X_{\mathcal{U}}) \nonumber \\
    & \qquad\qquad+S(X_{[k]}Y_{[k]}\rightarrow T\|Z) \Big]
\end{align}
where the maximum is over all $p(x_{[k]})$.
\end{definition}
\begin{remark}
    One may replace $S(X_{[k]}Y_{[k]}\rightarrow T\|Z)$ by its upper bound $I(X_{[k]}Y_{[k]};T|Z)$ to obtain a simple upper bound on $V_{\omega, \lambda^{\cdot}}(q(t,y_{[k]},z\,|\,x_{[k]}))$.
\end{remark}

\begin{remark}\label{rmk7}
Consider $T=Z$, $\omega_{[k]}=1$, and $\omega_\mathcal{U}=0$ when $\mathcal{U}\neq [k]$. Let $\lambda$ be a fractional partition corresponding to $[k]$. We obtain
\begin{equation}
    V_{\omega, \lambda^{\cdot}}(q(t,y_{[k]},z|x_{[k]}))
    = \max \Big[ I_{\lambda}(X_{1}Y_{1};X_{2}Y_{2};\cdots;X_{k}Y_{k}|Z) -  I_{\lambda}(X_{1};X_{2};\cdots;X_{k} )\Big]
\end{equation}
where the maximum is over all $p(x_{[k]})$.
\end{remark}

We can now state our main upper bound.

\begin{theorem}\label{thm1}
Consider the main channel $p(y_{[k]},z\,|\,x_{[k]})$ and  $L$ parallel channels $q_\ell(y_{[k]},z\,|\,x_{[k]})$, $\ell\in[L]$, along with channel use rates $\alpha_\ell$
in \eqref{def-alpha}. Take auxiliary receivers $p(t\,|\,y_{[k]},z,x_{[k]})$ and $q_\ell(t\,|\,y_{[k]},z,x_{[k]})$ $(\ell=1,2,\cdots, L)$ for the main and parallel channels, respectively. 
 The group secret key rates $R_\mathcal{V}$ for $\mathcal{V}\subseteq[k]$ are achievable only if for any $(\omega, \lambda^{\cdot})$ (see Definition \ref{defnneqw2}) we have
\begin{align}
    & \sum_{\mathcal V}R_\mathcal V\left( \sum_{\mathcal{U}: \mathcal{V}\cap\mathcal{U} \neq \emptyset} \omega_\mathcal{U} \left(1 -\sum_{\mathcal B \subsetneq \mathcal{U}:~ \mathcal{V}\cap(\mathcal U - \mathcal B) = \emptyset} \lambda^\mathcal{U}_{\mathcal B}\right)\right) \nonumber \\
    & \leq V_{\omega, \lambda^{\cdot}} \big( p(y_{[k]},z|x_{[k]})p(t|x_{[k]},y_{[k]},z) \big) + \sum\nolimits_{\ell\in[L]} \alpha_\ell V_{\omega, \lambda^{\cdot}} \big( q_\ell(y_{[k]},z|x_{[k]}) q_\ell(t|x_{[k]},y_{[k]},z) \big). \label{eqntheB}
\end{align}
For the inner sum, if there is no $\mathcal B\subsetneq\mathcal{U}$ such that $ \mathcal{V}\cap(\mathcal U-\mathcal B)=\emptyset$, we take the sum to be zero.
\end{theorem}

Theorem~\ref{thm1} is proved in Section \ref{subsec:proof1} using Lemma \ref{lemma1} in Section \ref{subsec:DB-bound}. Intuitively, the expression
$$V_{\omega, \lambda^{\cdot}}\big(
p(y_{[k]},z|x_{[k]})p(t|x_{[k]},y_{[k]},z)\big)$$
is an upper bound on the contribution of the main channel to the total secret key, while
$$V_{\omega, \lambda^{\cdot}}\big(
q_\ell(y_{[k]},z|x_{[k]}) q_\ell(t|x_{[k]},y_{[k]},z) \big)$$
is an upper bound on the contribution of the $\ell$-th parallel channel.

\begin{remark}
The upper bound has a symmetric form in terms of $p(y_{1}, y_2, \cdots, y_k,z|x_1, x_2, \cdots, x_k)$ and the parallel channels $q_\ell(y_{1}, y_2, \cdots, y_k,z|x_1, x_2, \cdots, x_k)$. Suppose $\alpha_\ell\rightarrow\infty$, i.e., the parallel channel can be used as often as desired. Then, using \eqref{eqntheB} when $\alpha_\ell\rightarrow\infty$, one is restricted to $\omega, \lambda^{\cdot}$ for which
\begin{align}
    V_{\omega, \lambda^{\cdot}}\big(q_\ell(y_{[k]},z|x_{[k]})q_\ell(t|x_{[k]},y_{[k]},z)\big)=0.\label{eqnacf}
\end{align}
One can see this restriction explicitly when we specialize the general upper bound to the source model with silent terminals in Appendix \ref{appendixC}. If we consider noiseless or noisy parallel channels of finite capacity and assume $\alpha_\ell$ to be finite, our choice of $\omega, \lambda^{\cdot}$ is no longer required to satisfy \eqref{eqnacf}.
\end{remark}

\begin{remark}
Consider an auxiliary receiver $T$ described by $q(t|y_{[k]},z,x_{[k]})$. Then $V_{\omega, \lambda^{\cdot}}(q(t,y_{[k]},z|x_{[k]}))$ is computable if the $X_i$'s have finite alphabets. Thus, any choice of auxiliary receivers leads to a computable upper bound. Computing the best possible lower bound requires minimizing $V_{\omega, \lambda^{\cdot}}(q(t,y_{[k]},z|x_{[k]}))$ over all $q(t|y_{[k]},z,x_{[k]})$. The optimization will be an inf-max problem, and no cardinality bound on the alphabet of $T$ is known, even for the source model problem; see \cite{7976410}. 
\end{remark}

\begin{corollary}\label{corr9} Consider $\omega_{[k]}=1$ and $\omega_\mathcal{U}=0$ when $\mathcal{U}\neq [k]$.  Let $\lambda$ be a fractional partition for $[k]$. Then the group secret key rates $R_\mathcal{V}$ for $\mathcal{V}\subseteq[k]$ are achievable only if 
\begin{align}
    & \sum_{\mathcal V} R_\mathcal V \left(1-\sum_{\mathcal B:\,\mathcal V\subseteq \mathcal{B} \subsetneq [k]}\lambda_{\mathcal B}\right) \nonumber \\
    & \leq V_{\omega, \lambda^{\cdot}} \big( p(y_{[k]},z|x_{[k]}) p(t|x_{[k]},y_{[k]},z)\big) + \sum\nolimits_{\ell\in[L]} \alpha_\ell V_{\omega, \lambda^{\cdot}} \big( q_\ell(y_{[k]},z|x_{[k]}) q_\ell(t|x_{[k]},y_{[k]},z)\big).
    \label{eq:corr1}
\end{align}
where
\begin{align}
    V_{\omega, \lambda^{\cdot}}(q(t,y_{[k]},z|x_{[k]}))
    = \max \Big[  & I_{\lambda}(X_{1}Y_{1};X_{2}Y_{2};\cdots;X_{k}Y_{k}|T) 
    \nonumber \\
    & - I_{\lambda}(X_{1};X_{2};\cdots;X_{k}) + S(X_{[k]}Y_{[k]}\rightarrow T\|Z) \Big]
\end{align}
and the maximum is over all $p(x_{[k]})$.

\end{corollary}

The upper bound in Theorem~\ref{thm1} is rather general. 
Section~\ref{subsec:previous-results} demonstrates its versatility by recovering several known upper bounds as special cases, e.g., the bounds \eqref{eqnSM1} and \eqref{eqnSM1Add}.
We further use Theorem \ref{thm1} to derive a novel upper bound for a new setting in Section~\ref{subsec:new-bound-randomness-generation}.

\subsection{Proof of Theorem \ref{thm1}}
\label{subsec:proof1}

We first derive some consequences of \eqref{eqnCond-1}-\eqref{eqnCond-4}.
Observe that \eqref{eqnAlphabet} gives $|S_{\mathcal{V}}|=2^{mR_\mathcal{V}}$.
For any collection $\mathsf{B}'$ of subsets of $[k]$, we have
\begin{align}
    \frac1m H(\{S_{\mathcal{V}}:\mathcal{V}\in \mathsf{B}'\})
    & = \frac1m \Big[ H(\{S_{\mathcal{V}}:\mathcal{V}\subseteq [k]\})
    - H( \{S_{\mathcal{V}}:\mathcal{V}\notin \mathsf{B}'\}
    \,|\, \{S_{\mathcal{V}}:\mathcal{V}\in \mathsf{B}'\}) \Big] \nonumber \\
    & \overset{(a)}{\ge} 
    \frac1m \left(\sum\nolimits_{\mathcal V:\mathcal{V}\subseteq [k]} H(S_{\mathcal{V}}) \right) - \epsilon
    - \sum\nolimits_{\mathcal V:\mathcal{V}\notin \mathsf{B}'} R_{\mathcal V} \nonumber \\
    & \overset{(b)}{\ge} \left(\sum\nolimits_{\mathcal V:\mathcal{V}\in \mathsf{B}'} R_{\mathcal V} \right) - 2^k \epsilon
    \label{eq:Bprimebound}
\end{align}
where step $(a)$ uses \eqref{eqnCond-3} and $|S_{\mathcal{V}}|=2^{mR_\mathcal{V}}$, and step $(b)$ uses \eqref{eqnCond-1}. Next, \eqref{eqnCond-2} gives
\begin{align}
    \mathbb{P}[\cup_{i\in\mathcal V} \{ S_{i,\mathcal{V}} \ne S_{\mathcal{V}} \} ] < \epsilon \quad\implies\quad 
    \mathbb{P}[S_{i,\mathcal{V}} \ne S_{\mathcal{V}}] < \epsilon, \quad \forall\,i\in\mathcal V
\end{align}
and hence, Fano's inequality gives
\begin{align}
    H( \{S_{\mathcal{V}}:\mathcal{V}\in \mathsf{B}'\}
    \,|\, \{S_{i,\mathcal{V}}:\mathcal{V}\in \mathsf{B}'\})
    & \le m k(\epsilon), \quad i\in\mathcal V 
    \label{eq:Fanobound1} \\
    H( \{S_{i,\mathcal{V}}: i\in\mathcal{V}, \mathcal{V}\in \mathsf{B}'\}
    \,|\, \{S_{\mathcal{V}}:\mathcal{V}\in \mathsf{B}'\})
    & \le m k'(\epsilon)
    \label{eq:Fanobound2}
\end{align}
where $k(\epsilon)\rightarrow0$ and $k'(\epsilon)\rightarrow0$ as $\epsilon\rightarrow0$. Let $\mathbf{M}_i=(
S_{i,\mathcal{V}}:\mathcal{V}\cap\{i\}\ne\emptyset)$ be the string of keys generated by the $i$-th terminal. We have
\begin{align}
    \frac1m I(\mathbf{M}_{[k]};Z^n)
    & \le \frac1m I(\mathbf{M}_{[k]}, \{S_{\mathcal{V}}:\mathcal{V}\subseteq [k]\} ; Z^n) \nonumber \\
    & \le \frac1m I(\{S_{\mathcal{V}}:\mathcal{V}\subseteq [k]\} ; Z^n)
    + \frac1m H(\mathbf{M}_{[k]}|\{S_{\mathcal{V}}:\mathcal{V}\subseteq [k]\}) \nonumber \\
    & \overset{(a)}{\le} \epsilon + k'(\epsilon)
    \label{eq:security2}
\end{align}
where step $(a)$ follows by \eqref{eqnCond-4}, and by \eqref{eq:Fanobound2} with $\mathsf{B}'$ being all subsets of $[k]$.

Next, for the set $\mathcal{U}=\{i_1, i_2, \cdots, i_u\}$, let $X_{\mathcal{U}j}=(X_{i_1j},X_{i_2j},\cdots, X_{i_uj})$ and similarly for $Y_{\mathcal{U}j}$.
For the $j$-th time instance, let $P_{T_j|X_{[k]j},Y_{[k]j},Z_j}$ be the auxiliary channel equal to $p(t|x_{[k]},y_{[k]},z)$ if we use the main channel at time instance $j$, or $q_\ell(t|x_{[k]},y_{[k]},z)$ if we use the $\ell$-th parallel channel at time instance $j$.
Define $T^n$ via
\begin{align}
    P_{T^n|X_{[k]}^n,Y_{[k]}^n,Z^n}
    = \prod\nolimits_{j \in [n]} P_{T_j|X_{[k]j},Y_{[k]j},Z_j} .
\end{align}

Since $\mathbf{M}_i=(
S_{i,\mathcal{V}}:\mathcal{V}\cap\{i\}\ne\emptyset)$ is the string of keys generated by the $i$-th terminal, the collection of keys
$\mathbf{M}_{\mathcal{U}}$ should be the target keys $S_{\mathcal{V}}$ for all $\mathcal{V}$ satisfying $\mathcal{V}\cap\mathcal{U}\ne\emptyset$, which we write as $S_{\mathcal{V}:\mathcal{V}\cap\mathcal{U}\ne\emptyset}$, and with the target rate
$\sum_{\mathcal V:\mathcal{V}\cap\mathcal{U}\neq\emptyset}R_\mathcal V$. We have
\begin{align}
    \frac1m H(\mathbf{M}_{\mathcal{U}})
    & = \frac1m \left[ H\big(\mathbf{M}_{\mathcal{U}},S_{\mathcal{V}:\mathcal{V}\cap\mathcal{U}\ne\emptyset} \big)
    - H\big( S_{\mathcal{V}:\mathcal{V}\cap\mathcal{U}\ne\emptyset} | \mathbf{M}_{\mathcal{U}} \big) 
    \right] \nonumber \\
    & \overset{(a)}{\geq}  \left(\sum\nolimits_{\mathcal V:\mathcal{V}\cap\mathcal{U}\neq\emptyset} R_\mathcal V \right) - k_1(\epsilon)
\end{align}
where $k_1(\epsilon)\rightarrow0$ as $\epsilon\rightarrow0$, and step $(a)$ follows from \eqref{eq:Bprimebound} and \eqref{eq:Fanobound1}. Similarly, for any $\mathcal{B}\subsetneq \mathcal{U}$, we have
\begin{align}
    \frac1m H(\mathbf{M}_{\mathcal B}|\mathbf{M}_{\mathcal{U}-{\mathcal B}})
    & \leq \frac1m H\big( \mathbf{M}_{\mathcal B}, S_{\mathcal{V}: \mathcal{V}\cap\mathcal{B}\ne\emptyset,\mathcal{V}\cap(\mathcal{U}-\mathcal{B})=\emptyset} \,|\, \mathbf{M}_{\mathcal{U}-{\mathcal B}} \big) \nonumber \\
    & \le \frac1m \left[ 
    H\big( S_{\mathcal{V}: \mathcal{V}\cap\mathcal{B}\ne\emptyset,\mathcal{V}\cap(\mathcal{U}-\mathcal{B})=\emptyset} \big)
    + H\big( \mathbf{M}_{\mathcal B}
    \,|\, \mathbf{M}_{\mathcal{U}-{\mathcal B}}, S_{\mathcal{V}: \mathcal{V}\cap\mathcal{B}\ne\emptyset,\mathcal{V}\cap(\mathcal{U}-\mathcal{B})=\emptyset} \big) \right]
    \nonumber \\
    & \overset{(a)}{\leq}  \left(\sum_{\mathcal V:~\mathcal V\cap \mathcal{B}\neq \emptyset, \mathcal V\cap (\mathcal U-\mathcal{B})=\emptyset} R_\mathcal V \right) + k_2(\epsilon)
\end{align}
where $k_2(\epsilon)\rightarrow0$ as $\epsilon\rightarrow0$, and
step $(a)$ uses $|S_{\mathcal{V}}|=2^{mR_\mathcal{V}}$ and \eqref{eq:Fanobound2}. We thus have
\begin{align}
    \frac1m I_{\omega, \lambda^{\cdot}}(\mathbf{M}_1;\mathbf{M}_2;\cdots;\mathbf{M}_k)
    & = \frac1m \sum\nolimits_\mathcal{U} \omega_\mathcal{U} \left( H(\mathbf{M}_{\mathcal{U}})
    - \sum\nolimits_{\mathcal{B}\subsetneq \mathcal{U}} \lambda^\mathcal{U}_{\mathcal B} H(\mathbf{M}_{\mathcal B}|\mathbf{M}_{\mathcal{U}-{\mathcal B}})
\right) \nonumber
\\
&\geq -k_3(\epsilon)+\sum_\mathcal{U}\omega_\mathcal{U}\left(
\sum_{\mathcal V:\mathcal{V}\cap\mathcal{U}\neq\emptyset}R_\mathcal V-\sum_{\mathcal B\subsetneq\mathcal{U}}\lambda^\mathcal{U}_{\mathcal B}\sum_{\mathcal V:~\mathcal V\cap \mathcal{B}\neq \emptyset, \mathcal V\cap (\mathcal U-\mathcal{B})=\emptyset}R_\mathcal V \right)
\label{eq:thm1-2}
\end{align}
where $k_3(\epsilon)\rightarrow0$ as $\epsilon\rightarrow0$.
We reformulate \eqref{eq:thm1-2} as
\begin{align}
    & \sum_{\mathcal V}R_\mathcal V\left(\sum_{\mathcal{U}: \mathcal{V}\cap\mathcal{U}\neq\emptyset}\omega_\mathcal{U}\left(1 -\sum_{\mathcal B\subsetneq\mathcal{U}:~ \mathcal{V}\cap(\mathcal U-\mathcal B)=\emptyset}\lambda^\mathcal{U}_{\mathcal B}\right)\right)
    \leq \frac{1}{m}I_{\omega, \lambda^{\cdot}}(\mathbf{M}_1;\mathbf{M}_2;\cdots;\mathbf{M}_k)+k_3(\epsilon).
\end{align}

Next, using the conditioning inequality for $I_\lambda$ of Proposition \ref{propos1} in Appendix~\ref{sec:appendixA}, we have
\begin{align}
    I_{\omega, \lambda^{\cdot}}(\mathbf{M}_1;\mathbf{M}_2;\cdots;\mathbf{M}_k)
    & \leq I_{\omega, \lambda^{\cdot}}(\mathbf{M}_1;\mathbf{M}_2;\cdots;\mathbf{M}_k|T^n) + I(\mathbf{M}_{[k]};T^n) \nonumber \\
    & \overset{(a)}{\leq} I_{\omega, \lambda^{\cdot}}(\mathbf{M}_1;\mathbf{M}_2;\cdots;\mathbf{M}_k|T^n) + I(\mathbf{M}_{[k]};T^n)
   \nonumber \\
   &\quad - I(\mathbf{M}_{[k]};Z^n)+mk_4(\epsilon)
\end{align}
for some $k_4(\epsilon)\rightarrow0$ as $\epsilon\rightarrow0$, where step $(a)$ uses \eqref{eq:security2}. Observe that
\begin{align}
   I(\mathbf{M}_{[k]};T^n)-I(\mathbf{M}_{[k]};Z^n)
   &=\sum_{j \in [n]} I(\mathbf{M}_{[k]};T_j|Z_{j+1}^n, T^{j-1})-I(\mathbf{M}_{[k]};Z_j|Z_{j+1}^n, T^{j-1}) \nonumber \\
   &=\sum_{j \in [n]} I(V_j;T_j|U_jA_j)-I(V_j;Z_j|U_jA_j)
\end{align}
where $V_j=\mathbf{M}_{[k]}$, $U_j=Z_{j+1}^n$ and $A_j=T^{j-1}$. Note that
\begin{align}
   A_j\mkv X_{[k]j}\mkv Y_{[k]j}T_jZ_j \\
   U_jV_jA_j\mkv X_{[k]j}Y_{[k]j}\mkv T_jZ_j
\end{align}
form Markov chains.  Next, we have
\begin{align}
   I_{\omega, \lambda^{\cdot}}(\mathbf{M}_1
   & ;\mathbf{M}_2;\cdots;\mathbf{M}_k|T^n)
   \overset{(a)}{\leq}
   I_{\omega, \lambda^{\cdot}}(W_{1}Y_{1}^n; W_{2}Y_{2}^n; \cdots; W_{k}Y_{k}^n | T^n)
   \nonumber \\
   & =I_{\omega, \lambda^{\cdot}}(W_{1}Y_{1}^n; W_{2}Y_{2}^n; \cdots; W_{k}Y_{k}^n | T^n)
   - I_{\omega, \lambda^{\cdot}}(W_{1}; W_{2}; \cdots; W_{k})
   \nonumber \\
   & \overset{(b)}{\leq} \sum_{j \in [n]} I_{\omega, \lambda^{\cdot}}(X_{1j}Y_{1j}; X_{2j}Y_{2j}; \cdots; X_{kj}Y_{kj} | T_j,T^{j-1})
   - \sum_{j \in [n]} I_{\omega, \lambda^{\cdot}}(X_{1j}; X_{2j}; \cdots; X_{kj} | T^{j-1})
   \nonumber \\
   & \qquad -\sum\nolimits_{j \in [n]} \sum\nolimits_{\mathcal{U}} \omega_\mathcal{U} \left(1-\sum\nolimits_{\mathcal B\subsetneq \mathcal{U}}\lambda^{\mathcal{U}}_{\mathcal B}\right) I(X_{[k]j}; Y_{\mathcal{U}j},T_j | X_{\mathcal{U}j}, T^{j-1})
   \label{eqnFL}
\end{align}
where \(W_i\) is the local (private) random variable available at the \(i\)-th party (see \eqref{eqgenS}), step $(a)$ follows from the data processing inequality for $I_\lambda$, see Proposition \ref{propos1} in Appendix~\ref{sec:appendixA}, step $(b)$ follows from the DB constraint of Lemma \ref{lemma1} in Section \ref{subsec:DB-bound}, and $k_3(\epsilon)\rightarrow0$ as $\epsilon\rightarrow0$.

Collecting the above results, we obtain
\begin{align}
   &\sum_{\mathcal V}R_\mathcal V\left(\sum_{\mathcal{U}: \mathcal{V}\cap\mathcal{U}\neq\emptyset}\omega_\mathcal{U}\left(1
   -\sum_{\mathcal B\subsetneq\mathcal{U}:~ \mathcal{V}\cap(\mathcal U-\mathcal B)=\emptyset}\lambda^\mathcal{U}_{\mathcal B}\right)\right)\nonumber
   \\ &\leq \frac1m
   \sum_{j \in [n]}\bigg[I_{\omega, \lambda^{\cdot}}(X_{1j}Y_{1j};X_{2j}Y_{2j};\cdots;X_{kj}Y_{kj}|T_j,A_j)-I_{\omega, \lambda^{\cdot}}(X_{1j};X_{2j};\cdots;X_{kj}|A_j)
   \nonumber \\
   & \qquad\qquad -\sum\nolimits_{\mathcal{U}} \omega_\mathcal{U} \bigg(1-\sum\nolimits_{\mathcal B\subsetneq \mathcal{U}}\lambda^{\mathcal{U}}_{\mathcal B}\bigg)I(X_{[k]j};Y_{\mathcal{U}j},T_j|X_{\mathcal{U}j}A_j)
   \nonumber\\
   & \qquad\qquad + I(V_j;T_j|U_jA_j) -I(V_j;Z_j|U_jA_j)\bigg] + k_3(\epsilon) + k_4(\epsilon).
\label{eqnAAB}
\end{align}
Consider the set of $m$ indices $j_1, j_2, \cdots, j_m\in[n]$ where the main channel is used. We have
\begin{align}
   & \sum_{b=1}^m \bigg[I_{\omega, \lambda^{\cdot}}(X_{1j_b}Y_{1j_b}; X_{2j_b}Y_{2j_b}; \cdots; X_{kj_b}Y_{kj_b} | T_{j_b},A_{j_b})
   - I_{\omega, \lambda^{\cdot}}(X_{1j_b}; X_{2j_b}; \cdots; X_{kj_b}|A_{j_b})
   \nonumber \\
   & \qquad\qquad -\sum\nolimits_{\mathcal{U}} \omega_\mathcal{U} \bigg(1-\sum\nolimits_{\mathcal B\subsetneq \mathcal{U}}\lambda^{\mathcal{U}}_{\mathcal B}\bigg) I(X_{[k]j_b};Y_{\mathcal{U}j_b},T_{j_b} | X_{\mathcal{U}j_b}A_{j_b})
   \nonumber \\
   & \qquad\qquad + I(V_{j_b};T_{j_b} | U_{j_b}A_{j_b}) - I(V_{j_b};Z_{j_b} | U_{j_b}A_{j_b})\bigg]
   \nonumber \\
   & \leq m\cdot V_{\omega, \lambda^{\cdot}} \big(p(y_{[k]},z|x_{[k]}) \cdot p(t|x_{[k]},y_{[k]},z) \big).
   \label{eqnDD}
\end{align}
A similar argument shows that the sum of the terms in \eqref{eqnAAB} where the parallel channel $q_\ell(y_{[k]},z|x_{[k]})$ is used, is bounded from above by 
\begin{align}
  m\cdot\alpha_\ell \cdot V_{\omega, \lambda^{\cdot}}\big(q_\ell(y_{[k]},z|x_{[k]})q_\ell(t|x_{[k]},y_{[k]},z)\big).
\end{align}

\subsection{Relation with Existing Results}\label{subsec:previous-results}
Introducing the auxiliary variable $T$ allows one to recover existing bounds for the two-terminal source model discussed below.

\subsubsection{Two-Terminal Source Model Problem}

Corollary \ref{corr9} recovers the current best upper bound for the source model \cite{gohari2010information1}. Suppose $k=2$ and $X_1$ and $X_2$ are constants. Choosing $\lambda_{\{1\}}=\lambda_{\{2\}}=1$, the $\lambda$-multivariate information reduces to the ordinary conditional mutual information. For any $p(t|y_{1},y_2,z)$, we obtain
\begin{align}
   V_{\omega, \lambda^{\cdot}}&(p(t,y_{1},y_2,z|x_1,x_2))
   = \max [
   I(X_{1}Y_{1};X_{2}Y_{2}|T)- I(X_{1};X_{2})
   +I(V;T|U)-I(V;Z|U)]
\end{align}
where the maximum is over all $p(x_{[k]})$ and auxiliary random variables $U, V$ for which the joint distribution of the random variables factors as
\begin{equation}
   p_{X_1,X_2} \cdot p_{Y_1,Y_2,T,Z|X_1,X_2} \cdot p_{U,V|X_1,X_2,Y_1,Y_2}.
\end{equation}
Since $X_1$ and $X_2$ are constants, we have
\begin{align}
   I(X_{1}Y_{1};X_{2}Y_{2}|T) -  I(X_{1};X_{2})=I(Y_{1};Y_{2}|T)
\end{align}
and
\begin{align}
   & V_{\omega, \lambda^{\cdot}}(p(t,y_{1},y_2,z|x_1,x_2)) 
   = I(Y_{1};Y_{2}|T) + \max_{(V,U)\mkv (Y_1,Y_2)\mkv (T,Z)}
   \left[I(V;T|U)-I(V;Z|U)\right].
\end{align}

Next, consider one parallel channel of the form $Y_1=Y_2=Z=(X_1,X_2)$ where $X_1$ and $X_2$ are binary, i.e., each use of the parallel channel is equivalent to broadcasting one bit. We now utilize the auxiliary receiver $T=Z$. Since $H(X_{[k]},Y_{[k]}|Z)=0$ in the parallel channel, we have 
\begin{equation}
    V_{\omega, \lambda^{\cdot}}\big( q_\ell(y_{[k]},z|x_{[k]})q_\ell(t|x_{[k]},y_{[k]},z)\big) \le 0
\end{equation}
and
\begin{align}
   R_{[k]} & \le V_{\omega, \lambda^{\cdot}} \big( p(y_{[k]},z|x_{[k]})p(t|x_{[k]},y_{[k]},z) \big)
   \nonumber \\
   & = I(Y_{1};Y_{2}|T)
   + \max_{(V,U)\mkv (Y_1,Y_2)\mkv (T,Z)} \left[I(V;T|U)-I(V;Z|U)\right].
\end{align}
Note that the channel-use rate $\alpha_1$ does not appear in the upper bound and can be set to infinity, allowing free public discussion. This recovers the current best upper bound for the source model for two users \cite{7976410}. A similar argument shows that Corollary \ref{corr9} recovers \eqref{eqnSM1}.

\subsubsection{Two-Terminal Channel Model Problem}

Suppose $X_2$ and $Y_1$ are constants in the main channel. This case is similar to the one discussed above. Take some arbitrary $p(t|x_{1},y_2,z)$ for which we obtain
\begin{align}
   & V_{\omega, \lambda^{\cdot}}(p(t,y_{1},y_2,z|x_1,x_2))=\max \left[
   I(X_{1};Y_{2}|T)+I(V;T|U)-I(V;Z|U)\right]
\end{align}
where the maximum is over $p(x_1)$ and all auxiliary random variables $U,V$ for which the joint distribution of the random variables factors as
\begin{equation}
    p_{X_1} \cdot p_{Y_2,T,Z|X_1} \cdot p_{U,V|X_1,Y_2}.
\end{equation}
As above, the corresponding term for the parallel (public) channel vanishes.
This recovers the current best upper bound for the channel model problem for two users \cite{gohari2010information2}. A similar argument shows that Corollary \ref{corr9} recovers  \eqref{eqnSM1Add}.

\subsubsection{Source Model Problem}
\label{subsubsec:kterminalssec}

Next, consider a $k$ terminal network $p(y_{[k]},z|x_{[k]})$ where $|\mathcal{X}_i|=1$ in the main network, i.e., the inputs are constant and the main network is described by $p(y_{[k]},z)$. Moreover, assume that $R_\mathcal{V}=0$ when $\mathcal{V}\neq [k]$. In other words, the terminals aim to create a shared secret key. Only the first $u$ terminals can participate in public discussion while terminals $u+1, u+2, \cdots, k$ remain silent. This public discussion can be modeled by the parallel channel $Y_1=Y_2=\cdots=Y_k=Z=X_{[u]}$ with $X_{u+1},\cdots, X_k$ being constants. 

Consider the assumption $H(Z|Y_i)=0$ for $i=1,2,\cdots, k$.
In this case, deriving the capacity requires using the general upper bound with suitable weights $\omega_\mathcal{U}$. This is done in Appendix \ref{appendixC}. Here, we consider $u=k$, so all terminals can speak, and model the public discussion by the parallel channel $Y_1=Y_2=\cdots=Y_k=Z=X_{[k]}$. Using the private key capacity result of \cite{1362897}, we obtain the maximum value for 
$R_\mathcal{V}$ as 
\begin{align}
\min_{\lambda}
I_\lambda(Y_1;Y_2;\cdots;Y_k|Z).
\end{align}

To recover this value from Corollary \ref{corr9}, choose the auxiliary receiver $T=Z$ for the main channel. Since $X_i$'s are constants, after some simplification, we obtain
\begin{align}
V_{\omega, \lambda^{\cdot}}(p(t,y_{[k]},z|x_{[k]}))
&=
I_{\lambda}(Y_{1};Y_{2};\cdots;Y_{k}|Z).
\end{align}
Next, consider the parallel channel $Y_1=Y_2=\cdots=Y_k=Z=X_{[k]}$ with density $q_1(y_{[k]},z|x_{[k]})$ and use the auxiliary receiver $T=Z$ for the parallel channel. Since
$I_{\lambda}(X_{1}Y_{1};X_{2}Y_{2};\cdots;X_{k}Y_{k}|Z,A)=0$ it is immediate that $V_{\omega, \lambda^{\cdot}}(q_1(t,y_{[k]},z|x_{[k]}))\leq 0$. As before, $\alpha_1$ does not appear in the upper bound and can be set to infinity (free public discussion). Since $\lambda$ was arbitrary, we obtain the upper bound
$\min_{\lambda} I_\lambda(Y_1;Y_2;\cdots;Y_k|Z)$.

\subsubsection{Wiretap Channel with Rate-Limited Secure Feedback}
\label{Wiretapratelimited}

We next discuss wiretap channels with rate-limited secure feedback. Consider $k=2$ and suppose $X_2$ and $Y_1$ are constants in the main channel, so we obtain a wiretap channel $p(y_2,z|x_1)$. For the parallel channel, consider a secure rate-limited feedback link as in \cite{Ardestanizadeh}. We model this by a parallel channel where $Y_2$ and $Z$ are constant while $Y_1 = X_{2}$ with the desired feedback rate $R_f$. We also set the parallel channel-use rate to $\alpha_1=1$. The main result of \cite{Ardestanizadeh} is the following upper bound on the rate of secure and reliable communication from the first terminal to the second terminal:
\begin{align}
    R\leq \max_{p(x_1)}
    \min\big(I(X_1; Y_2), R_f + I(X_1; Y_2|Z) \big).
\end{align}
The authors in \cite{Ardestanizadeh} do not consider the secret key rate that can be shared between the two terminals; instead, they focus on the rate of private communication from the first terminal to the second. Only the term $R_f + I(X_1; Y_2|Z)$ constitutes an upper bound on the
secret key rate that can be shared between the two terminals. To obtain the latter bound from our bound in Corollary \ref{corr9}, choose  $\lambda_{\{1\}}=\lambda_{\{2\}}=1$ and the auxiliary receiver $T=Z$. For the main channel, we can simplify $V_{\omega, \lambda^{\cdot}}(p(t,y_{1},y_2,z|x_1,x_2))$ because $Y_1$ and $X_2$ are constants:
\begin{align}
   V_{\omega, \lambda^{\cdot}}&(p(t,y_{1},y_2,z|x_1,x_2))
   = I(X_{1};Y_{2}|Z).
\end{align}
For the parallel channel, set $Y_1=X_2$, choose $X_1$ and $Z$ as constants, and use the auxiliary receiver $T=Z$ to obtain
\begin{align}
   V_{\omega, \lambda^{\cdot}} &(q_1(t,y_{1},y_2,z|x_1,x_2))
   = \max_{p(x_2)} I(Y_{1};X_{2})\le R_f.
\end{align}
These results yield the upper bound $R_f + I(X_1; Y_2|Z)$.

\subsection{New Bound for Randomness Generation}
\label{subsec:new-bound-randomness-generation}

Suppose $Z=\emptyset$ and $L=0$, so there are no parallel channels. This removes the secrecy aspect, and the problem reduces to generating common randomness among different subsets of terminals at given rates. We have the following result.

\begin{corollary}
The common randomness rates $R_\mathcal{V}$ for $\mathcal{V}\subseteq[k]$ are achievable only if for any $(\omega, \lambda^{\cdot})$ (see Definition \ref{defnneqw2}) we have
\begin{align}
    & \sum_{\mathcal V}R_\mathcal V\left(\sum_{\mathcal{U}: \mathcal{V}\cap\mathcal{U}\neq\emptyset}\omega_\mathcal{U} \left( 1-\sum_{\mathcal B\subsetneq\mathcal{U}:\, \mathcal{V}\cap(\mathcal U-\mathcal B)=\emptyset}\lambda^\mathcal{U}_{\mathcal B}\right)\right) \nonumber \\
    & \leq I_{\omega, \lambda^{\cdot}}(Y_1;Y_2;\cdots;Y_k|X_{[k]}) + \sum_\mathcal{U}\omega_\mathcal{U}\sum_{\mathcal B\subsetneq \mathcal{U}}\lambda_{\mathcal B}^{{\mathcal U}} I(X_{[k]};Y_{\mathcal{U}-\mathcal{B}}|X_{\mathcal{U}-\mathcal{B}})
\end{align}
for some $p(x_{[k]})$.
\end{corollary}
\begin{proof}
Consider \eqref{eqntheB} for $L=0$, $Z=\emptyset$, and $T=\emptyset$ for which we have
\begin{align}
    V_{\omega, \lambda^{\cdot}}(p(t,y_{[k]},z|x_{[k]}))
    = & \max\bigg[ I_{\omega, \lambda^{\cdot}} (X_{1}Y_{1};X_{2}Y_{2};\cdots;X_{k}Y_{k})- I_{\omega, \lambda^{\cdot}}(X_{1};X_{2};\cdots;X_{k}) \nonumber \\
    & \qquad\qquad - \sum\nolimits_{\mathcal{U}} \omega_\mathcal{U} \left(1-\sum\nolimits_{\mathcal B\subsetneq \mathcal{U}} \lambda^{\mathcal{U}}_{\mathcal B}\right) I(X_{[k]};Y_{\mathcal{U}}|X_{\mathcal{U}}) \bigg].
\end{align}
Now, observe the identity
\begin{align}
   & I_{\omega, \lambda^{\cdot}}
   (X_{1}Y_{1};X_{2}Y_{2};\cdots;X_{k}Y_{k})- I_{\omega, \lambda^{\cdot}}(X_{1};X_{2};\cdots;X_{k}) - \sum_{\mathcal{U}}\omega_\mathcal{U} \left(1-\sum_{\mathcal B\subsetneq \mathcal{U}}\lambda^{\mathcal{U}}_{\mathcal B}\right) I(X_{[k]};Y_{\mathcal{U}}|X_{\mathcal{U}}) \nonumber \\
   & = I_{\omega, \lambda^{\cdot}}(Y_1;Y_2;\cdots;Y_k|X_{[k]}) + \sum\nolimits_\mathcal{U}\omega_\mathcal{U}\sum\nolimits_{\mathcal B\subsetneq \mathcal{U}} \lambda_{\mathcal B}^{{\mathcal U}} I(X_{[k]};Y_{\mathcal{U}-\mathcal{B}}|X_{\mathcal{U} - \mathcal{B}}).
\end{align}
This completes the proof.
\end{proof}

Thus, setting $\omega_{[k]}=1$ and $\omega_\mathcal{U}=0$ when $\mathcal{U}\neq [k]$, common randomness generation at rate $R_\mathcal{V}$ for subset $\mathcal{V}$ is possible only if
\begin{align}
    & \sum_{\mathcal V} R_\mathcal V \left(1-\sum_{\mathcal B:\, \mathcal V \subseteq \mathcal{B}\subsetneq[k]} \lambda_{\mathcal B}\right) \leq I_{\lambda}(Y_1;Y_2;\cdots;Y_k|X_{[k]}) + \sum_{\mathcal B}\lambda_{\mathcal B} I(X_{\mathcal B};Y_{\mathcal{B}^c}|X_{{\mathcal B}^c})
\end{align}
for some $p(x_{[k]})$. 
For example, consider $k=2$ and a two-way channel $p(y_1,y_2|x_1,x_2)$. The rate of the shared randomness that can be produced between the two terminals is at most
\begin{align}
I(X_1;Y_2|X_2)+I(X_2;Y_1|X_1)+I(Y_1;Y_2|X_1,X_2)
\end{align}
for some $p(x_1,x_2)$. The terms $I(X_1;Y_2|X_2)$ and $I(X_2;Y_1|X_1)$ correspond to cut-set terms for generating common randomness by communicating bits from one terminal to the other, and $I(Y_1;Y_2|X_1,X_2)$ can be interpreted as an upper bound on the randomness generated through the channel noise. A similar interpretation holds for a general network $p(y_{[k]}|x_{[k]})$. The expression  $I(X_{\mathcal B};Y_{\mathcal{B}^c}|X_{{\mathcal B}^c})$ can be interpreted as a cut-set upper bound on the information flow, and $I_{\lambda}(Y_1;Y_2;\cdots;Y_k|X_{[k]})$ can be interpreted as an upper bound on the randomness generated through the channel noise.

\section{Multiterminal Communication}
\label{sec:mt-networks}

\subsection{System Model}
\label{subsec:mt-system-model}

Consider a memoryless network with the channel $p(y_{[k]}|x_{[k]})$ where $X_i$ and $Y_i$ are the respective channel inputs and outputs of the $i$-th transceiver, $i\in[k]$. Terminal $i$ wishes 
to reliably send a message $M_{i\mathcal S}$ with alphabet $[2^{nR_{i\mathcal S}}]$ of rate $R_{i\mathcal S}$ to terminals in $\mathcal S\subseteq [k]-\{i\}$ by using the channel $n$ times. The  messages $M_{i\mathcal S}$ are mutually independent and the channel input of user $i$ at time $j$ has the form $X_{ij}=f_{ij}(W_i,Y_{i[j-1]})$ where $W_i=(M_{i\mathcal{S}}, \mathcal S\subseteq [k]-\{i\})$; see \eqref{eqnconditionx}. Terminal $i$ outputs the estimates $\hat{M}^{(i)}_{j\mathcal{S}}=g_i(W_i,Y_{i[n]})$ for every $j\neq i$ and $\mathcal{S}$ that contains $i$. The uniformity and reliability requirements are
\begin{subequations}
\begin{align}
    \frac 1n H(M_{i\mathcal{S}}) & \ge R_{i\mathcal{S}}-\epsilon, \quad i\in[k],\; \mathcal S\subseteq [k]-\{i\} 
    \label{eqnCond2-1} \\
    \mathbb{P}\left[\medcap_{i\neq j, i\in\mathcal{S}} \{\hat{M}^{(i)}_{j\mathcal{S}} = M_{j\mathcal{S}}\}\right] & \ge 1-\epsilon .
    \label{eqnCond2-2}
\end{align}
\end{subequations}

We remark that relay networks are included in the setting described above. For example, even if the first terminal has no messages to transmit, i.e., \( R_{1\mathcal{S}} = 0 \) for all $\mathcal{S}$, it can act as a relay to assist communication. Various cooperative strategies can be employed, such as \emph{decode-and-forward}, \emph{compress-and-forward}, or \emph{amplify-and-forward}.

A general outer bound on the capacity region is the cut-set bound that we state explicitly. 

\begin{proposition}[Cut-set bound] \label{prop1}
    The achievable rate tuples $\{R_{i\mathcal{S}}\}$ satisfy
\begin{align}
    \sum\nolimits_{i\in \mathcal{S},\mathcal{L}\cap \mathcal{S}^c\neq \emptyset} R_{i\mathcal{L}} \leq I(X_{\mathcal S};Y_{\mathcal S^c}|X_{\mathcal S^c}),\quad \forall\, \mathcal S\subseteq[k], \label{eqn-cutset}
\end{align}
    for some joint distribution $p(x_{[k]})$.
\end{proposition}

The cut-set bound appeared in \cite{aref80,ElGamal-NTC81} (cf. \cite{kramer03} for general multicast) and coincides with the capacity region in some interesting cases: (i) point-to-point channels; (ii) two-user Gaussian MACs with output feedback \cite{ozarow-IT84}; (iii) symmetric $k$-user Gaussian MACs with output feedback and high signal-to-noise ratio \cite{Kramer-IT02}; (iv) relay channels with feedback from the receiver to the relay and the transmitter \cite{Cover79},\cite[Theorem 17.3]{elk11}, and (v) Gaussian relay channels with phase uncertainty when the relay is near the source \cite{Kramer05}.
However, the cut-set bound is loose even in basic cases such as MACs without feedback (where it can easily be modified to give the capacity region by adding a time-sharing variable) and three-terminal relay channels with one message \cite{el2022strengthened}.

We next develop a new and general capacity outer bound that improves the cut-set bound. We apply the bound to Gaussian MACs with generalized feedback, including Gaussian relay channels. One attractive feature that our bound shares with the cut-set bound is that Gaussian distributions are optimal.

\subsection{General Outer Bound}
\label{subsec:macgf-outer-bound}

In this section, we use auxiliary receivers similar to the \emph{parallel channel} extension of the DB constraint in \cite[Section V]{hekstra1989dependence}. We extend the idea to several auxiliary receivers with channel outputs $Z_m$, $m\in[a]$.

Lemma \ref{lemma1} yields the following outer bound on the capacity region.

\begin{theorem} \label{thm2}
    Consider an auxiliary channel $p(z_{[a]}|x_{[k]},y_{[k]})$. Any achievable rate tuples $\{R_{i\mathcal{S}}\}$ satisfy
\begin{align}
    \sum\nolimits_{i\in \mathcal{S},\mathcal{L}\cap \mathcal{S}^c\neq \emptyset} R_{i\mathcal{L}} \leq I(X_{\mathcal S};Z_m,Y_{\mathcal S^c}|X_{\mathcal S^c},T_m),\qquad \forall\, \mathcal S\subseteq[k],\; m \in [a]\label{eqn50-2}
\end{align}
    for some joint distribution that factorizes as
\begin{align} \label{eq:p-factorization-2}
    p(x_{[k]}) \cdot \left( \prod\nolimits_{m\in[a]} p(t_{m}|x_{[k]}) \right) \cdot p(y_{[k]}|x_{[k]}) \cdot p(z_{[a]}|x_{[k]},y_{[k]})
\end{align}
such that, for any $\mathcal{U}\subseteq[k]$ where $|\mathcal{U}|\geq 2$, any fractional partition $\lambda$ for indices in $\mathcal{U}$, and all $m\in[a]$, we have the DB constraints
\begin{align} 
    & I_\lambda(X_{i_1}Y_{i_1};X_{i_2}Y_{i_2};\cdots;X_{i_u}Y_{i_u}|Z_m,T_m)
    \nonumber \\
    & \ge I_\lambda(X_{i_1};X_{i_2};\cdots;X_{i_u}|T_m)
    + \left(1-\sum\nolimits_{\mathcal B\subsetneq \mathcal{U}}\lambda_{\mathcal B}\right) I(X_{[k]};Z_m,Y_{\mathcal{U}}|X_{\mathcal{U}},T_m).
    \label{eqnDBs1}
\end{align}
Moreover, one may assume
\begin{align} \label{eq:card-bound}
    |\mathcal{T}_m| \leq \prod\nolimits_{i\in[k]} |\mathcal{X}_i| + (2^k-1) + \binom{2^{k}-1+k}{2^{k}-1}, \quad \forall\, m\in[a] .
\end{align}
\end{theorem}

\begin{proof}
For $i\in [k]$, let $W_i=(M_{i\mathcal{S}}, \mathcal{S}\subsetneq[k]-\{i\}) $ be the collection of messages of user $i$ intended for other receivers. Consider any $\mathcal{U}\subseteq[k]$ and fractional partition $\lambda$ of the entries in $\mathcal{U}$.  Using Proposition~\ref{propos1} and Lemma \ref{lemma1}, we have
\begin{align}
    0 & \leq I_{\lambda}( W_{i_1}Y_{i_1}^n ; W_{i_2}Y_{i_2}^n ;\cdots; W_{i_u}Y_{i_u}^n \,|\, Z_m^n)
    - \underbrace{I_{\lambda}(W_{i_1};W_{i_2};\cdots;W_{i_u})}_{ \displaystyle = 0} \nonumber \\
    & \leq \sum_{j \in [n]}I_{\lambda}(X_{i_1j}Y_{{i_1}j}; X_{i_2j}Y_{{i_2}j} ;\cdots; X_{i_uj}Y_{{i_u}j} \,|\, Z_m^{j-1}, Z_{mj}) - \sum_{j \in [n]}I_{\lambda}(X_{i_1j} ; X_{i_2j} ;\cdots; X_{i_uj} \,|\, Z_m^{j-1}) \nonumber \\
    & \qquad -\sum_{j \in [n]} \left(1-\sum\nolimits_{\mathcal B\subsetneq \mathcal{U}} \lambda_{\mathcal B}\right) I(X_{[k]j} ; Z_{mj}, Y_{\mathcal {U}j} \,|\, Z_m^{j-1}, X_{\mathcal{U}j}).
\end{align}
Let $M_{\mathcal{S},\mathcal{S}^c}=(M_{i\mathcal{L}}: i\in\mathcal{S}, \mathcal{L}\cap \mathcal{S}^c\neq \emptyset)$. 
Then, for any $\mathcal{S}\subseteq [k]$, Fano's inequality gives 
\begin{align}
    n\sum\nolimits_{i\in \mathcal{S},\mathcal{L}\cap \mathcal{S}^c\neq \emptyset} R_{i\mathcal{L}}
    = H(M_{\mathcal{S},\mathcal{S}^c}|W_{S^c})
    \le I(M_{\mathcal{S},\mathcal{S}^c} ; Z_m^n,Y_{{\mathcal S^c}}^n|W_{S^c})+nk(\epsilon)
\end{align}
where $k(\epsilon)\rightarrow0$ as $\epsilon\rightarrow0$. We further have
\begin{align}
    I(M_{\mathcal{S},\mathcal{S}^c}; Z_m^n, Y_{{\mathcal S^c}}^n|W_{S^c})
    & = \sum\nolimits_{j \in [n]} I(M_{\mathcal{S},\mathcal{S}^c} ; Z_{mj}, Y_{{\mathcal S^c}j} | Z_m^{j-1}, Y_{{\mathcal S^c}}^{j-1}, W_{S^c},X_{\mathcal{S}^cj}) \nonumber \\
    & \leq \sum\nolimits_{j \in [n]} I(M_{\mathcal{S},\mathcal{S}^c},X_{\mathcal{S}j}; Z_{mj}, Y_{{\mathcal S^c}j} | Z_m^{j-1}, Y_{{\mathcal S^c}}^{j-1}, W_{S^c},X_{\mathcal{S}^cj}) \nonumber \\
    & \leq \sum\nolimits_{j \in [n]} I(X_{\mathcal{S}j}; Z_{mj}, Y_{{\mathcal S^c}j} | Z_m^{j-1},X_{S^c j}).
\end{align}
Defining $T_m=(Q,Z_m^{Q-1})$ for a time-sharing variable $Q$ gives the desired inequalities for some $p(x_{[k]},t_{[a]})$. Moreover, one may replace $p(x_{[k]},t_{m})$ with \eqref{eq:p-factorization-2} because all mutual information terms depend only on the marginals $p(x_{[k]},t_{m})$ for $m\in[a]$.

The cardinality bound \eqref{eq:card-bound} follows by standard arguments; we sketch the proof in Appendix~\ref{appendixD}.
\end{proof}

\begin{remark} One can interpret  $Z_m$ as being provided by a genie to all terminals, i.e., $Y_i$ is replaced with $Y_{i}'=(Y_i,Z_m)$ for all $i\in[k]$. The bounds in \eqref{eqn50-2} and \eqref{eqnDBs1} apply to this enhanced channel.
\end{remark}

\begin{remark} \label{remark-cut-set}
One recovers the cut-set bound with $Z_m$ a constant. To see this, note that the constraints \eqref{eqnDBs1} are redundant by the chain rule in Appendix~\ref{sec:appendixA} and the non-negativity of $\lambda$-multivariate and mutual information. We further have $I(X_{\mathcal S};Y_{{\mathcal S^c}}|X_{\mathcal S^c},T_m)\leq I(X_{\mathcal S};Y_{{\mathcal S^c}}|X_{\mathcal S^c})$ so it is optimal to choose $T_m$ independent of $X_{[k]}$. Of course, the interpretation that a constant $Z_m$ represents an ``auxiliary receiver'' is a formal one.
\end{remark}

\begin{remark} \label{remark-active-terminals}
Let $\mathcal{U}$ be the set of potentially active terminals, i.e., $|\mathcal{X}_i|>1$ for $i\in\mathcal{U}$ and $H(X_i)=0$ otherwise.
Using the chain rule in Appendix~\ref{sec:appendixA}, the DB constraints \eqref{eqnDBs1} are
\begin{align} 
    & I_\lambda(X_{i_1};X_{i_2};\cdots;X_{i_u} | T_m)
    \le I_\lambda(X_{i_1};X_{i_2};\cdots;X_{i_u} | Z_m,T_m) \nonumber \\
    & \qquad + I_\lambda(Y_{i_1};Y_{i_2};\cdots;Y_{i_u} | X_{\mathcal{U}},Z_m,T_m)
    + \sum\nolimits_{\mathcal{B} \subsetneq \mathcal{U}} \lambda_{\mathcal B} I(X_{\mathcal{B}};Y_{\mathcal{B}^c} | X_{\mathcal{B}^c}, Z_m, T_m) 
    \label{eqnDBs1-2a}
\end{align}
where $\mathcal{B}^c$ is here the complement of $\mathcal{B}$ in $\mathcal{U}$. The sum over $\mathcal B$ in \eqref{eqnDBs1-2a} vanishes by choosing $Z_m=X_{\mathcal U}$ or $Z_m=Y_{\mathcal U}$, for example. Also, for additive-noise channels with $Y_i=g_i(X_{\mathcal{U}})+N_i$ for some functions $g_i(\cdot)$ and all $i\in[k]$, and where the $N_1,N_2,\cdots,N_k$ are mutually independent of each other and $X_{\mathcal{U}}$, we have
\begin{align}
    I_\lambda(Y_{i_1};Y_{i_2};\cdots;Y_{i_u} | X_{\mathcal{U}},Z_m,T_m) = I_\lambda(N_{i_1};N_{i_2};\cdots;N_{i_u} | X_{\mathcal{U}},Z_m,T_m)
\end{align}
which is zero if one chooses $Z_m$ that are combinations of the $X_i$ and $Y_i$.
\end{remark}

\begin{remark}
    Suppose terminal $i$ is a relay, i.e., $R_{i\mathcal{S}} = 0$ for all $\mathcal{S} \subseteq [k]\setminus\{i\}$ and $W_i$ is a constant. Assume that $H(Y_i|Z_m) = 0$. Then $H(X_{ij}|Z_m^{j-1}) = 0$ for all $j$.
    Consequently, we have $H(X_i|T_m) = 0$ and can write $T_m = (X_i, T'_m)$ for some auxiliary random variable $T'_m$.
\end{remark}

\begin{remark}
    An extension of Theorem~\ref{thm2} considers adaptive parallel channels in which the $Z_{[a]}$ depend on the conditional distribution $p_{X_{[k]}|T_m = t_m}$; see \cite[Section~VI]{hekstra1989dependence}. Specifically, for each realization $T_m = t_m$, define the auxiliary receivers through a conditional distribution $P_{Z_{[a]}|X_{[k]}}$ that depends on $p_{X_{[k]}|T_m}(\cdot \mid t_m)$. We do not explore this idea here, but emphasize that it appears promising.
\end{remark}

\subsubsection{Refinement}
\label{subsubsec:refinement}

The DB constraint \eqref{eqnDBs1} seems most useful with $\mathcal {U}=[k]$, which means the final mutual information term vanishes. However, this approach treats all messages equally. For example, for $k=3$ the constraints \eqref{eqnDBs1} are
\begin{align} 
    I_\lambda(X_1;X_2;X_3|T_m) \le 
    I_\lambda(X_1Y_1 ; X_2Y_2; X_3Y_3|Z_m,T_m).
\end{align}
Instead, one might wish to focus on a subset $\mathcal{V}\subsetneq[k]$ of terminals whose messages are destined for receivers in $\mathcal{V}^c$. To accomplish this, we provide $W_{\mathcal{V}^c}$ to all terminals. Consider $Z_{mj}=Z_{mj}' W_{\mathcal{V}^c} Y_{\mathcal{V}^cj}$, where $Z_{mj}'$ plays the role of $Z_{mj}$ previously. This $Z_{mj}$ satisfies the DB Markov chain \eqref{eq:DB-Markov-chain2}. One might also wish to consider $Z_j=Z_j' W_{\mathcal{V}^c}$.

Now consider $\mathcal{U}=\mathcal{V}$; similar steps are possible for $\mathcal{U}\ne\mathcal{V}$. 
We identify $T_{m}=(Q,Z_m^{Q-1})$ and follow the steps of the proof of Theorem~\ref{thm2} to obtain
\begin{align} 
    & I_\lambda(X_{i_1}Y_{i_1};X_{i_2}Y_{i_2};\cdots;X_{i_u}Y_{i_u}|Z_m,T_m)
    \nonumber \\
    & \overset{(a)}{=} I_\lambda(X_{i_1}Y_{i_1};X_{i_2}Y_{i_2};\cdots;X_{i_u}Y_{i_u}|Z_m',X_{\mathcal{U}^c},Y_{\mathcal{V}^c},T_m)  \nonumber \\
    & \overset{(b)}{\ge} I_\lambda(X_{i_1};X_{i_2};\cdots;X_{i_u}|X_{\mathcal{U}^c},T_m)
    + \left(1-\sum\nolimits_{\mathcal B\subsetneq \mathcal{U}}\lambda_{\mathcal B}\right) \underbrace{I(X_{[k]};Z_m',Y_{\mathcal{U}}|X_{\mathcal{U}},X_{\mathcal{U}^c},T_m)}_{\displaystyle =0}
    \label{eqnDBs1new}
\end{align}
where steps $(a)$ and $(b)$ follow because $W_{\mathcal{V}^c}$ is part of $T_m$. We also obtain the rate bounds
\begin{align}
    \sum\nolimits_{i\in \mathcal{S},\mathcal{L}\cap \mathcal{S}^c\neq \emptyset} R_{i\mathcal{L}}
    \leq I(X_{\mathcal S};Z_m',Y_{\mathcal S^c}|X_{\mathcal S^c},T_m),
    \quad \forall\, \mathcal S\subseteq \mathcal{U} . \label{eqn50-2new}
\end{align}

For example, consider $k=3$ and $\mathcal{U}=\{1,2\}$ so $\mathcal{U}^c=\{3\}$. We then have
\begin{subequations}
\begin{align}
    R_{12}+R_{13}+R_{1\{2,3\}} & \leq I(X_1;Z_m',Y_2,Y_3|X_2,X_3,T_m) \\
    R_{21}+R_{23}+R_{2\{1,3\}} & \leq I(X_2;Z_m',Y_1,Y_3|X_1,X_3,T_m) \\
    R_{13}+R_{1\{2,3\}} + R_{23}+R_{2\{1,3\}} & \le I(X_1,X_2;Z_m',Y_3|X_3,T_m) \\
    I(X_1;X_2|X_3,T_m) 
    & \le I(X_1Y_1;X_2Y_2|Z_m',X_3,Y_3,T_m) .
\end{align}
\end{subequations}
Note that choosing $Z_m'$ as a constant gives the same bounds as $Z_m'=Y_3$.

\subsubsection{Gaussian Networks}
\label{subsubsec:Gaussian-Networks}

Consider real-valued $k$-user channels and auxiliary receivers $Z_{[a]}$ of the form 
\begin{align}
    (Y_{[k]},Z_{[a]}) = X_{[k]} A + N_{[k+a]}
    \label{eq:Gaussian-channel}
\end{align}
for some  $k\times (k+a)$ matrix $A$ and a Gaussian noise vector $N_{[k+a]}$ that is independent of $X_{[k]}$. Consider the average block power constraints
\begin{align}
    \frac1n \sum_{j\in[n]} \mathbb{E}[X_{ij}^2] \leq P_i, \quad \forall\, i\in[k].
    \label{eq:power-constraints}
\end{align}
The outer bound in Theorem \ref{thm2} is valid for {\color{black} continuous channels and the power constraints \eqref{eq:power-constraints}, see Sec. \ref{subsubsec:continuous-rvs}.}

\begin{theorem}\label{proof-of-Gaussian-Optimality}
To evaluate the outer bound in Theorem \ref{thm2} for Gaussian channels and auxiliary receivers, it suffices to consider jointly Gaussian $X_{[k]}, T_{[a]}$ satisfying $\mathrm{E}[X_i^2]\leq P_i$, $i\in[k]$. Moreover, $T_m$ has dimension at most $k$ for all $m\in[a]$.
\end{theorem}
\begin{proof}
    See Appendix \ref{appendixE}.
\end{proof}

\begin{remark}
For $a=1$, one can assume that $T_1$ is a constant random variable. The complexity of evaluating the outer bound is then equivalent to that of evaluating the cut-set bound. To see this, consider jointly Gaussian $X_{[k]}$ and $T_1$, define $T'_1$ as a constant, and let
\[
    K_{X'_{[k]}} = K_{X_{[k]}|T_1}.
\]
Now replace $(X_{[k]}, T_1)$ with $(X'_{[k]}, T'_1)$. The new random variables satisfy the power constraints and yield the same outer bound as $(X_{[k]}, T_1)$.
\end{remark}

\begin{remark}
For $a>1$, evaluating the outer bound is more difficult because the unconditional covariance matrix $K_{X_{[k]}}$ links the $T_m$. For example, the conditional covariance matrices must satisfy
\begin{align} \label{eq:remark-a1}
    K_{X_{[k]}|T_m} \preceq K_{X_{[k]}}, \quad \forall m. 
\end{align}
To illustrate the restrictions, consider $k = 2$ and $P_1=P_2=1$, and suppose we would like to use  
\[
    K_{X_{[k]}|T_1} = \begin{pmatrix} 1 & 1 \\ 1 & 1 \end{pmatrix}, \quad  
    K_{X_{[k]}|T_2} = \begin{pmatrix} \phantom{-}1 & -1 \\ -1 & \phantom{-}1 \end{pmatrix}.  
\]  
However, this choice is invalid because there is no $K_{X_{[k]}}$ satisfying \eqref{eq:remark-a1} and the power constraints.
\end{remark}

\begin{remark}
A natural choice for $Z_m$ is to select subsets of channel inputs and/or outputs, possibly their noisy versions. For example, we may define:
$Z_1 = Y_{\mathcal{S}_1}$ and 
$Z_2 = Y_{\mathcal{S}_2}$
for some subsets $\mathcal{S}_1, \mathcal{S}_2 \subseteq [k]$. 
When $\mathcal{S}_2 \subset \mathcal{S}_1$, this induces the Markov chain $X_{[k]}Y_{[k]} \mkv Z_1 \mkv Z_2$.\footnote{Another example is when $Z_1 = Y_{\mathcal{S}}$ and 
$Z_2 = \tilde Y_{\mathcal{S}}$ where $\tilde{Y}_i$ is $Y_i$ plus noise.} Moreover, for all $j \in [n]$, we have the Markov chains
\[
X_{[k]j}\mkv Z_1^{j-1} \mkv Z_2^{j-1}.
\]
Consequently, we obtain $K_{X_{[k]}|T_1} \preceq K_{X_{[k]}|T_2}$, since $T_1$ and $T_2$ represent the past of $Z_1$ and $Z_2$ respectively.

To formalize this claim, we can adapt the proof of Theorem~\ref{proof-of-Gaussian-Optimality} to account for the Markov condition while maintaining the joint Gaussianity of the random variables. We omit the detailed proof; the key modifications are as follows.
\begin{enumerate} \itemsep 0pt
    \item Replace each $T_1$ with $(T_1, T_2)$ in the outer bound of Theorem~\ref{thm2} and show that the bound remains valid under the Markov chain $X_{[k]}Y_{[k]} \mkv Z_1 \mkv Z_2$.
    \item Modify the factorization in \eqref{eq:p-factorization-2} to
    \begin{align} 
        p(x_{[k]}) \cdot p(t_{1},t_2|x_{[k]}) \cdot \left( \prod\nolimits_{m\geq 3} p(t_{m}|x_{[k]}) \right) \cdot p(y_{[k]}|x_{[k]}) \cdot p(z_{[a]}|x_{[k]},y_{[k]}).
    \end{align}
\end{enumerate}
The arguments in Appendix~\ref{appendixE} can be extended to this modified outer bound structure.
\end{remark}

\subsection{MAC with Generalized Feedback}
\label{subsec:MAC-GF}

A $k$-user MAC with generalized feedback is a memoryless network with $k+1$ terminals and the channel
\begin{align}
p(y,y_{1},y_{2},\cdots, y_{k}|x_1,x_2, \cdots, x_k)
\end{align}
where we write $Y:=Y_{k+1}$. Terminal $i$, $i \in [k]$, sends a message with rate $R_i$ to the destination.

The MAC with $k=2$ users has been the subject of many studies; see \cite{Gaarder-IT75,king78,
Cover-Leung-IT81,
carleial1982multiple,willems82,hekstra1989dependence,
Kramer-IT02,kramer03,
Sendonaris-COMM03,Sendonaris-COMM03b,
Laneman-ISIT04,
kramer2006dependence,gastpar06,gastpar06b,
wigger08,
tandon2011dependence,Sula-IT20,kramer-ITW21, kosut2023perfect,kosut2025switched}. However, even characterizing the rate pairs $(R_1, R_2)$ with $R_2=0$ remains an open problem. This case is the relay channel where the second user has no message but supports communication, e.g., by enabling range extension or higher rates. The MAC with $k>2$ users has been studied in
\cite{Kramer-IT02,kramer03,kramer2006dependence,Sula-IT20,kramer-ITW21}. 

Theorem \ref{thm2} with $\mathcal{U}=[k]$ yields the following result; see Remark~\ref{remark-active-terminals}.

\begin{corollary} \label{cor-3}
    Consider an auxiliary channel $p(z_{[a]}|x_{[k]},y,y_{[k]})$. Any achievable rate tuple $(R_1,\cdots, R_k)$ for a $k$-user MAC with generalized feedback satisfies
\begin{align}
    \sum\nolimits_{i\in \mathcal{S}} R_i \leq I(X_{\mathcal S}; Z_m,Y,Y_{{\mathcal S^c}} | X_{\mathcal S^c},T_m),\qquad \forall\, \mathcal S\subseteq[k],\; m \in [a]\label{eqn50}
\end{align}
    for some joint distribution that factorizes as
\begin{align} \label{eq:p-factorization}
    p(x_{[k]}) \cdot \left( \prod\nolimits_{m\in[a]} p(t_{m}|x_{[k]}) \right) \cdot p(y,y_{{[k]}}|x_{[k]}) \cdot p(z_{[a]}|x_{[k]},y,y_{{[k]}})
\end{align}
such that, for any $\mathcal{V}=\{i_1,i_2,\cdots,i_v\}\subseteq[k]$ where $|\mathcal{V}|=v\geq 2$, any fractional partition $\lambda$ of $\mathcal{V}$, and all $m\in[a]$, we have the DB constraints
\begin{align} 
    & I_\lambda(X_{i_1}Y_{{i_1}};X_{i_2}Y_{{i_2}};\cdots;X_{i_v}Y_{{i_v}}|Z_m,T_m)
    \nonumber \\
    & \ge I_\lambda(X_{i_1};X_{i_2};\cdots;X_{i_v}|T_m)
    + \left(1-\sum\nolimits_{\mathcal B\subsetneq \mathcal{V}}\lambda_{\mathcal B}\right) I(X_{[k]}; Z_m, Y_{{\mathcal{V}}} | X_{\mathcal{V}},T_m).
    \label{eqnDBs1-2}
\end{align}
Moreover, one may assume the cardinality bounds \eqref{eq:card-bound}.
\end{corollary}

\subsubsection{One Auxiliary Receiver}
\label{subsec:one-auxiliary-receiver}

Corollary~\ref{cor-3} improves the cut-set bound for $k$-user MACs with generalized feedback. For example, one can generalize the bounds in \cite{gastpar06,gastpar06b} by using $\mathcal{V}=[k]$ and $a=1$ with $Z_1=Y$.

\begin{corollary}
\label{cor-4}
Consider a $k$-user MAC with generalized feedback. Any achievable $(R_1, \cdots, R_k)$ satisfies
\begin{align}
    \sum\nolimits_{i\in \mathcal{S}} R_i \leq I(X_{\mathcal S};Y,Y_{{\mathcal S^c}}|X_{\mathcal S^c},T),\qquad \forall\, \mathcal S\subseteq[k]
\end{align}
for some $p(t,x_{[k]})\cdot p(y,y_{1},y_{2},\cdots, y_{k}|x_{[k]})$ such that for any $\mathcal{V}\subseteq[k]$ where $|\mathcal{V}|\geq 2$ and any fractional partition $\lambda$ for indices in $\mathcal{V}$ we have the DB constraint
\begin{align} \label{eqnDBs2}
    & I_\lambda(X_{i_1}Y_{{i_1}};X_{i_2}Y_{{i_2}};\cdots;X_{i_v}Y_{{i_v}}|Y,T)
    \nonumber \\
    & \ge I_\lambda(X_{i_1};X_{i_2};\cdots;X_{i_v}|T)
    + \left(1-\sum\nolimits_{\mathcal B\subsetneq \mathcal{V}}\lambda_{\mathcal B}\right) I(X_{[k]};Y,Y_{{\mathcal{V}}}|X_{\mathcal{V}},T)
\end{align}
and the cardinality of $T$ can be limited as in \eqref{eq:card-bound}.
\end{corollary}

\begin{remark}
One recovers the cut-set bound by discarding the dependence balance constraint \eqref{eqnDBs2}; the best $T$ is then a constant.
\end{remark}

The following example illustrates the benefit of using Corollary~\ref{cor-4} with $\mathcal{V}\neq [k]$ in \eqref{eqnDBs2}. Suppose $X_k$ does not significantly affect the channel outputs; assume $X_k$ is constant for simplicity. However, suppose the feedback is the informative
\begin{align} \label{eq:informativeY}
    Y_{k} = X_{[k-1]} Y_{[k-1]}.
\end{align}
The choice $\mathcal{V}=[k]$ can yield weak bounds, since $X_kY_{k}$ is informative even when $X_k$ is a constant. On the other hand, the choice $\mathcal{V}=[k-1]$ makes the term
$I(X_{[k]} ; Y Y_{[k-1]} |X_{\mathcal{V}},T)$
vanish since $X_k$ is a constant. Moreover,
$I_\lambda(X_{1}Y_{{1}};X_{2}Y_{{2}};\cdots;X_{k-1}Y_{{k-1}}|T,Y)$ does not include $Y_{k}$. 

As another example, let  $Y_{i}=Y$ for all $i$, i.e., the terminals have a common output. The DB constraint \eqref{eqnDBs2} can be written as 
\begin{align}
    & \left( 1-\sum\nolimits_{\mathcal B\subsetneq \mathcal{V}} \lambda_{\mathcal B} \right) H(X_\mathcal{V}|Y,T) + \sum\nolimits_{\mathcal B\subsetneq \mathcal{V}} \lambda_{\mathcal B} H(X_{\mathcal{V}-\mathcal B}|Y,T) \nonumber \\
    & \geq \left( 1-\sum\nolimits_{\mathcal B\subsetneq \mathcal{V}} \lambda_{\mathcal B} \right) H(X_\mathcal{V}|T) + \sum\nolimits_{\mathcal B\subsetneq \mathcal{V}} \lambda_{\mathcal B} H(X_{\mathcal{V}-\mathcal B}|T) +
    \left( 1-\sum\nolimits_{\mathcal B\subsetneq \mathcal{V}} \lambda_{\mathcal B} \right) I(X_{[k]};Y|X_{\mathcal{V}},T)
\end{align}
which simplifies as
\begin{align}
    I(X_{[k]};Y|T) \leq \sum\nolimits_{\mathcal B\subsetneq \mathcal{V}} \lambda_{\mathcal B}\, I(X_{[k]} ; Y | X_{\mathcal{V}-\mathcal{B}}, T) \label{eqnlamb3}
\end{align}
where the sum is over a fractional partition of $\mathcal V$. We argue that $\mathcal{V}=[k]$ gives the strongest bound because one can convert the fractional partition of $\mathcal{V} \subseteq [k]$ to a fractional partition of $[k]$. Let $i_1\in\mathcal{V}$. For any $\mathcal{B}\subsetneq\mathcal{V}$, let
$
\lambda'_{\mathcal{B}}=\lambda_{\mathcal{B}}
$
if $i_1\notin \mathcal{B}$. For $i_1\in \mathcal{B}$, let
$
\lambda'_{\mathcal{B}~\cup ([k]-\mathcal{V})} = \lambda_{\mathcal{B}}
$. 
Finally, assign  $\lambda'_{\mathcal{B}'}=0$ for all the other sets $\mathcal{B}'\subsetneq[k]$ that are not of these two forms. Observe that $\mathcal{V}=[k]$ recovers the refined DB equations of \cite{kramer2006dependence} if one optimizes over $\lambda$; see Appendix~\ref{subsec:appendixA1}.

\begin{remark}
Choosing $\lambda_{[k]-\{i\}}=1/(k-1)$ for $i\in[k]$ in \eqref{eqnlamb3} gives (cf. Appendix~\ref{subsec:appendixA1} and \eqref{eq:frac-example})
\begin{align} \label{eq:MACS-GF-DB}
 I(X_{[k]};Y|T) \leq \frac{1}{k-1} \sum\nolimits_{|\mathcal B|=k-1} 
 I(X_{\mathcal{B}} ; Y | X_{\mathcal{B}^c}, T).
\end{align}
This bound gives the sum-rate capacity for $k$-user Gaussian MACs with symmetric channel coefficients and power constraints; see \cite{Kramer-IT02,Sula-IT20} and Section \ref{subsubsec:Gaussian} below. It is interesting to consider whether other partitions $\lambda$ give capacity points, including for asymmetric channel coefficients and power constraints.
\end{remark}

\subsubsection{Two Auxiliary Receivers}
\label{subsec:two-auxiliary-receivers}

We next consider $a=2$ auxiliary receivers.
One can generalize the bounds in  \cite{gastpar06,gastpar06b,tandon2011dependence} by using $\mathcal{V}=[k]$, $Z_1=Y_{{[k]}}$, and $Z_2$ a constant to include the cut-set bounds (cf. Remark \ref{remark-cut-set}).

\begin{corollary}[Extension of {\cite[Theorem~3]{gastpar06} and \cite[Theorem~1]{tandon2011dependence}} to $k\ge2$]
\label{cor-5}
Consider a $k$-user MAC with generalized feedback. Any achievable $(R_1, \cdots, R_k)$ satisfies
\begin{align}
    \sum\nolimits_{i\in \mathcal{S}}
    R_i & \leq \min\left(\, I(X_{\mathcal S};Y,Y_{{[k]}}|X_{\mathcal S^c},T), \,  I(X_{\mathcal S};Y,Y_{{\mathcal S^c}}|X_{\mathcal S^c}) \,\right),
    \quad \forall\, \mathcal S\subseteq[k]
\end{align}
for some $p(t,x_{[k]})\cdot p(y,y_{{[k]}}|x_{[k]})$ satisfying
\begin{align}
    I_\lambda(X_{1};X_{2};\cdots;X_k|T)
    & \leq I_\lambda(X_{1};X_{2};\cdots;X_k|{Y}_{{[k]}},T)
\label{DBMAC6ann}
\end{align}
and the cardinality of $T$ can be limited as in \eqref{eq:card-bound}.
\end{corollary}

\begin{remark}
For $k=2$, the DB constraint \eqref{DBMAC6ann} appeared in \cite{gastpar06}. This paper also states that Gaussian variables are optimal for Gaussian channels by using the variance-based DB constraint of \cite[Theorem 2]{kramer2006dependence}. However, the proof in \cite{kramer2006dependence} is incorrect because \cite[Eq. (42)]{kramer2006dependence} is valid only if certain Markov chains transfer from general to Gaussian distributions. This is not always the case, as pointed out in \cite[Chapter 3]{wigger08}. The paper \cite{gastpar06b} instead uses Lagrange optimization and the entropy power inequality.
\end{remark}

\subsubsection{Two Users}

We specialize to $k=2$ users. We begin by stating Willems' achievable region and an outer bound of Tandon-Ulukus that uses $Z_1=Y_{[2]}$ and the sum-rate cut bound.

\begin{proposition}[Willems \cite{willems82}]
\label{prop2}
    An achievable region for the two-user MAC with generalized feedback is the set of rate pairs $(R_1, R_2)$ satisfying
\begin{subequations}
\begin{align}
    R_1 & \leq  I(X_1;Y|X_2,U_1,T) + I(U_1 ; Y_2 | X_2,T) \\
    R_2 & \leq  I(X_2;Y|X_1,U_2,T) + I(U_2 ; Y_1 | X_1,T) \\
    R_1+R_2&\leq I(X_1,X_2;Y|U_1,U_2,T) + I(U_1 ; Y_2 | X_2,T) + I(U_2 ; Y_1 | X_1,T) \\
    R_1+R_2&\leq I(X_1,X_2;Y)
\end{align}
\end{subequations}
    where $U_1X_1\mkv T\mkv U_2X_2$ forms a Markov chain.
\end{proposition}

\begin{proposition}[Tandon-Ulukus {\cite[Theorem 1]{tandon2011dependence}}]
\label{prop3}
    The capacity region of the two-user MAC with generalized feedback is a subset of the rate pairs $(R_1, R_2)$ satisfying
    \begin{subequations}
    \begin{align}
    R_1 & \le  I(X_1;Y,Y_1,Y_2|X_2,T) \\
    R_2 & \le  I(X_2;Y,Y_1,Y_2|X_1,T) \\
    R_1+R_2&\le \min\left( I(X_1,X_2;Y,Y_1,Y_2|T),\,
    I(X_1,X_2;Y) \right)\\
    I(X_1;X_2|T)&\le I(X_1;X_2|Y_1,Y_2,T)\label{tudbconst}
    \end{align}
    \end{subequations}
    where $|\mathcal T|\le |\mathcal X_1|\,|\mathcal X_2|+3$.
\end{proposition}

\begin{remark}\label{rmk16}
The Tandon-Ulukus bound is weaker than the cut-set bound in general. For example, if $R_2=0$ we have a relay channel with feedback to the transmitter, and the outer bound of Proposition \ref{prop2} is
\begin{align}
    R_1 \le \max_{P_{X_1,X_2}} \min \left( I(X_1;Y,Y_1,Y_2|X_2), I(X_1,X_2;Y) \right)
    \label{eq:TU-relay-bound}
\end{align}
where it is best to choose $T=X_2$ to satisfy the DB constraint. The cut-set bound improves \eqref{eq:TU-relay-bound} in general because it does not include $Y_1$.
\end{remark}

We next consider the special case of $a=2$ auxiliary receivers with $Z_1=Y_{[2]}$ and $Z_2=Y$ which improves Proposition~\ref{prop3}.

\begin{corollary}\label{cor-6}
Consider a two-user MAC with generalized feedback. Any achievable $(R_1, R_2)$ satisfies
\begin{subequations}
\begin{align}
    R_{1} & \leq \min\left(\, I(X_{1};Y,Y_{{1}},Y_{{2}}|X_{2},T_1),\,
    I(X_1;Y,Y_{2}|X_2,T_2) \,\right)
    \label{DBMAC-21} \\
    R_{2} & \leq \min\left(\, I(X_{2};Y,Y_{{1}},Y_{{2}}|X_{1},T_1),\, I(X_2;Y,Y_{1}|X_1,T_2) \,\right)
    \label{DBMAC-22} \\
    R_{1}+R_{2} & \leq \min\left(\, I(X_{1},X_{2};Y,Y_{{1}},Y_{{2}}|T_1), \, I(X_{1},X_{2};Y|T_2) \,\right)
    \label{DBMAC-23}
\end{align}
\end{subequations}
for some $p(x_1,x_2)\cdot p(t_1|x_1,x_2)\, p(t_2|x_1,x_2) \cdot p(y,y_{1},y_{2}|x_1,x_2)$ satisfying
\begin{subequations}
\begin{align}
    I(X_{1};X_{2}|T_1) & \leq I(X_{1};X_{2}|Y_{{1}},Y_{{2}},T_1) 
    \label{DBMAC-2a} \\
    I(X_1;X_2|T_2) & \leq I(X_1 Y_{1}; X_2 Y_{2} | Y,T_2) .
    \label{DBMAC-2b}
\end{align}
\end{subequations}
Moreover, one can bound $|\mathcal{T}_1|\leq |\mathcal{X}_1||\mathcal{X}_2|+3$ and $|\mathcal{T}_2|\leq |\mathcal{X}_1||\mathcal{X}_2|+3$.
\end{corollary}

\begin{remark}
By discarding the DB constraint \eqref{DBMAC-2b}, the best $T_2$ is constant. Thus, we recover the cut-set bound if $R_2=0$, which improves Proposition \ref{prop3} in general. 
\end{remark}

\subsubsection{Gaussian Channels}
\label{subsubsec:Gaussian}

Consider a Gaussian MAC with outputs
\begin{subequations}
\begin{align}
        Y & = g_1 X_1 + g_2 X_2 + N \\
    Y_{1} & = g_{21} X_2 + N_1 \\
    Y_{2} & = g_{12} X_1 + N_2
\end{align}
\end{subequations}
where the $g_i,g_{ij}$ are channel coefficients and $N,N_1,N_2$ are Gaussian noise variables, i.e., $(N,N_1,N_2)$ is independent of $(X_1,X_2)$ but the $N,N_1,N_2$ may be correlated.

\begin{remark}\label{rmk20}
For the Gaussian MAC, the bound \eqref{DBMAC-2b} can be written as
\begin{align}
    I(X_{1},X_{2};Y|T_2)
    \le I(X_1;Y,Y_{{2}}|X_2,T_2) + I(X_2;Y,Y_{{1}}|X_1,T_2)
    + I(N_1;N_2|N).
    \label{eq:mac-gf-uc}
\end{align}
To see this, observe that the chain rule gives
\begin{align}
    I(X_{1};X_{2}|T_2) \le I(X_{1}Y_{1};X_{2}Y_{2}|Y,T_2)
    & = I(X_{1};X_{2}|Y,T_2) + I(Y_{1};Y_{2} \,|\, X_{[2]},Y,T_2)  \nonumber \\
    & \quad + I(X_1;Y_{2}|Y,X_2,T_2) + I(X_2;Y_{1}|Y,X_1,T_2).
\end{align}
One obtains \eqref{eq:mac-gf-uc} by rewriting terms.
\end{remark}

The paper \cite{tandon2011dependence} studied two types of feedback:
\begin{itemize} \itemsep 0pt
    \item \emph{Noisy feedback}: $Y_{1}=Y+N_1'$, $Y_{2}=Y+N_2'$ where $N,N_1',N_2'$ are independent;
    \item \emph{User cooperation}: $Y_{1}=g_{21} X_2+N_1$, $Y_{2}=g_{12} X_1+N_2$, and $N,N_1,N_2$ are independent;
\end{itemize}
The two types of feedback are related. For example, under noisy feedback, users 1 and 2 can compute $\tilde{Y}_{1}=g_2 X_2+(N+N_1)$ and $\tilde{Y}_{2}=g_1 X_1+(N+N_2)$, respectively. Thus, noisy feedback is a special case of user cooperation with correlated noise. We discuss the noisy feedback setting in Appendix \ref{noisy-feedback}.

\begin{theorem}
    For user cooperation where $N,N_1,N_2$ are mutually independent, the bound in Corollary \ref{cor-6} collapses to the bound in Proposition~\ref{prop3}.
\end{theorem}
\begin{proof}
The bound in Corollary \ref{cor-6} is always a subset of the bound in Proposition~\ref{prop3}. To show the other direction, it suffices to show that the maximum weighted sum-rate $\lambda R_1+R_2$ of the region in Proposition~\ref{prop3} is less than or equal to the  maximum weighted sum-rate $\lambda R_1+R_2$ of the region in Corollary \ref{cor-6} for any arbitrary $\lambda\geq 1$ (the proof for $R_1+\lambda R_2$ is similar). Assume that $(R^*_1,R^*_2)$ reaches the maximum weighted sum-rate $\lambda R_1+R_2$ of the region in Proposition~\ref{prop3}  via some $p_{X_1,X_2,T}$. It suffices to show that  $(R^*_1,R^*_2)$ also belongs to the region in Corollary \ref{cor-6}.
 
We claim that there is a maximizer $p_{X_1,X_2,T}$ for the $\lambda$ sum-rate of the region  in Proposition~\ref{prop3} satisfying
\begin{align}
    I(X_1,X_2;Y) \leq I(X_{1},X_{2};Y,Y_{{1}},Y_{{2}}|T).\label{eqneqTU}
\end{align}
We first show how to complete the proof assuming  \eqref{eqneqTU}. We show  $(R_1,R_2)$ belongs to the region in Corollary \ref{cor-6} with the choice of $T_1=T$ and $T_2$ being a constant random variable. For user cooperation, the bounds \eqref{DBMAC-21}--\eqref{DBMAC-2b} for the choice of $T_1=T$ and $T_2$ being a constant reduce to
\begin{subequations}
\begin{align}
    R^*_{1} & \leq  \min\left(\, I(X_{1};Y,Y_{{2}}|X_{2},T),\,
    I(X_1;Y,Y_{{2}}|X_2) \,\right)
    \label{DBMAC-21-2} \\
    R^*_{2} & \leq \min\left(\, I(X_{2};Y,Y_{{1}}|X_{1},T),\, I(X_2;Y,Y_{{1}}|X_1) \,\right)
    \label{DBMAC-22-2} \\
    R^*_{1}+R^*_{2} & \leq \min\left(\, I(X_{1},X_{2};Y,Y_{{1}},Y_{{2}}|T), \, I(X_{1},X_{2};Y) \,\right)
    \label{DBMAC-23-2} \\
    I(X_{1};X_{2}|T) & \leq I(X_{1};X_{2}|Y_{1},Y_{2},T)
    \label{DBMAC-2a-2} \\
    I(X_1;X_2) & \leq I(X_1,Y_{1}; X_2,Y_{2} | Y) .
    \label{DBMAC-2b-2}
\end{align}
\end{subequations}

Note that the second bounds in \eqref{DBMAC-21-2}-\eqref{DBMAC-22-2} are redundant by the inequalities
\begin{subequations}
\begin{align}
I(X_{1};Y,Y_{{2}}|X_{2},T)
    &\le I(X_{1};Y,Y_{{2}}|X_{2}) 
    \label{mac-gf-uc-1} \\
I(X_{2};Y,Y_{{1}}|X_{1},T)
    &\le I(X_{2};Y,Y_{{1}}|X_{1}).
    \label{mac-gf-uc-2}
\end{align}
\end{subequations}
Compared to the constraints in Proposition~\ref{prop3}, we need to show \eqref{DBMAC-2b-2}. 
It is shown in \cite[Eq. (151)]{tandon2011dependence} that the constraint \eqref{tudbconst}
implies $I(X_1;X_2|T)=0$.
Since $I(X_1;X_2|T)=0$, we obtain
\begin{align}
I(X_{1},X_{2};Y,Y_{{1}},Y_{{2}}|T)
    & \le I(X_1;Y,Y_{{1}},Y_{{2}}|X_2,T) + I(X_2;Y,Y_{{1}},Y_{{2}}|X_1,T)
    \nonumber \\
    & = I(X_1;Y,Y_{{2}}|X_2,T) + I(X_2;Y,Y_{{1}}|X_1,T)\label{submodTU1}
\end{align}
where the last step uses the independence of $N_1,N_2,N$. Observe that
\begin{align}
    I(X_{1},X_{2};Y)
    & \leq I(X_{1},X_{2};Y,Y_{{1}},Y_{{2}}|T) \quad \text{... by \eqref{eqneqTU}}
    \nonumber \\
    & \le I(X_1;Y,Y_{{2}}|X_2,T) + I(X_2;Y,Y_{{1}}|X_1,T)
    \nonumber \\
    & \le I(X_1;Y,Y_{{2}}|X_2) + I(X_2;Y,Y_{{1}}|X_1).
\end{align}
by \eqref{mac-gf-uc-1}-\eqref{mac-gf-uc-2}. 
This bound is the same as \eqref{DBMAC-2b-2}.

It remains to prove \eqref{eqneqTU}. Due to the submodularity constraint \eqref{submodTU1}, the maximum weighted sum-rate is
\begin{align}
    \max\lambda R_1+R_2
    = \max_{p_{T}p_{X_1|T}p_{X_2|T}} 
    & (\lambda-1) I(X_1;Y,Y_2|X_2,T)
    \nonumber \\
    & + \min\left(\, I(X_{1},X_{2};Y,Y_{{1}},Y_{{2}}|T), \, I(X_{1},X_{2};Y) \,\right).\label{eqnMLL}
\end{align}
The paper \cite{tandon2011dependence} shows there is a maximizer with $T$ a scalar (Gaussian) random variable. Suppose
\begin{equation}
    I(X_1X_2;Y) > I(X_1X_2;Y,Y_1,Y_2|T)
\end{equation}
holds for this maximizer so $I(X_1X_2;Y,Y_1,Y_2|T)$ is the (strictly) minimizing term in \eqref{eqnMLL}. 
Using $I(X_1;X_2|T)=0$, we can write
\begin{subequations}
\begin{align}
   X_1 & = a_1 T+b_1G_1 \\
   X_2 & = a_2 T+b_2G_2
\end{align}
\end{subequations}
for independent standard normal variables $T,G_1,G_2$. First, assume that $a_1>0$. If we decrease $a_1$ and increase $b_1$ such that $a_1^2\mathsf{Var}[T]+b_1^2$ is preserved, the variance of $X_1$ will be preserved while the terms  $I(X_1;Y,Y_2|X_2,T)$ and $I(X_{1},X_{2};Y,Y_{{1}},Y_{{2}}|T)$ would increase, a contradiction.  Thus, we must have $a_1=0$. A similar argument shows that $a_2=0$, since decreasing $a_2$ would increase the expression in \eqref{eqnMLL}. However, if $a_1=a_2=0$, we have
\begin{equation}
    I(X_1X_2;Y,Y_1,Y_2|T)=I(X_1X_2;Y,Y_1,Y_2)\geq I(X_1X_2;Y)
\end{equation}
which contradicts our assumption.
\end{proof}

\begin{remark}
    Choosing $Z_3=(Y,Y_{1},Y_{2})$ gives the same rate bounds as $Z_1=(Y_{1},Y_{2})$ (with a $T_3$ rather than a $T_2$) but with the DB constraint
\begin{align}
    I(X_1 ; X_2 | T_3)
    \le I(X_1 ; X_2 | Y,Y_1,Y_2,T_3).
    \label{eq:mac-gf-uc2}
\end{align}
    By choosing $p_{T_3|X_1,X_2}=p_{T_1|X_1,X_2}$ we have $X_1\mkv T_3\mkv X_2$. Thus, the bound \eqref{eq:mac-gf-uc2} is redundant, and so are the rate bounds. This shows that this choice of $Z_3$ is redundant.
\end{remark}

We show that a more sophisticated choice of $Z_1$ and $Z_2$ strictly improves the bound in Proposition~\ref{prop3} for the user cooperation setup. First, as discussed in Remark \ref{rmk16}, Proposition~\ref{prop3} gives the following bound when $R_2=0$:
\begin{align}
    R_1 \le \max_{P_{X_1,X_2}} \min \left( I(X_1;Y,Y_1,Y_2|X_2), I(X_1,X_2;Y) \right).
\end{align}
For user-cooperation, $I(X_1;Y,Y_1,Y_2|X_2)=I(X_1;Y,Y_2|X_2)$ and the above bound reduces to the cut-set bound. Therefore, we must improve on the cut-set bound. Observe that the scalar Gaussian relay channel is a special case of user cooperation when $g_{21}=0$ and $R_2=0$. Thus, it suffices to improve the cut-set bound for the scalar relay channel. This is done in the next subsection.

\subsubsection{Relay Channel}

Fig.~\ref{setup:relay-1} shows a relay channel $p(y_{\mathrm{r}}, y | x, x_{\mathrm{r}})$ with $k=3$ transceivers. The bound in Theorem \ref{thm2} yields the following for $a=1$ auxiliary receiver ($\mathcal{U}$ is the set with the transmitter and relay indexes):
\begin{subequations}
\begin{align}
    & R \leq \min\big[ I(X;Y,Y_{\mathrm{r}},Z_1|X_{\mathrm{r}},T_1), I(X,X_{\mathrm{r}};Y,Z_1|T_1) \big] \label{bt21}\\
    & R \leq \min\big[ I(X;Y,Y_{\mathrm{r}}|X_{\mathrm{r}}), I(X,X_{\mathrm{r}};Y) \big] \label{bt22}
\end{align}
\end{subequations}
for some $p_{X_1,X_\rmr,T_1}$ satisfying
\begin{align}
    & I(X;X_{\mathrm{r}}|T_1) \leq I(X;X_{\mathrm{r}},Y_{\mathrm{r}}|T_1,Z_1).\label{bt23}
\end{align}

\begin{figure*}[t!]
\centering
	\begin{tikzpicture}
	\node at (-0.95,-0.95) {$M$};
	\draw [->,thick] (-0.6,-1) -- (0.2,-1);
	\draw (0.2,.-1.5) rectangle +(1.6,1); \node at (1,-1) {Encoder};
	\draw [->,thick] (1.8,-1)-- (3.1,-1); \node at (2.45,-0.7) {$X^n$};
	\draw [->,thick] (3.6,-0.5)-- (3.6,0.75); \node at (3.1,0.25) {$Y_{\mathrm{r}}^{i-1}$};
	\draw [<-,thick] (4.9,-0.5)-- (4.9,0.75); \node at (5.3,0.25) {$X_{\mathrm{r}i}$};
	\draw (3.1,-1.5) rectangle +(2.4,1); \node at (4.3,-1) {$p(y,y_{\mathrm{r}}|x,x_{\mathrm{r}})$};
	\draw (3.1,0.75) rectangle +(2.4,1); \node at (4.3,1.25) {Relay Encoder};
	\draw [->,thick] (5.5,-1) --(6.6,-1); \node at (6.05,-.7) {$Y^n$};
	\draw (6.6,-1.5) rectangle +(1.6,1); \node at (7.4,-1) {Decoder};
	\draw [->,thick] (8.2,-1) --(9.0,-1); \node at (9.25,-0.945) {$\hat M$};
	\end{tikzpicture}
	\caption{Relay channel.}
	\label{setup:relay-1}
\end{figure*}

The Gaussian relay channel is characterized by the equations:
\begin{subequations}
\begin{align} 
   Y_\mathrm{r} & = g_{12} X + N_\mathrm{r} \\
   Y & = g_{13} X + g_{23} X_\mathrm{r} + N_e
\end{align}
\end{subequations}
where $g_{12}$, $g_{13}$, and $g_{23}$ are channel gain coefficients, while $N_e \sim \mathcal{N}(0,1)$ and $N_\mathrm{r} \sim \mathcal{N}(0,1)$ are independent Gaussian noise terms. Additionally, both input signals $X$ and $X_\mathrm{r}$ are subject to an average power constraint $P$. Let $C(P)$ be the capacity under the power constraint $P$.
The cut-set bound is
\begin{equation}
    \max \min\{I(X,X_{\mathrm{r}};Y),
I(X;Y,Y_{\mathrm{r}}|X_{\mathrm{r}})\}
\end{equation}
where the maximum is over $P_{X,X_\rmr}$ satisfying the power constraints
\begin{align}
    \mathbb{E}[X^2]\le P,\qquad \mathbb{E}[X_\mathrm{r}^2]\le P. \label{eqnpwc}
\end{align}
The cut-set bound is optimized by Gaussian inputs \cite[Sec 16.2]{elk11}. Define
\begin{equation}
   \text{Cut-set}(P_1,P_2,\rho)
   = \min\{I(X,X_{\mathrm{r}};Y),
   I(X;Y,Y_{\mathrm{r}}|X_{\mathrm{r}})\}
\end{equation}
with $(X,X_\rmr)$ distributed as
\begin{align}
    (X,X_{\mathrm{r}}) \sim  \mathcal{N}\left(0,\begin{bmatrix} P_1 & \rho \sqrt{P_1P_2} \\ \rho \sqrt{P_1P_2} & P_2 \end{bmatrix}\right).\label{eqn190} 
\end{align}
The cut-set bound states that
\begin{align}
    C(P)&\leq \max_{P_1\leq P, P_2\leq P, \rho\in[-1,1]}\text{Cut-set}(P_1,P_2,\rho)
    \label{eqnUB1}.
\end{align}

Next, consider a Gaussian auxiliary channel of the form
\begin{align}
Z_1&=\alpha X+\beta X_\mathrm{r}+\gamma N_e+\eta N_\mathrm{r}+\zeta N
\end{align}
where $N_e, N_\rmr, N$ are mutually independent standard normal variables. The bound in \eqref{bt21}-\eqref{bt22} applies under the constraints \eqref{eqnpwc}, and jointly Gaussian inputs optimize the bound. For $(X,X_\rmr)$ distributed as in \eqref{eqn190}, let $\text{UB}(P_1,P_2,\rho)$ be the maximum of
\begin{align}
    & \min\big[ I(X;Y,Y_{\mathrm{r}},Z_1|X_{\mathrm{r}},T_1), I(X,X_{\mathrm{r}};Y,Z_1|T_1) \big]
\end{align}
over jointly Gaussian $P_{T_1,X,X_\rmr}$ satisfying \eqref{eqn190} and
\begin{align}
    & I(X;X_{\mathrm{r}}|T_1) \leq I(X;X_{\mathrm{r}},Y_{\mathrm{r}}|T_1,Z_1)\label{const-188}.
\end{align}
The upper bound in Theorem \ref{thm2} for auxiliary variable $Z_1$ is 
\begin{align}
    C(P)&\leq \max_{P_1\leq P, P_2\leq P, \rho\in[-1,1]}\min\left\{\text{UB}(P_1,P_2,\rho),\text{Cut-set}(P_1,P_2,\rho)\right\}.\label{eqnUB2}
\end{align}

\begin{lemma}\label{lemma2relay}
Let $\mathcal{S}$  be the set of all $(Q_1,Q_2,\tilde{\rho})$ such that $Q_1\in[0,P]$, $Q_2\in[0,P]$ and $\tilde\rho\in[-1,1]$ satisfy
\begin{equation} 
\begin{bmatrix} Q_1 & \tilde{\rho} \sqrt{Q_1Q_2} \\ \tilde{\rho} \sqrt{Q_1Q_2} & Q_2 
\end{bmatrix}\preceq \begin{bmatrix} P & \rho P \\ \rho P & P
\end{bmatrix}.
\end{equation}
We have
\begin{equation}
\text{UB}(P,P,\rho) = \max_{(Q_1,Q_2,\tilde{\rho})\in\mathcal{S}}\min(F_1,F_2)
\end{equation}
subject to 
\begin{align}
    &\log\left(\gamma^2+\zeta^2+Q_1 (1 - \tilde{\rho}^2) \left[ (\alpha - \eta g_{12})^2 + g_{12}^2 (\gamma^2 + \zeta^2) \right]\right)-\log(\gamma^2+\zeta^2)
    \nonumber \\ & \ge \log(\alpha^2Q_1+\beta^2Q_2+2\alpha\beta\tilde{\rho} \sqrt{Q_1Q_2}+\gamma^2+\eta^2+\zeta^2)-\log(\beta^2Q_2(1-\tilde{\rho}^2)+\gamma^2+\eta^2+\zeta^2).
\end{align}
Here, we have
\begin{align}
   F_1 &= \frac12\log\left(\zeta^2 + Q_1(1-\tilde{\rho}^2)\left[(g_{12}\eta + g_{13}\gamma - \alpha)^2 + \zeta^2(g_{12}^2 + g_{13}^2)\right]\right)-\frac12\log(\zeta^2)
   \\
    F_2 &= \frac12\log
    \bigg\{(g_{13}^2Q_1+g_{23}^2Q_2+2g_{13}g_{23}\tilde{\rho} \sqrt{Q_1Q_2}+1)(\alpha^2Q_1+\beta^2Q_2+2\alpha\beta\tilde{\rho} \sqrt{Q_1Q_2}+\gamma^2+\eta^2+\zeta^2)
    \nonumber \\ & \qquad\qquad- (\alpha g_{13}Q_1+\beta g_{23}Q_2+(\alpha g_{23}+\beta g_{13})\tilde{\rho} \sqrt{Q_1Q_2}+\gamma )^2\bigg\}-\frac12\log(\eta^2+\zeta^2).
\end{align}
\end{lemma}
\begin{proof}
    See Appendix \ref{RelayAppendix}.
\end{proof}

We next show that the upper bound in \eqref{eqnUB2} can improve the cut-set bound \eqref{eqnUB1}. Consider a Gaussian relay channel with
\begin{subequations}
\begin{align} 
   Y_\mathrm{r} &= 0.97 X+N_\mathrm{r} \label{eqnGRC1} \\
   Y &= 0.85 X+0.02 X_\mathrm{r}+N_e
   \label{eqnGRC2}
\end{align}
\end{subequations}
and the power constraints $P=1$ on $X$ and $X_{\rmr}$.
Consider the auxiliary receiver
\begin{align}
   Z_1 & = 1.26 X+0.16 X_\mathrm{r}+ N_e+N_\mathrm{r} .
\end{align}
The maximum in 
\eqref{eqnUB1} is obtained \emph{uniquely} at $P_1=P_2=P$ and some $\rho^* \in (0,1)$ satisfying
$I(X,X_{\mathrm{r}};Y) = I(X;Y,Y_{\mathrm{r}}|X_{\mathrm{r}})$.
Thus, to show that \eqref{eqnUB2} strictly improves \eqref{eqnUB1}, it suffices to restrict to $P_1=P_2=P$. 
The resulting functions $\rho\mapsto UB(1,1,\rho)$ and $\rho\mapsto\text{Cut-set}(1,1,\rho)$ are plotted in Fig.~\ref{fig1}. The curve $\rho\mapsto\text{Cut-set}(1,1,\rho)$ is maximized at $\rho^*_1\approx 0.741...$; it is strictly increasing for $\rho\leq \rho^*_1$ and strictly decreasing for $\rho> \rho^*_1$. 
As the figure shows, we have
\begin{equation}
   \text{UB}(1,1,\rho_1^*)<\text{Cut-set}(1,1,\rho_1^*).
\end{equation}

Note that $\rho\mapsto\text{UB}(1,1,\rho)$ is maximized at $\rho^*_2\approx 0.395$; the curve is strictly increasing for $\rho\leq \rho^*_2$ and strictly decreasing for $\rho> \rho^*_2$. For $\rho\leq \rho^*_2$, the constraint \eqref{const-188} is inactive for the maximizer $p_{T_1|X,X_\rmr}$, while \eqref{const-188} holds with equality for the maximizer when $\rho> \rho^*_2$.

\begin{figure}[t!]
    \begin{center}
    \includegraphics[scale=0.45]{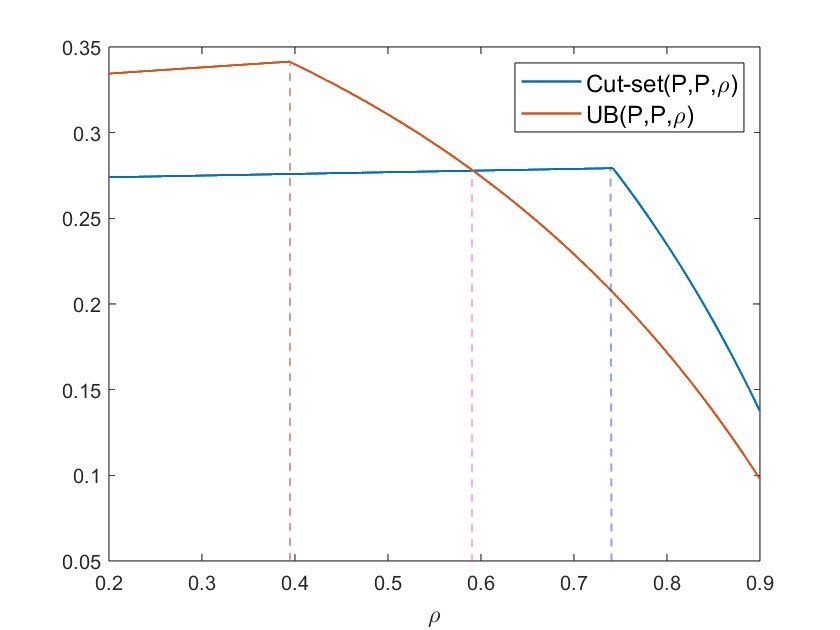} 
    \end{center}
    \caption{The cut-set and dependence balance bounds for a Gaussian relay channel.}
\label{fig1}
\end{figure}

The above result is noteworthy because the cut-set bound for Gaussian relay channels was only recently improved in \cite{gohari2021outer,el2022strengthened}.
The relationship between \eqref{eqnUB2} and the bound in \cite{el2022strengthened} is unclear.
The upper bounds proposed in \cite[Theorem 1]{gohari2021outer} and \cite{el2022strengthened} utilize a different auxiliary random variable identification $(Y^{i-1},J_{i+1}^n)$ (in \cite{el2022strengthened}, $J$ is taken as $Y_\rmr$). 
Our limited numerical simulations did not identify cases where \eqref{eqnUB2} improves upon \cite{el2022strengthened}, but further investigation is warranted. Note that our general DB bound has the distinct advantage of applying to any multiterminal network, whereas the bounds in \cite{el2022strengthened} are limited to the one-relay setting.

\subsubsection{Choice of Auxiliary Receivers}

Hekstra and Willems consider MACs with a single output $Y_{1}=Y_{2}=Y$. Moreover, they show that a judicious choice of the auxiliary receiver may lead to capacity \cite[Section V]{hekstra1989dependence}. Consider $a=2$, $Z_1=(X_1,Y_{1})$, and $Z_2$ is a constant (cut-set bound). This leads to the following bound.

\begin{corollary} \label{cor-7}
Consider a two-user MAC with generalized feedback. Any achievable $(R_1, R_2)$ satisfies
\begin{subequations}
\begin{align}
    R_{1} & \leq  \min\left(H(X_{1}|T_1), I(X_1;Y,Y_{2}|X_2) \right) \\
    R_{2} & \leq  I(X_{2};Y,Y_{{1}}|X_{1},T_1),
\end{align}
\end{subequations}
for some $p(t_1,x_1,x_2)\cdot p(y,y_{1},y_{2}|x_1,x_2)$ satisfying
\begin{align}
    I(X_{1};X_{2}|T_1) & =0.
\end{align}
\end{corollary}

The above bound generalizes the one in \cite{hekstra1989dependence} and reduces to the outer bound in \cite[Theorem 3]{zhang1986new} for $Y=Y_{1}=Y_{2}$. The above bound is tight for some MAC channels with feedback; see Section V and Corollary 2 in \cite{hekstra1989dependence}. 
We provide another example, showing that a careful choice of the auxiliary receiver gives good bounds. First, consider the special case $Y_{1}=Y_{2}=Y$. In this case, Corollary \ref{cor-6} simplifies to
\begin{subequations}
\begin{align}
    R_{1} & \leq I(X_{1};Y|X_{2},T)\\
    R_{2} & \leq I(X_{2};Y|X_{1},T)\\
    R_{1}+R_{2} & \leq I(X_{1},X_{2};Y|T) \\
    I(X_{1};X_{2}|T) & \leq I(X_{1};X_{2}|Y,T)
\end{align}
\end{subequations}
for some $p(t,x_1,x_2)\cdot p(y|x_1,x_2)$. On the other hand, one may alternatively set $Y_{1}=(Y,X_1)$ and $Y_{2}=(Y,X_2)$ because the $i$-th transmitter knows $X_i$. With this choice, Corollary \ref{cor-6} gives the bounds
\begin{subequations}
\begin{align}
    R_1 & \leq H(X_1|T) \label{eq:example-sec434-1} \\
    R_2 & \leq H(X_2|T) \label{eq:example-sec434-2} \\
    R_1+R_2 & \leq I(X_1,X_2;Y) \label{eq:example-sec434-3} \\
    I(X_1;X_2|T) & \leq 0. \label{eq:example-sec434-4}
\end{align}
\end{subequations}
These bounds can be loose. For example, suppose $T,X_1,X_2$ are jointly Gaussian with an invertible covariance matrix satisfying the Markov chain $X_1 \mkv T \mkv X_2$. In this case, $H(X_1 | T)$ and $H(X_2 | T)$ become infinite. This shows that when $Y_{1}=(Y,X_1)$ and $Y_{2}=(Y,X_2)$ choosing the auxiliary receiver $Z_1=(Y_{1},Y_{2})$ may not be a good idea because $Z_1$ will include both $X_1$ and $X_2$. 

\begin{remark}
    The region defined by \eqref{eq:example-sec434-1}-\eqref{eq:example-sec434-4} is the capacity region of MACs where $X_1=f_1(X_2,Y)$ and $X_2=f_2(X_1,Y)$ for some functions $f_1(.)$ and $f_2(.)$; see \cite{Cover-Leung-IT81} and \cite{willems82,willems-IT82}. For example, the binary adder channel with $Y=X_1+X_2$ and $\mathcal X_1=\mathcal X_2 = \{0,1\}$ has this property.
\end{remark}

\subsection{Communication under Privacy Constraints}

One can develop a version of Theorem \ref{thm2} for privacy constraints. For example, we derive an outer bound for a relay broadcast channel with such constraints. Consider a relay channel $p(y, y_\mathrm{r} | x, x_\mathrm{r})$ as above. The transmitter aims to send a private message $M_1$ to the relay (partially hidden from the destination) and a message $M_2$ to the destination; see Fig.~\ref{setup:broadcastrelay}. This setting is referred to as the ``cooperative relay broadcast channel with a single-sided cooperative link'' in \cite{ekrem2010secrecy}.

\begin{figure*}[t!]
\centering
	\begin{tikzpicture}
	\node at (-1.5,-0.95) {$(M_1,M_2)$};
	\draw [->,thick] (-0.6,-1) -- (0.2,-1);
	\draw (0.2,.-1.5) rectangle +(1.6,1); \node at (1,-1) {Encoder};
	\draw [->,thick] (1.8,-1)-- (3.1,-1); \node at (2.45,-0.7) {$X^n$};
	\draw [->,thick] (3.6,-0.5)-- (3.6,0.75); \node at (3.1,0.25) {$Y_{\mathrm{r}}^{i-1}$};
	\draw [<-,thick] (4.9,-0.5)-- (4.9,0.75); \node at (5.3,0.25) {$X_{\mathrm{r}i}$};
	\draw (3.1,-1.5) rectangle +(2.4,1); \node at (4.3,-1) {$p(y,y_{\mathrm{r}}|x,x_{\mathrm{r}})$};
	\draw (3.1,0.75) rectangle +(2.4,1); \node at (4.3,1.25) {Relay };
	\draw [->,thick] (5.5,-1) --(6.6,-1); \node at (6.05,-.7) {$Y^n$};
	\draw (6.6,-1.5) rectangle +(1.6,1); \node at (7.4,-1) {Decoder};
	\draw [->,thick] (8.2,-1) --(9.0,-1); \node at (9.25,-0.945) {$\hat M_2$};
    \draw [->,thick] (5.5,1.25) --(6.2,1.25); \node at (6.45,1.25) {$\hat M_1$};
	\end{tikzpicture}
	\caption{Memoryless relay broadcast channel setup.}
	\label{setup:broadcastrelay}
\end{figure*}

Due to the privacy constraint, the transmitter and the relay may wish to use private randomization. Let $W$ and $W_\mathrm{r}$ be the private randomness available at the transmitter and relay, respectively. We assume $M_1, M_2, W, W_\mathrm{r}$ are mutually independent, and the message pair $(M_1, M_2)$ has the rates $(R_1, R_2)$. Apart from the usual reliability constraints, we impose the privacy constraint 
\begin{align}
    \frac{1}{n} H(M_1| Y^n) \geq R_{e_1}-\epsilon
\end{align}
on the information the destination gains about $M_1$.
One may, as in \cite{ekrem2010secrecy}, also consider a privacy constraint
\begin{align}
\frac{1}{n} H(M_2| Y_\rmr^n,X_\rmr^n,W_\rmr) \geq R_{e_2}-\epsilon.
\end{align}
for $M_2$. However, as pointed out in \cite{ekrem2010secrecy}, the case with $R_{e_2}=R_2=0$ is already challenging. The authors of \cite[Remark 10, Remark 13]{ekrem2010secrecy} claim that deriving an upper bound on $R_{e_1}$ based solely on the channel inputs and outputs is unlikely to be feasible because the relay can leverage its observation \( Y_{\rmr} \) to encode its input \( X_{\rmr} \), introducing temporal correlation between its channel inputs and outputs. Additionally, the relay can enhance its own secrecy rate by transmitting jamming signals. However, we prove the following simple bound:
\begin{align}
    R_{e_1}\leq \max_{p(x,x_\rmr)}I(X; Y_{\mathrm{r}},X_{\mathrm{r}}| Y) - I(X; X_{\mathrm{r}}).\label{eqnE1b}
\end{align}

Consider first the outer bound in \cite{ekrem2010secrecy} for the set of achievable triples rates $(R_1,R_2,R_{e_1})$:

\begin{theorem}[\cite{ekrem2010secrecy}]
    A rate triple $(R_1, R_2, R_{e_1})$ is achievable only if 
\begin{subequations}
\begin{align}
& R_1 \leq I(V_1; Y_{\mathrm{r}} | X_{\mathrm{r}}) \\
& R_2 \leq I(V_2; Y ) \\
& R_{e_1} \leq \min(R_1, I(V_1;Y_{\mathrm{r}} | U)-I(V_1;Y|U),I(V_1;Y_{\mathrm{r}} | V_2)-I(V_1;Y|V_2))
\end{align}
\end{subequations}
for some joint distribution
$p_{X, X_\mathrm{r}} p_{Y,Y_\rmr|X,X_\rmr} p_{V_1,V_2|X,X_\rmr,Y_\rmr} p_{U|V_1,V_2}$.
\end{theorem}
\noindent Observe that the optimal choice for $V_1$ is $Y_\rmr$ since all terms increase when we replace $V_1$ by $V_\rmr$. For instance, we have
\begin{align}
    I(V_1;Y_{\mathrm{r}} | U)-I(V_1;Y|U)&\leq H(Y_{\mathrm{r}} | U)-I(V_1,Y_{\mathrm{r}} ;Y|U)
    \nonumber \\
    & \le H(Y_{\mathrm{r}} | U)-I(Y_{\mathrm{r}} ;Y|U)
\end{align}
Moreover, without loss of generality, we can set $U=V_2$. Thus, the bound reduces to
\begin{subequations}
\begin{align}
   & R_1 \leq H(Y_{\mathrm{r}} | X_{\mathrm{r}}) \\
   & R_2 \leq I(V_2; Y ) \\
   & R_{e_1} \le \min(R_1, H(Y_{\mathrm{r}} | V_2,Y)) .
\end{align}
\end{subequations}
for some
$p_{X, X_\mathrm{r}} p_{Y,Y_\rmr|X,X_\rmr} p_{V_2|X,X_\rmr,Y_\rmr}$. Note that the above bound becomes vacuous for Gaussian channels as $H(Y_\rmr|X_\rmr)=\infty$. 

Next, we develop a version of Theorem \ref{thm2} for the setting in Fig.~\ref{setup:broadcastrelay}. This upper bound is the cut-set bound with a DB constraint appearing as an equivocation rate constraint. This outer bound implies the inequality claimed in \eqref{eqnE1b}.

\begin{theorem}\label{thmbs}
A rate triple $(R_1, R_2, R_{e_1})$ is achievable only if
\begin{subequations}
\begin{align}
   & R_1 + R_2 \leq I(X; Y, Y_{\mathrm{r}} | X_{\mathrm{r}}, T_1) \label{eqnst1} \\
   & R_2 \leq I(X, X_\mathrm{r}; Y | T_1) \label{eqnst2} \\
   & R_{e_1} \leq I(X; Y_{\mathrm{r}},X_{\mathrm{r}}|T_1, Y) - I(X; X_{\mathrm{r}} | T_1) \label{eqnst3}
\end{align}
\end{subequations}
for some joint distribution $p_{X, X_\mathrm{r}, T_1}$.
\end{theorem}
\begin{proof}
Equations \eqref{eqnst1} and \eqref{eqnst2} follow from the constraint \eqref{eqn50-2} for the choice $Z_1 = Y$. The DB constraint in Theorem \ref{thm2} for the set $\mathcal{U}$ consisting of the transmitter and the relay yields
\begin{equation}
    I(X; X_{\mathrm{r}} | T_1) \leq I(X; X_{\mathrm{r}}, Y_{\mathrm{r}} | T_1, Y)
\end{equation}
which is weaker than \eqref{eqnst3}. However, Lemma \ref{lemma1} yields
\begin{equation}
   \frac{1}{n} I(W, X^n; W_\mathrm{r}, Y_\mathrm{r}^n | Y^n) \leq I(X; Y_{\mathrm{r}},X_{\mathrm{r}}| T_1, Y) - I(X; X_{\mathrm{r}} | T_1)
\end{equation}
and instead of bounding
$\frac{1}{n} I(W, X^n; W_\mathrm{r}, Y_\mathrm{r}^n | Y^n)$
by zero as in the proof of Theorem \ref{thm2}, it can be bounded from below by $R_{e_1}$, yielding \eqref{eqnst3}.
\end{proof}

\section{Variants of the DB bound in Lemma~\ref{lemma1}}

{\color{black}
We expect there will be many variations of the DB bound in Lemma~\ref{lemma1}. For example, one variant for diamond channels is as follows.
\begin{lemma}\label{lemmar3}
    Suppose there is set \(\mathcal{V}\subset [k]-\mathcal{U}\) for which, under the setup of Lemma~\ref{lemma1}, we have
    \begin{align}
        P_{Y_{\mathcal{V}}|X_{[k]}} = P_{Y_{\mathcal{V}}|X_{\mathcal{V}}}.
        \label{eqn-db-const-1}
    \end{align}
    Then we have the bound
    \begin{align}
        &I_{\lambda}(W_{i_1}Y_{i_1}^n;W_{i_2}Y_{i_2}^n;\cdots;W_{i_u}Y_{i_u}^n|Z^n) - I_{\lambda}(W_{i_1};W_{i_2};\cdots;W_{i_u}) \nonumber \\
        &\le \sum_{j\in[n]} \Bigg[ I_{\lambda}(X_{i_1j}Y_{i_1j};X_{i_2j}Y_{i_2j};\cdots;X_{i_uj}Y_{i_uj}|Z^{j-1},Z_j) - I_{\lambda}(X_{i_1j};X_{i_2j};\cdots;X_{i_uj}|Z^{j-1}) \nonumber \\
        &\qquad \qquad - \left(1-\sum_{\mathcal B\subsetneq \mathcal{U}}\lambda_{\mathcal B}\right) \left\{ I(X_{[k]j} ; Z_j Y_{\mathcal{U}j} | Z^{j-1} X_{\mathcal{U}j}) + I( Y_{\mathcal{V}j}; Z_jY_{\mathcal{U}j} \big| X_{[k]j})\right\} \Bigg] \nonumber \\
        &\qquad + \left(1-\sum_{\mathcal B\subsetneq \mathcal{U}}\lambda_{\mathcal B}\right) I( X_{\mathcal{V}}^n; Z^nY_{\mathcal{U}}^n \big| W_{\mathcal{U}} ).
    \end{align}
\end{lemma}

\begin{proof}
The gap in step $(d)$ in \eqref{eqn16justn} is
\begin{align}
    - \left(1-\sum_{\mathcal B\subsetneq [u]}\lambda_{\mathcal B}\right) I( X_{[k]j}; Z_jY_{[u]j} \big| Z^{j-1} X_{[u]j} W_{[u]} Y_{[u]}^{j-1}).
\end{align}
We have
\begin{align}
    &\sum_{j\in[n]} I( X_{[k]j}; Z_jY_{[u]j} \big| Z^{j-1} X_{[u]j} W_{[u]} Y_{[u]}^{j-1}) \nonumber \\
    &= \sum_{j\in[n]} I( X_{[k]j}; Z_jY_{[u]j} \big| Z^{j-1} W_{[u]} Y_{[u]}^{j-1}) \nonumber \\
    &= \sum_{j\in[n]} I( X_{[k]j}Y_{\mathcal{V}j}; Z_jY_{[u]j} \big| Z^{j-1} W_{[u]} Y_{[u]}^{j-1}) - \sum_{j\in[n]} I( Y_{\mathcal{V}j}; Z_jY_{[u]j} \big| X_{[k]j}) \nonumber \\
    &= \sum_{j\in[n]} I( X_{\mathcal{V}}^nX_{[k]j}Y_{\mathcal{V}j}; Z_jY_{[u]j} \big| Z^{j-1} W_{[u]} Y_{[u]}^{j-1}) - \sum_{j\in[n]} I( Y_{\mathcal{V}j}; Z_jY_{[u]j} \big| X_{[k]j}) \label{eqnj1} \\
    &\geq \sum_{j\in[n]} I( X_{\mathcal{V}}^n; Z_jY_{[u]j} \big| Z^{j-1} W_{[u]} Y_{[u]}^{j-1}) - \sum_{j\in[n]} I( Y_{\mathcal{V}j}; Z_jY_{[u]j} \big| X_{[k]j}) \nonumber \\
    &= I( X_{\mathcal{V}}^n; Z^nY_{[u]}^n \big| W_{[u]} ) - \sum_{j\in[n]} I( Y_{\mathcal{V}j}; Z_jY_{[u]j} \big| X_{[k]j})
\end{align}
where in \eqref{eqnj1} we used the Markov chain
\begin{align}
    (Z^{j-1},W_{[u]},Y_{[u]}^{j-1},X_{\mathcal{V}}^n) \mkv (X_{[k],j},Y_{\mathcal{V},j}) \mkv (Z_j,Y_{[u],j})
\end{align}
which follows from \eqref{eqn-db-const-1}.
\end{proof}
}

\begin{figure}[h]
\centering
\begin{tikzpicture}[
    >=stealth,
    node distance=1.5cm and 1.5cm,
    block/.style={draw, rectangle, minimum height=1cm, minimum width=1.5cm, align=center},
    mac/.style={draw, rectangle, minimum height=4cm, minimum width=1.5cm, align=center}
]

\node[block] (source) {Source};
\node[block, right=1cm of source] (encoder) {Encoder};

\node[block, above right=0.5cm and 1.5cm of encoder] (relay1) {Relay 1};
\node[block, below right=0.5cm and 1.5cm of encoder] (relay2) {Relay 2};

\node[mac, right=4.5cm of encoder] (mac) {MAC \\ \(p_{Y|X_1,X_2}\)};

\node[block, right=1.5cm of mac] (decoder) {Decoder};
\node[block, right=1cm of decoder] (sink) {Sink};

\draw[->] (source) -- node[above] {\(M\)} (encoder);

\draw[->] (encoder.east) -- node[above left] {\(Y_{R_1}^n\)} (relay1.west);
\draw[->] (encoder.east) -- node[below left] {\(Y_{R_2}^n\)} (relay2.west);

\draw[->] (relay1.east) -- node[above] {\(X^n_1\)} (mac.west |- relay1.east);
\draw[->] (relay2.east) -- node[above] {\(X^n_2\)} (mac.west |- relay2.east);

\draw[->] (mac.east) -- node[above] {\(Y^n\)} (decoder.west);
\draw[->] (decoder.east) -- node[above] {\(\hat{M}\)} (sink.west);

\end{tikzpicture}
\caption{The diamond channel.}
\label{fig:diamond-channel}
\end{figure}

\color{black}
We apply this bound to the diamond channel depicted in Fig.~\ref{fig:diamond-channel}. The channel is a two-hop relay network with one source node, two parallel relay nodes, and one destination node, formed by cascading a broadcast channel (BC) and a multiple access channel (MAC). The BC connects the source with the relays via two noiseless bit pipes of capacities \(C_1\) and \(C_2\), respectively. Let \(X_S\) be the source input, and \(Y_{R_1}\) and \(Y_{R_2}\) the relay outputs. Having two bit pipes means that we may write
\[
    X_S = (Y_{R_1}, Y_{R_2}).
\]
where the alphabet of \(Y_{R_1}\) and \(Y_{R_2}\) have cardinalities $2^{C_1}$ and $2^{C_2}$, respectively.
The MAC is characterized by a conditional distribution \(p_{Y|X_1,X_2}\), where \(X_1\) and \(X_2\) are the relay inputs and \(Y\) is the destination output. We are interested in maximizing the rate \(R\) of message $M$.

We may as well assume that \(X_{Si}\) is independent of \((X_{1}^{i},X_{2}^i)\). Now fix \(P_{Z|X_1,X_2,Y}\) on the MAC side, and let \(T_i = Z^{i-1}\). Observe that \(X_{Si}\) is independent of \((T_i, X_{1i},X_{2i},Y_i,Z_i)\). Take \(\mathcal{U}\) to be the set of the two relays, and \(\mathcal{V}\) to be the source node. Then \eqref{eqn-db-const-1} is satisfied because the source node has a constant output.  Thus, Lemma \ref{lemmar3} yields
\begin{align}
    0 &\leq \sum_i I(X_{1i}, Y_{R_1i}; X_{2i}, Y_{R_2i} | T_i, Z_i) - I(X_{1i}; X_{2i} | T_i) \nonumber \\
      &\quad + \sum_i I(X_{Si}; Y_{R_1i}, Y_{R_2i}, Z_i | X_{1i}, X_{2i}, T_i) \nonumber \\
      &\quad - I(X_S^n; Y_{R_1}^n, Y_{R_2}^n, Z^n).
\end{align}
By Fano's inequality, we have
\begin{align}
    \frac{1}{n} I(X_S^n; Y_{R_1}^n, Y_{R_2}^n, Z^n) \geq R - \epsilon_n
\end{align}
where \(\epsilon_n \to 0\) as \(n \to \infty\). 

Moreover, using the independence of \(X_{Si}=(Y_{R_1i},Y_{R_2i})\) and \((T_i, X_{1i},X_{2i},Z_i)\), we obtain
\begin{align}
    I(X_{1i}, Y_{R_1i}; X_{2i}, Y_{R_2i} | T_i, Z_i) 
    &= I(X_{1i}; X_{2i} | T_i, Z_i) + I(Y_{R_1i}; Y_{R_2i})
\end{align}
and
\begin{align}
    I(X_{Si}; Y_{R_1i}, Y_{R_2i}, Z_i | X_{1i}, X_{2i}, T_i) 
    &= H(Y_{R_1i}, Y_{R_2i}) \nonumber \\
    &= H(Y_{R_1i}) + H(Y_{R_2i}) - I(Y_{R_1i}; Y_{R_2i}) \nonumber \\
    &\le C_1 + C_2 - I(Y_{R_1i}; Y_{R_2i}).
\end{align}
Summarizing, we obtain
\begin{align}
    0 &\le \sum_i \left[ I(X_{1i}; X_{2i} | T_i, Z_i) - I(X_{1i}; X_{2i} | T_i) \right] + n(C_1 + C_2 - R).
\end{align}
This recovers a bound in \cite[Theorem 3]{bidokhti2016capacity}.

\color{black}
\section{Conclusion and Future Work}
\label{sec:conclusions}
 
We developed a unified framework that leverages $\lambda$-multivariate information and auxiliary receivers to derive general dependence-balance (DB) constraints for multiterminal networks. The DB bounds strengthen outer bounds for (i) secret key and common randomness generation, including wiretap models with public or secure feedback, and (ii) reliable communication, yielding improvements over classic cut-set bounds for several models.

The following open problems are of interest for future study.
\begin{itemize}
  \item New auxiliary designs: are there methods beyond those discussed in Section~\ref{subsec:auxiliary-discussion} (modifying inactive terminals and output enhancement) to obtain systematically stronger bounds?
  \item Better bounds for Gaussian networks and relays: Can our DB bounds be combined with the upper bounds in~\cite{gohari2021outer,el2022strengthened} to yield better converses for Gaussian relay channels?
\item Adaptive auxiliary receivers: Hekstra and Willems showed that adaptive parallel channels can yield stronger bounds \cite[Section~VI]{hekstra1989dependence}. Can one similarly strengthen the bounds in this paper?
\end{itemize}

\color{black}

\newpage
\appendix

\section{Properties of Fractional Partition Multivariate Information}
\label{sec:appendixA}

The following proposition follows from the arguments in \cite{csiszar2008secrecy}.

\begin{proposition}\label{propos1}
$\lambda$-multivariate information satisfies the following properties.
\begin{itemize}
\item (Non-negativity):
    $I_{\lambda}(X_1;X_2;\cdots;X_k)\geq 0$ with equality if the $X_1,\dots,X_k$ are mutually independent.
\item (Conditioning): We have
\begin{align}
    I_{\lambda}(X_1;X_2;\cdots;X_k) - I_{\lambda}(X_1;X_2;\cdots;X_k|T)\leq I(X_{[k]};T) .
\end{align}
\item (Data processing): If $p(x'_{[k]},x_{[k]})=p(x_{[k]})\prod\nolimits_{i=1}^k p(x'_i|x_i)$ then we have
\begin{align}
    I_{\lambda}(X_{1};X_{2};\cdots;X_{k})\geq I_{\lambda}(X'_{1};X'_{2};\cdots;X'_{k}).
\end{align}
\item (Chain rule): We have
\begin{align}
    & I_{\lambda}(X_{1}Y_{1};X_{2}Y_{2};\cdots;X_{k}Y_{k}) \nonumber \\
    & = I_{\lambda}(X_{1};X_{2};\cdots;X_{k}) + I_{\lambda}(Y_{1};Y_{2};\cdots;Y_{k} \big| X_{[k]})  + \sum\nolimits_{\mathcal B\subsetneq [k]}\lambda_{\mathcal B} I(X_{{\mathcal B}};Y_{{\mathcal B}^c} \big| X_{{\mathcal B}^c}).
\end{align}
\item (Concavity): $I_{\lambda}(X_1;X_2;\cdots;X_k)$ is concave in $p(x_k)$ for a fixed $p(x_{[k-1]}|x_k)$; see \cite[Lemma A.1]{csiszar2008secrecy} for a proof.
\end{itemize}
\end{proposition}

\begin{proof}
For non-negativity, we have
\begin{align}
    H(X_{[k]})
    & =\sum\nolimits_{i} \left( \sum\nolimits_{{\mathcal B}: i\in {\mathcal B}} \lambda_{\mathcal B} \right) H(X_i|X^{i-1}) \nonumber \\
    & \overset{(a)}{\geq} \sum\nolimits_{{\mathcal B}} \sum\nolimits_{i\in {\mathcal B}} \lambda_{\mathcal B}\, H(X_i|X_{[i-1]\cap \mathcal B}, X_{\mathcal B^c}) \nonumber \\
    & =\sum\nolimits_{{\mathcal B}} \lambda_{\mathcal B} \, H(X_{\mathcal B}|X_{{\mathcal B}^c})
\end{align}
with equality in step $(a)$ if the $X_1,\dots,X_k$ are mutually independent. We remark that one can have $I_{\lambda}(X_1;X_2;\cdots;X_k)=0$ without mutual independence; an example is $k=3$ with $\lambda_{\{1,2\}}=\lambda_{\{3\}}=1$ and where $X_1=X_2$ is independent of $X_3$.

The conditioning inequality follows from the identity
\begin{align}
    I_{\lambda}(X_1;X_2;\cdots;X_k)-I_{\lambda}(X_1;X_2;\cdots;X_k|T)
    = I(X_{[k]};T) -\sum\nolimits_{\mathcal B} \lambda_{\mathcal B} I(X_{\mathcal B} ; T | X_{{\mathcal B}^c}).
\end{align}

The data processing inequality follows from functional representation: one can find variables $Y_{[k]}$, mutually independent of each other and $X_{[k]}$, such that $H(X'_i|X_i,Y_i)=0$. Since adding private noise $Y_i$ to $X_i$ does not change the $\lambda$-multivariate information, we have
\begin{align}
    I_{\lambda}(X_{1};X_{2};\cdots;X_{k})=
    I_{\lambda}(X_{1}Y_{1};X_{2}Y_{2};\cdots;X_{k}Y_{k})
\end{align}
and it suffices to show
\begin{align}
    I_{\lambda}(X'_{1}X_1Y_{1};X'_{2}X_2Y_{2};\cdots;X'_{k}X_kY_{k})\geq I_{\lambda}(X'_{1};X'_{2};\cdots;X'_{k}).
\end{align}
This inequality follows from
\begin{align}
    & I_{\lambda}(X'_{1}X_1Y_{1};X'_{2}X_2Y_{2}; \cdots; X'_{k}X_kY_{k}|T) - I_{\lambda}(X'_{1};X'_{2};\cdots;X'_{k}) \nonumber \\
    & = I_{\lambda}(X_1Y_1;X_2Y_2;\cdots;X_kY_k|X'_{[k]}) + \sum\nolimits_{\mathcal B} \lambda_{\mathcal B} I(X'_{\mathcal B} ; X_{\mathcal{B}^c} Y_{\mathcal{B}^c} | X'_{{\mathcal B}^c}).
\end{align}

The chain rule follows by
\begin{align}
    & I_{\lambda}(X_{1}Y_{1};X_{2}Y_{2};\cdots;X_{k}Y_{k})
    \nonumber \\
    & = \left(1-\sum\nolimits_{\mathcal B\subsetneq [k]}\lambda_{\mathcal B}\right) H(X_{[k]} Y_{[k]}) + \sum\nolimits_{\mathcal B\subsetneq [k]}\lambda_{\mathcal B} H(X_{{\mathcal B}^c} Y_{{\mathcal B}^c}) \nonumber \\
    & = I_{\lambda}(X_{1};X_{2};\cdots;X_{k}) + \left(1-\sum\nolimits_{\mathcal B\subsetneq [k]}\lambda_{\mathcal B}\right) H(Y_{[k]} \big| X_{[k]}) + \sum\nolimits_{\mathcal B\subsetneq [k]}\lambda_{\mathcal B} H(Y_{{\mathcal B}^c} \big| X_{{\mathcal B}^c})
\end{align}
and by writing $H(Y_{{\mathcal B}^c} \big| X_{{\mathcal B}^c}) = H(Y_{{\mathcal B}^c} \big| X_{[k]}) + I(X_{{\mathcal B}};Y_{{\mathcal B}^c} \big| X_{{\mathcal B}^c})$.
\end{proof}

\subsection{Relation to Another Definition of Multivariate Information}
\label{subsec:appendixA1}

Several other types of multivariate information have been studied. For instance, the \emph{$K$-information} is defined in \cite{McGill-P54} as
\begin{align} \label{eq:K-information}
    K(X_1;X_2;\cdots;X_k) = \sum\nolimits_{i\in[k]} (-1)^{i-1} \sum\nolimits_{|\mathcal{B}|=i} H(X_{\mathcal{B}}).
\end{align}
This information measure is motivated by Venn diagrams and appears in \cite{fano61,han80,hekstra1989dependence}, for example.

Another multivariate information more closely related to $\lambda$-multivariate information is
\begin{align}
    J(X_1;X_2;\cdots;X_k) = -H(X_{[k]}) + \sum\nolimits_{i} H(X_i).
\end{align}
We can relate this \emph{$J$-information} to $\lambda$-multivariate information. Let $\lambda_{\mathcal B}=1/(k-1)$ if $|{\mathcal B}|=k-1$, and $\lambda_{\mathcal B}=0$ otherwise; see \eqref{eq:frac-example} and \eqref{eq:MACS-GF-DB}. We then have
\begin{align}
    I_{\lambda}(X_1;X_2;\cdots;X_k)
    & = H(X_{[k]}) - \frac{1}{k-1} \sum\nolimits_{i}H(X_{[k]-i}|X_i) \nonumber \\
    & = \frac{1}{k-1}\left( -H(X_{[k]}) + \sum\nolimits_{i}H(X_i) \right) \nonumber \\
    & = \frac{1}{k-1}J(X_1;X_2;\cdots;X_k).
\end{align}
Another interesting relation is as follows. Let $\Pi=(\mathcal{P}_1, \mathcal{P}_2, \cdots, \mathcal{P}_r)$ be a partition of $[k]$ into $r\geq 2$ sets. Let $\lambda_{\mathcal B}=\frac{1}{r-1}$ if $\mathcal B=[k]-\mathcal{P}_i$ for some $i\in[r]$, and $\lambda_{\mathcal B}=0$ otherwise. We have
\begin{align}
I_{\lambda}(X_1;X_2;\cdots;X_k)
= \frac{1}{r-1} J(X_{\mathcal{P}_1}; X_{\mathcal{P}_2}; \cdots; X_{\mathcal{P}_r}).
\end{align}
Consequently, we have
\begin{align}
   \min_{\lambda} I_{\lambda}(X_1;X_2;\cdots;X_k)
   & \leq \min_{\Pi} \frac{1}{r-1} J(X_{\mathcal{P}_1}; X_{\mathcal{P}_2}; \cdots; X_{\mathcal{P}_r})
   \label{eqnnfn2}
\end{align}
where the minimum is over all $r\geq 2$ and over all partitions $\Pi=(\mathcal{P}_1, \mathcal{P}_2, \cdots, \mathcal{P}_r)$  of $[k]$ into $r$ sets. 

The following theorem complements the above example by showing that \eqref{eqnnfn2} holds with equality.
\begin{theorem}\cite[Theorem 4.1]{chan2015multivariate} For any $X_1, X_2, \cdots, X_k$, we have
\begin{align}
     \min_{\lambda}I_{\lambda}(X_1;X_2;\cdots;X_k)= \min_{\Pi}\frac{1}{r-1}J(X_{\mathcal{P}_1};X_{\mathcal{P}_2};\cdots;X_{\mathcal{P}_r})
\end{align}
where the minimum is over all $r\geq 2$ and over all partitions $\Pi=(\mathcal{P}_1, \mathcal{P}_2, \cdots, \mathcal{P}_r)$  of $[k]$ into $r$ sets. 
\end{theorem}

\subsection{Proof of \eqref{neqtp2} using the definition~\eqref{eq:defMI-2}}
\label{appendixNew2}
{\color{black}
In this section, we complete the argument in Section \ref{subsubsec:continuous-rvs} using Definition~\eqref{eq:defMI-2} of multivariate information to establish \eqref{neqtp2}. Definition~\eqref{eq:defMI-2} extends to continuous or mixed random variables.

Using the chain rule for $\lambda$-information, we have (see Lemma \ref{lemmaad2} below)
\begin{align}
    & I_{\lambda}(X_{1}Y_{1};X_{2}Y_{2};\cdots;X_{k}Y_{k}|T) \nonumber \\
    & = I_{\lambda}(X_{1};X_{2};\cdots;X_{k}|T) + I_{\lambda}(Y_{1};Y_{2};\cdots;Y_{k} \big| X_{[k]}T)  + \sum\nolimits_{\mathcal B\subsetneq [k]}\lambda_{\mathcal B} I(X_{{\mathcal B}};Y_{{\mathcal B}^c} \big| X_{{\mathcal B}^c}T).
    \label{eq:chain-rule}
\end{align}
We thus have
\begin{align}
    &\sum\nolimits_{j \in [n]} \bigg[ I_{\lambda}(W_{1}Y_{1}^jX_{1j} ; W_{2}Y_{2}^jX_{2j} ; \cdots ; W_{u}Y_{u}^jX_{uj} \big| Z^j )\nonumber \\
    & \qquad\qquad \qquad\color{black}- I_{\lambda}(W_{1}Y_{1}^{j-1}X_{1j} ; W_{2}Y_{2}^{j-1}X_{2j} ; \cdots; W_{u}Y_{u}^{j-1}X_{uj}\big| Z^{j-1}) \bigg] \nonumber
    \\
    &=\sum\nolimits_{j \in [n]} \bigg[ I_{\lambda}(Y_{1}^jX_{1j} ; Y_{2}^jX_{2j} ; \cdots ;Y_{u}^jX_{uj} \big| Z^j )\nonumber \\
    & \qquad\qquad \qquad\color{black}- I_{\lambda}(Y_{1}^{j-1}X_{1j} ; Y_{2}^{j-1}X_{2j} ; \cdots; Y_{u}^{j-1}X_{uj}\big| Z^{j-1}) \bigg] \nonumber
    \\
    &+\sum\nolimits_{j \in [n]} \bigg[ I_{\lambda}(W_1 ; W_2; \cdots ;W_u\big| Z^j Y_{[u]}^jX_{[u]j} )\nonumber \\
    & \qquad\qquad \qquad\color{black}- I_{\lambda}(W_1 ;W_2 ; \cdots; W_u\big| Z^{j-1}Y_{[u]}^{j-1}X_{[u]j}) \bigg] \nonumber
     \\
    &+\sum\nolimits_{j \in [n]} \bigg[ 
    \sum\nolimits_{\mathcal B\subsetneq [u]}\lambda_{\mathcal B} \left\{I(
    Y_{{\mathcal B}}^{j}X_{{\mathcal B}j}
    ;W_{{\mathcal B}^c} \big| Y_{{\mathcal B^c}}^{j}X_{{\mathcal B^c}j}Z^{j})-I(
    Y_{{\mathcal B}}^{j-1}X_{{\mathcal B}j}
    ;W_{{\mathcal B}^c} \big| Y_{{\mathcal B^c}}^{j-1}X_{{\mathcal B^c}j}Z^{j-1})\right\}
    \bigg].
\end{align}
It suffices to prove the following two equalities:
\begin{align}
    &\sum\nolimits_{j \in [n]} \bigg[ I_{\lambda}(Y_{1}^jX_{1j} ; Y_{2}^jX_{2j} ; \cdots ;Y_{u}^jX_{uj} \big| Z^j )\nonumber \\
    & \qquad\qquad \qquad\color{black}- I_{\lambda}(Y_{1}^{j-1}X_{1j} ; Y_{2}^{j-1}X_{2j} ; \cdots; Y_{u}^{j-1}X_{uj}\big| Z^{j-1}) \bigg] \nonumber
     \\&=\sum_{j\in[n]} \bigg[I_{\lambda}(X_{1j}Y_{1j};X_{2j}Y_{2j};\cdots;X_{uj}Y_{uj}|Z^{j-1},Z_j) - I_{\lambda}(X_{1j};X_{2j};\cdots;X_{uj}|Z^{j-1}) \nonumber \\
    &\qquad \qquad\color{black} - \left(1-\sum\nolimits_{\mathcal B\subsetneq [u]}\lambda_{\mathcal B}\right)  I(Y_{[u]}^{j-1} ; Z_jY_{[u]j} \big| Z^{j-1} X_{[u]j} )
    \nonumber \\
    &\qquad \qquad\color{black}-\sum\nolimits_{\mathcal B} \lambda_{\mathcal B}I(Y_{{\mathcal B}^c}^{j-1} ; Z_j Y_{{\mathcal B}^cj} \big| Z^{j-1} X_{\mathcal{B}^cj})\bigg] 
    \label{eqnttt1}
\end{align}
and
\begin{align}
    &\sum\nolimits_{j \in [n]} \bigg[ I_{\lambda}(W_1 ; W_2; \cdots ;W_u\big| Z^j Y_{[u]}^jX_{[u]j} )\nonumber \\
    & \qquad\qquad \qquad\color{black}- I_{\lambda}(W_1 ;W_2 ; \cdots; W_u\big| Z^{j-1}Y_{[u]}^{j-1}X_{[u]j}) \bigg] \nonumber
     \\
    &+\sum\nolimits_{j \in [n]} \bigg[ 
    \sum\nolimits_{\mathcal B\subsetneq [u]}\lambda_{\mathcal B} \left\{I(
    Y_{{\mathcal B}}^{j}X_{{\mathcal B}j}
    ;W_{{\mathcal B}^c} \big| Y_{{\mathcal B^c}}^{j}X_{{\mathcal B^c}j}Z^{j})-I(
    Y_{{\mathcal B}}^{j-1}X_{{\mathcal B}j}
    ;W_{{\mathcal B}^c} \big| Y_{{\mathcal B^c}}^{j-1}X_{{\mathcal B^c}j}Z^{j-1})\right\}
    \bigg] \nonumber
    \\&=\sum_{j\in[n]} \bigg[- \left(1-\sum\nolimits_{\mathcal B\subsetneq [u]}\lambda_{\mathcal B}\right)  I(W_{[u]}  ; Z_jY_{[u]j} \big| Z^{j-1} X_{[u]j}Y_{[u]}^{j-1} )
    \nonumber \\
    &\qquad \qquad\color{black}-\sum\nolimits_{\mathcal B} \lambda_{\mathcal B}I(W_{{\mathcal B}^c} ; Z_j Y_{{\mathcal B}^cj} \big| Z^{j-1} X_{\mathcal{B}^cj}Y_{{\mathcal B}^c}^{j-1})\bigg].\label{eqnttt2}
\end{align}

To prove \eqref{eqnttt1}, we use the chain rule \eqref{eq:chain-rule} to obtain
\begin{align}
    &I_{\lambda}(Y_{1}^jX_{1j} ; Y_{2}^jX_{2j} ; \cdots ;Y_{u}^jX_{uj} \big| Z^j )- I_{\lambda}(Y_{1}^{j-1}X_{1j} ; Y_{2}^{j-1}X_{2j} ; \cdots; Y_{u}^{j-1}X_{uj}\big| Z^{j-1}) 
    \nonumber \\
    &=I_{\lambda}(X_{1j}Y_{1j} ; X_{2j}Y_{2j} ; \cdots ;X_{uj}Y_{uj} \big| Z^j )- I_{\lambda}(X_{1j} ; X_{2j} ; \cdots; X_{uj}\big| Z^{j-1})
    \nonumber \\
    &\qquad+
    I_{\lambda}(Y_{1}^{j-1} ; Y_{2}^{j-1} ; \cdots ;Y_{u}^{j-1}\big|X_{[u]j}Y_{[u]j} Z^j )- I_{\lambda}(Y_{1}^{j-1} ; Y_{2}^{j-1} ; \cdots; Y_{u}^{j-1}\big|X_{[u]j} Z^{j-1}) 
    \nonumber \\
    &\qquad+\sum\nolimits_{\mathcal B\subsetneq [u]}\lambda_{\mathcal B} \left\{I(X_{{\mathcal B}j}Y_{{\mathcal B}j};Y_{{\mathcal B}^c}^{j-1} \big| X_{{\mathcal B}^cj}Y_{{\mathcal B}^cj}Z^j)-I(X_{{\mathcal B}j};Y_{{\mathcal B}^c}^{j-1} \big| X_{{\mathcal B}^cj}Z^{j-1})\right\}
        \nonumber \\
    &=I_{\lambda}(X_{1j}Y_{1j} ; X_{2j}Y_{2j} ; \cdots ;X_{uj}Y_{uj} \big| Z^j )- I_{\lambda}(X_{1j} ; X_{2j} ; \cdots; X_{uj}\big| Z^{j-1}) \nonumber \\
    &\qquad+
    I_{\lambda}(Y_{1}^{j-1} ; Y_{2}^{j-1} ; \cdots ;Y_{u}^{j-1}\big|X_{[u]j}Y_{[u]j} Z^j )- I_{\lambda}(Y_{1}^{j-1} ; Y_{2}^{j-1} ; \cdots; Y_{u}^{j-1}\big|X_{[u]j} Z^{j-1}) 
    \nonumber \\
    &\qquad+\sum\nolimits_{\mathcal B\subsetneq [u]}\lambda_{\mathcal B} \left\{I(X_{{\mathcal B}j}Y_{[u]j}Z_j;Y_{{\mathcal B}^c}^{j-1} \big| X_{{\mathcal B}^cj}Z^{j-1})-I(X_{{\mathcal B}j};Y_{{\mathcal B}^c}^{j-1} \big| X_{{\mathcal B}^cj}Z^{j-1})\right\}
    \nonumber \\
    &\qquad-\sum\nolimits_{\mathcal B\subsetneq [u]}\lambda_{\mathcal B}I(Y_{{\mathcal B}^c}^{j-1} ; Z_j Y_{{\mathcal B}^cj} \big| Z^{j-1} X_{\mathcal{B}^cj})
    \nonumber \\
    &=I_{\lambda}(X_{1j}Y_{1j} ; X_{2j}Y_{2j} ; \cdots ;X_{uj}Y_{uj} \big| Z^j )- I_{\lambda}(X_{1j} ; X_{2j} ; \cdots; X_{uj}\big| Z^{j-1})
    \nonumber \\
    &\qquad+
    I_{\lambda}(Y_{1}^{j-1} ; Y_{2}^{j-1} ; \cdots ;Y_{u}^{j-1}\big|X_{[u]j}Y_{[u]j} Z^j )- I_{\lambda}(Y_{1}^{j-1} ; Y_{2}^{j-1} ; \cdots; Y_{u}^{j-1}\big|X_{[u]j} Z^{j-1}) 
    \nonumber \\
    &\qquad+\sum\nolimits_{\mathcal B\subsetneq [u]}\lambda_{\mathcal B} \left\{I(Y_{[u]j}Z_j;Y_{{\mathcal B}^c}^{j-1} \big| X_{[u]j}Z^{j-1})-I(Y_{{\mathcal B}^c}^{j-1} ; Z_j Y_{{\mathcal B}^cj} \big| Z^{j-1} X_{\mathcal{B}^cj})\right\}.
\end{align}
Thus, to prove \eqref{eqnttt1}, it suffices to verify that
\begin{align}
    & I_{\lambda}(Y_{1}^{j-1} ; Y_{2}^{j-1} ; \cdots ;Y_{u}^{j-1}\big|X_{[u]j}Y_{[u]j} Z^j )- I_{\lambda}(Y_{1}^{j-1} ; Y_{2}^{j-1} ; \cdots; Y_{u}^{j-1}\big|X_{[u]j} Z^{j-1})  \nonumber \\
    &\qquad+\sum\nolimits_{\mathcal B\subsetneq [u]}\lambda_{\mathcal B} I(Y_{[u]j}Z_j;Y_{{\mathcal B}^c}^{j-1} \big| X_{[u]j}Z^{j-1})
     \nonumber \\
     &=- \left(1-\sum\nolimits_{\mathcal B\subsetneq [u]}\lambda_{\mathcal B}\right)  I(Y_{[u]}^{j-1} ; Z_jY_{[u]j} \big| Z^{j-1} X_{[u]j} ).
\end{align}
Let $S=Z_jY_{[u]j}$ and $T=Z^{j-1} X_{[u]j}$, and $G_i=Y_{i}^{j-1}$. We would like to show
\begin{align}
    & I_{\lambda}(G_1; G_2 ; \cdots ;G_u\big|ST )- I_{\lambda}(G_1 ; G_2 ; \cdots; G_u\big|T)  +\sum\nolimits_{\mathcal B\subsetneq [u]}\lambda_{\mathcal B} I(S;G_{{\mathcal B}^c} \big| T)
    \nonumber \\
    &=- \left(1-\sum\nolimits_{\mathcal B\subsetneq [u]}\lambda_{\mathcal B}\right)  I(G_{[u]}; S \big| T ).
\end{align}
This follows from part two of Lemma \ref{lemmaad2}.

To prove \eqref{eqnttt2}, observe that
\begin{align}
    &I(
    Y_{{\mathcal B}}^{j}X_{{\mathcal B}j}
    ;W_{{\mathcal B}^c} \big| Y_{{\mathcal B^c}}^{j}X_{{\mathcal B^c}j}Z^{j})-I(
    Y_{{\mathcal B}}^{j-1}X_{{\mathcal B}j}
    ;W_{{\mathcal B}^c} \big| Y_{{\mathcal B^c}}^{j-1}X_{{\mathcal B^c}j}Z^{j-1}) \nonumber \\
    &= I(Z_j Y_{{\mathcal B}^cj}
    Y_{{\mathcal B}}^{j}X_{{\mathcal B}j}
    ;W_{{\mathcal B}^c} \big| Y_{{\mathcal B^c}}^{j-1}X_{{\mathcal B^c}j}Z^{j-1})-I(
    Y_{{\mathcal B}}^{j-1}X_{{\mathcal B}j}
    ;W_{{\mathcal B}^c} \big| Y_{{\mathcal B^c}}^{j-1}X_{{\mathcal B^c}j}Z^{j-1}) 
    \nonumber \\ &
    \qquad-I(W_{{\mathcal B}^c} ; Z_j Y_{{\mathcal B}^cj} \big| Z^{j-1} X_{\mathcal{B}^cj}Y_{{\mathcal B}^c}^{j-1})
    \nonumber \\
    &= I(Z_j Y_{[u]j}
    ;W_{{\mathcal B}^c} \big| Y_{{[u]}}^{j-1}Z^{j-1}X_{[u]j})
    -I(W_{{\mathcal B}^c} ; Z_j Y_{{\mathcal B}^cj} \big| Z^{j-1} X_{\mathcal{B}^cj}Y_{{\mathcal B}^c}^{j-1}).
\end{align}
For the final step, it suffices to show that
\begin{align}
    &I_{\lambda}(W_1 ; W_2; \cdots ;W_u\big| Z^j Y_{[u]}^jX_{[u]j} )- I_{\lambda}(W_1 ;W_2 ; \cdots; W_u\big| Z^{j-1}Y_{[u]}^{j-1}X_{[u]j}) \nonumber \\
    & + \sum\nolimits_{\mathcal B\subsetneq [u]}\lambda_{\mathcal B} I(Z_j Y_{[u]j}
    ;W_{{\mathcal B}^c} \big| Y_{{[u]}}^{j-1}Z^{j-1}X_{[u]j})
    \nonumber \\
    &=- \left(1-\sum\nolimits_{\mathcal B\subsetneq [u]}\lambda_{\mathcal B}\right)  I(W_{[u]}  ; Z_jY_{[u]j} \big| Z^{j-1} X_{[u]j}Y_{[u]}^{j-1} ).
\end{align}
Let $S=Z_jY_{[u]j}$ and $T=Z^{j-1} X_{[u]j}Y_{[u]}^{j-1}$. We would like to show
\begin{align}
    &I_{\lambda}(W_1 ; W_2; \cdots ;W_u\big| ST)- I_{\lambda}(W_1 ;W_2 ; \cdots; W_u\big| T) 
     +
    \sum\nolimits_{\mathcal B\subsetneq [u]}\lambda_{\mathcal B} I(S
    ;W_{{\mathcal B}^c} \big| T)
    \nonumber \\
    &=- \left(1-\sum\nolimits_{\mathcal B\subsetneq [u]}\lambda_{\mathcal B}\right)  I(W_{[u]}  ; S \big|T ).
\end{align}
This follows from part two of Lemma \ref{lemmaad2}.

\begin{lemma}\label{lemmaad2}
    The following two identities hold for the definition of $I_\lambda$ given in~\eqref{eq:defMI-2}:
    \begin{itemize}
        \item The chain rule:
        \begin{align}
    & I_{\lambda}(X_{1}Y_{1};X_{2}Y_{2};\cdots;X_{k}Y_{k}|T) \nonumber \\
    & = I_{\lambda}(X_{1};X_{2};\cdots;X_{k}|T) + I_{\lambda}(Y_{1};Y_{2};\cdots;Y_{k} \big| X_{[k]}T)  + \sum\nolimits_{\mathcal B\subsetneq [k]}\lambda_{\mathcal B} I(X_{{\mathcal B}};Y_{{\mathcal B}^c} \big| X_{{\mathcal B}^c}T).
\end{align}
        \item The identity:
 \begin{align}
    &I_{\lambda}(X_1 ; X_2; \cdots ;X_k\big| ST)- I_{\lambda}(X_1 ;X_2 ; \cdots; X_k\big| T) 
    \nonumber \\
    &=-
    I(X_{[k]}  ; S \big|T )+\sum\nolimits_{\mathcal B\subsetneq [k]}\lambda_{\mathcal B} I(X_{[k]}  ; S \big|X_{{\mathcal B}^c}T )
    \\&=- \left(1-\sum\nolimits_{\mathcal B\subsetneq [k]}\lambda_{\mathcal B}\right)  I(X_{[k]}  ; S \big|T )-
    \sum\nolimits_{\mathcal B\subsetneq [k]}\lambda_{\mathcal B} I(S
    ;X_{{\mathcal B}^c} \big| T).
\end{align}
\end{itemize}

\end{lemma}
\begin{proof}
    To prove the chain rule, we need to show
\begin{align}
    &\sum\nolimits_{{\mathcal B\subsetneq [k]}} \sum\nolimits_{i\in {\mathcal B}} \lambda_{\mathcal B}\,
    I(X_iY_i;X_{\mathcal B^c}Y_{\mathcal B^c} | TX^{i-1}Y^{i-1}) \\
    & = \sum\nolimits_{{\mathcal B\subsetneq [k]}} \sum\nolimits_{i\in {\mathcal B}} \lambda_{\mathcal B}\,
    \left\{I(X_i;X_{\mathcal B^c} | TX^{i-1}) + 
    I(Y_i;Y_{\mathcal B^c} | X_{[k]}T Y^{i-1})\right\}
    \nonumber \\
    & \qquad + \sum\nolimits_{\mathcal B\subsetneq [k]}\lambda_{\mathcal B} I(X_{{\mathcal B}};Y_{{\mathcal B}^c} \big| X_{{\mathcal B}^c}T).
\end{align}
We have
\begin{align}
     & I_{\lambda}(X_{1}Y_{1};X_{2}Y_{2};\cdots;X_{k}Y_{k}|T)
     \nonumber \\
     &=\sum\nolimits_{{\mathcal B\subsetneq [k]}} \sum\nolimits_{i\in {\mathcal B}} \lambda_{\mathcal B}\,
    I(X_iY_i;X_{\mathcal B^c}Y_{\mathcal B^c} | TX^{i-1}Y^{i-1}) \nonumber \\
    & = \sum\nolimits_{{\mathcal B\subsetneq [k]}} \sum\nolimits_{i\in {\mathcal B}} \lambda_{\mathcal B}\,
    \left\{I(X_iY_i;X_{\mathcal B^c}Y_{\mathcal B^c} Y^{i-1}| TX^{i-1})-I(X_iY_i; Y^{i-1}| TX^{i-1})\right\}
    \nonumber \\
    & = \sum\nolimits_{{\mathcal B\subsetneq [k]}} \sum\nolimits_{i\in {\mathcal B}} \lambda_{\mathcal B}\,
    \left\{I(X_i;X_{\mathcal B^c}| TX^{i-1})+I(X_i;Y_{\mathcal B^c} Y^{i-1}| TX_{\mathcal B^c}X^{i-1})+I(Y_i;X_{\mathcal B^c}Y_{\mathcal B^c} Y^{i-1}| TX^{i})\right\}
    \nonumber \\
    &\qquad-\sum\nolimits_{{\mathcal B\subsetneq [k]}} \sum\nolimits_{i\in {\mathcal B}} \lambda_{\mathcal B}I(X_iY_i; Y^{i-1}| TX^{i-1})
    \nonumber \\
    & = \sum\nolimits_{{\mathcal B\subsetneq [k]}} \sum\nolimits_{i\in {\mathcal B}} \lambda_{\mathcal B}\,
    \left\{I(X_i;X_{\mathcal B^c}| TX^{i-1})+I(X_i;Y_{\mathcal B^c} Y^{i-1}| TX_{\mathcal B^c}X^{i-1})+
    I(Y_i;Y_{\mathcal B^c} Y^{i-1}| TX_{[k]})\right\}
    \nonumber \\
    &\qquad
    -\sum\nolimits_{{\mathcal B\subsetneq [k]}} \sum\nolimits_{i\in {\mathcal B}} \lambda_{\mathcal B}\,
    \left\{I(Y_i;X_{[k]}| TX^{i}X_{\mathcal B^c}Y_{\mathcal B^c} Y^{i-1})\right\}
    \nonumber \\
    &\qquad+\sum\nolimits_{{\mathcal B\subsetneq [k]}} \sum\nolimits_{i\in {\mathcal B}} \lambda_{\mathcal B}\left\{
    I(Y_i;X_{[k]}| TX^{i})-I(X_iY_i; Y^{i-1}| TX^{i-1})\right\}
    \nonumber \\
    & = I_{\lambda}(X_{1};X_{2};\cdots;X_{k}|T) + I_{\lambda}(Y_{1};Y_{2};\cdots;Y_{k} \big| X_{[k]}T)  \nonumber \\
    & \qquad+ \sum\nolimits_{{\mathcal B\subsetneq [k]}} \sum\nolimits_{i\in {\mathcal B}} \lambda_{\mathcal B}\,
    \left\{I(Y_i;Y^{i-1}|X_{[k]}T)+I(X_i;Y_{\mathcal B^c} Y^{i-1}| TX_{\mathcal B^c}X^{i-1})\right\}
    \nonumber \\
    &\qquad
    -\sum\nolimits_{{\mathcal B\subsetneq [k]}} \sum\nolimits_{i\in {\mathcal B}} \lambda_{\mathcal B}\,
    \left\{I(Y_i;X_{[k]}| TX^{i}X_{\mathcal B^c}Y_{\mathcal B^c} Y^{i-1})\right\}
    \nonumber \\
    &\qquad+\sum\nolimits_{{\mathcal B\subsetneq [k]}} \sum\nolimits_{i\in {\mathcal B}} \lambda_{\mathcal B}\left\{
    I(Y_i;X_{[k]}| TX^{i})-I(X_iY_i; Y^{i-1}| TX^{i-1})\right\}.
\end{align}
Observe that
\begin{align}
    & \sum\nolimits_{{\mathcal B\subsetneq [k]}} \sum\nolimits_{i\in\mathcal{B}}\lambda_{\mathcal B}\,
    \left\{I(X_i;Y_{\mathcal B^c} Y^{i-1}| TX_{\mathcal B^c}X^{i-1})\right\}
    \nonumber \\
    &= \sum\nolimits_{{\mathcal B\subsetneq [k]}} \sum\nolimits_{i\in\mathcal{B}}\lambda_{\mathcal B}\,
    \left\{I(X_i;Y_{\mathcal B^c}| TX_{\mathcal B^c}X^{i-1})+I(X_i;Y^{i-1}| TX_{\mathcal B^c}X^{i-1}Y_{\mathcal B^c} )\right\}
    \nonumber \\
    &= \sum\nolimits_{{\mathcal B\subsetneq [k]}} \sum\nolimits_{i=1}^k \lambda_{\mathcal B}\,
    I(X_i;Y_{\mathcal B^c}| TX_{\mathcal B^c}X^{i-1})+\sum\nolimits_{{\mathcal B\subsetneq [k]}} \sum\nolimits_{i\in\mathcal{B}} \lambda_{\mathcal B}\,I(X_i;Y^{i-1}| TX_{\mathcal B^c}X^{i-1}Y_{\mathcal B^c} )
    \nonumber \\
    &= \sum\nolimits_{{\mathcal B\subsetneq [k]}}  \lambda_{\mathcal B}\,
    I(X_{[k]};Y_{\mathcal B^c}| TX_{\mathcal B^c})+\sum\nolimits_{{\mathcal B\subsetneq [k]}} \sum\nolimits_{i\in\mathcal{B}} \lambda_{\mathcal B}\,
    I(X_i;Y^{i-1}| TX_{\mathcal B^c}X^{i-1}Y_{\mathcal B^c} )
    \nonumber \\
    &= \sum\nolimits_{{\mathcal B\subsetneq [k]}}  \lambda_{\mathcal B}\,
    I(X_{\mathcal{B}};Y_{\mathcal B^c}| TX_{\mathcal B^c})+\sum\nolimits_{{\mathcal B\subsetneq [k]}} \sum\nolimits_{i\in\mathcal{B}}\lambda_{\mathcal B}\,
    I(X_i;Y^{i-1}| TX_{\mathcal B^c}X^{i-1}Y_{\mathcal B^c} ).
\end{align}
Thus, to complete the proof, it suffices to show that
\begin{align}
    & \sum\nolimits_{i\in\mathcal{B}} \Big\{ I(X_i;Y^{i-1}| TX_{\mathcal B^c}X^{i-1}Y_{\mathcal B^c} ) - I(Y_i;X_{[k]}| TX^{i}X_{\mathcal B^c}Y_{\mathcal B^c} Y^{i-1}) \Big\}=0, \qquad \forall \mathcal B\subsetneq [k]
    \label {eq:lemma2-step} \\
    &\sum\nolimits_{{\mathcal B\subsetneq [k]}} \sum\nolimits_{i\in\mathcal{B}} \lambda_{\mathcal B}\,\left\{
    I(Y_i;Y^{i-1}|X_{[k]}T)+I(Y_i;X_{[k]}| TX^{i})-I(X_iY_i; Y^{i-1}| TX^{i-1})\right\}=0.
\end{align}
Consider the latter equation.
Using the chain rule of mutual information, we can expand:
\begin{align}
    I(X_iY_i; Y^{i-1}| TX^{i-1}) &= I(X_i; Y^{i-1}| TX^{i-1}) + I(Y_i; Y^{i-1}| TX^{i}). \label{eq:chain1}
\end{align}
Next, we have the following identity:
\begin{align}
    I(Y_i; X_{[k]} | TX^i)  &= I(Y_i; Y^{i-1} | TX^i)+I(Y_i; X_{[k]} | TX^i Y^{i-1}) - I(Y_i; Y^{i-1} | TX_{[k]}). \label{eq:chain2}
\end{align}
Substituting \eqref{eq:chain1} and \eqref{eq:chain2} into the last sum of the original expression, we obtain:
\begin{align}
    &\sum\nolimits_{{\mathcal B\subsetneq [k]}} \sum\nolimits_{i\in\mathcal{B}}\lambda_{\mathcal B}\, \Big\{ I(Y_i;Y^{i-1}|X_{[k]}T)+I(Y_i;X_{[k]}| TX^{i}) - I(X_iY_i; Y^{i-1}| TX^{i-1}) \Big\}
    \nonumber \\
    &= \sum\nolimits_{{\mathcal B\subsetneq [k]}} \sum\nolimits_{i\in\mathcal{B}}\lambda_{\mathcal B}\, \Big\{ I(Y_i; X_{[k]} | TX^i Y^{i-1}) - I(X_i; Y^{i-1}| TX^{i-1}) \Big\}
    \nonumber \\
    &= \sum_{i=1}^k\Big\{ I(Y_i; X_{[k]} | TX^i Y^{i-1}) - I(X_i; Y^{i-1}| TX^{i-1}) \Big\}
    \nonumber \\
    &=0
\end{align}
where the last step follows because
\begin{align}
    \sum_{i=1}^kI(X_i; Y^{i-1}| TX^{i-1})=\sum_{(i,j):i>j}I(X_i; Y_j | TX^{i-1} Y^{j-1}) \\
    \sum_{i=1}^kI(Y_i; X_{[k]} | TX^i Y^{i-1})=\sum_{(i,j):i<j}I(Y_i; X_j | TX^{j-1} Y^{i-1}).
\end{align}
The proof of \eqref{eq:lemma2-step} is similar because
\begin{align*}
\sum\nolimits_{i\in\mathcal{B}} I(X_i;Y^{i-1}| TX_{\mathcal B^c}X^{i-1}Y_{\mathcal B^c} ) &=\sum_{(i\in\mathcal{B},j\in\mathcal{B}):i>j}I(X_i; Y_j | TX^{i-1} Y^{j-1}X_{\mathcal B^c}Y_{\mathcal B^c})
    \\&=\sum\nolimits_{i\in\mathcal{B}}I(Y_i;X_{[k]}| TX^{i}X_{\mathcal B^c}Y_{\mathcal B^c} Y^{i-1}).
\end{align*}

\textbf{Proof of the second property:} To show the second property, we need to show that
\begin{align}
    &\sum\nolimits_{{\mathcal B\subsetneq [k]}} \sum\nolimits_{i\in {\mathcal B}} \lambda_{\mathcal B}\,
    \left\{I(X_i;X_{\mathcal B^c}| STX^{i-1})-I(X_i;X_{\mathcal B^c}| TX^{i-1})\right\}
    \nonumber \\
     &=-
    I(X_{[k]}  ; S \big|T )+\sum\nolimits_{\mathcal B\subsetneq [k]}\lambda_{\mathcal B} I(X_{[k]}  ; S \big|X_{{\mathcal B}^c}T ).
\end{align}
Observe that
\begin{align}
    &\sum\nolimits_{{\mathcal B\subsetneq [k]}} \sum\nolimits_{i\in {\mathcal B}} \lambda_{\mathcal B}\,
    \left\{I(X_i;X_{\mathcal B^c}| STX^{i-1})-I(X_i;X_{\mathcal B^c}| TX^{i-1})\right\}
    \nonumber \\
    &=\sum\nolimits_{{\mathcal B\subsetneq [k]}} \sum\nolimits_{i\in {\mathcal B}} \lambda_{\mathcal B}\,
    \left\{I(X_i;S| X_{\mathcal B^c}TX^{i-1})-I(X_i;S| TX^{i-1})\right\} 
    \nonumber \\
    &=\sum\nolimits_{{\mathcal B\subsetneq [k]}} \sum\nolimits_{i\in {\mathcal B}} \lambda_{\mathcal B}\,
    I(X_i;S| X_{\mathcal B^c}TX^{i-1})
    -\sum_{i=1}^k\left(\sum\nolimits_{{\mathcal B\subsetneq [k]}: i\in {\mathcal B}} \lambda_{\mathcal B}\right)I(X_i;S| TX^{i-1}) 
    \nonumber \\
    &=\sum\nolimits_{{\mathcal B\subsetneq [k]}} \sum\nolimits_{i\in [k]} \lambda_{\mathcal B}\,
    I(X_i;S| X_{\mathcal B^c}TX^{i-1})-\sum_{i=1}^kI(X_i;S| TX^{i-1}) 
    \nonumber \\
    &=\left[\sum\nolimits_{{\mathcal B\subsetneq [k]}}  \lambda_{\mathcal B}\,
    I(X_{[k]};S| X_{\mathcal B^c}T)\right]-I(X_{[k]};S| T).
\end{align}
\end{proof}
}

\section{Source Model with Silent Nodes}
\label{appendixC}

Consider the $k$-terminal source model with silent nodes when $H(Z|Y_i)=0$ for $i=1,2,\cdots, k$ and where the first $u$ terminals use the public channel.
The paper \cite[Theorem 6]{gohari2010information1} showed the maximum value for 
$R_{[k]}$ is
\begin{align}
    H(Y_{[u]}|Z) - \min_{(r_1,r_2,\cdots, r_u) \in \mathscr{R}}\sum\nolimits_i r_i
\end{align}
where $\mathscr{R}$ is the set of tuples
$(r_1,r_2,\cdots, r_u)$ such that for any proper set $\mathcal B$ satisfying $\mathcal B\cap [u]\neq \emptyset$ we have
\begin{align}
    \sum\nolimits_{j\in \mathcal B\cap[u]}r_j
    \geq H(Y_{\mathcal B\cap [u]}|Y_{\mathcal B^c}Z).
\end{align}
If $\mathcal B\cap[u]\neq [u]$, it is best to include $[k]-[u]$ in $\mathcal B$.  Thus, in this case, for any 
$\mathcal{B}\subsetneq [u]$ we have
\begin{align}
    \sum\nolimits_{j\in \mathcal B}r_j
    \geq H(Y_{\mathcal B}|Y_{[u]-\mathcal B}Z).
\end{align}
For the case $\mathcal B\cap[u]= [u]$, we obtain the following bound
\begin{align}
    \sum\nolimits_{i\in [u]}r_i\geq H(Y_{[u]}|Y_{j}Z),
    \qquad\forall\, j\in[k]-[u].
\end{align}
By writing the dual of the above linear program, we obtain the expression:
\begin{align}
    R_{[k]} = \min\bigg(
    & H(Y_{[u]}|Z)-\sum_{\mathcal{B}\subsetneq[u]} \zeta_\mathcal{B} H(Y_{B\cap [u]}|Y_{[u]-\mathcal B}Z)
    - \sum_{j\in[k]-[u]} \zeta_{\{j\}} H(Y_{[u]}|Y_{j}Z)
    \bigg)
\end{align}
where the minimum is over non-negative $\zeta_\mathcal{B}:~\mathcal{B}\subsetneq[u]$ and $\zeta_{\{j\}}$ for $j>u$ satisfying
\begin{align}
    \sum_{\mathcal{B}:~i\in\mathcal{B}} \zeta_{\mathcal{B}} + \sum_{j>u}\zeta_{\{j\}}
    = 1, \qquad \forall\, i\in[u].
\end{align}

To obtain this bound from our general upper bound, choose 
\begin{align}
    \begin{array}{ll}
    \omega_{[u]} = \sum_{\mathcal{B}:\,i\in\mathcal{B}} \zeta_{\mathcal{B}} & \\[4pt]
    \omega_{[u]\cup\{j\}} = \zeta_{\{j\}}, & \forall\, j\in[k]-[u] \\[4pt]
    \omega_\mathcal{U} = 0, & \text{ otherwise.}
    \end{array}
\end{align}
For the set $[u]$, define
\begin{align}
    \lambda_{\mathcal{B}}^{[u]}
    = \frac{\zeta_\mathcal{B}}{1-\sum_{j>u}\zeta_{\{j\}}}, \qquad\forall\, \mathcal{B}\subsetneq[u].
\end{align}
For the set $[u]\cup\{j\}$ for $j>u$, define $\lambda_{\mathcal{B}}^{[u]\cup\{j\}}=1$
if $\mathcal{B}=[u]$ or $\mathcal{B}=\{j\}$ and $\lambda_{\mathcal{B}}^{[u]\cup\{j\}}=0$ for all the other sets $\mathcal{B}$.
This choice of $\omega_\mathcal{U}$ and $\lambda^{\mathcal{U}}_\mathcal{B}$ yields the desired bound if the auxiliary receiver is $T=Z$ for the main and parallel channels. Note that the parallel channel is $Y_1=Y_2=\cdots=Y_k=Z=X_{[u]}$ with $X_{u+1},\cdots, X_k$ being constants. The proof of $V_{\omega, \lambda^{\cdot}}(q_1(t,y_{[k]},z|x_{[k]}))\leq 0$ for the parallel channel is similar to the one discussed in Section \ref{subsubsec:kterminalssec}; the only extra step is to show that
\begin{align}
-\sum_{\mathcal{U}}\omega_\mathcal{U}\left(1-\sum_{\mathcal B\subsetneq \mathcal{U}}\lambda^{\mathcal{U}}_{\mathcal B}\right)I(X_{[k]};Y_{\mathcal{U}},Z|X_{\mathcal{U}})=0.\label{eqnNeweqnns}
\end{align}
Note that we have $[u]\subseteq\mathcal{U}$ for the sets $\mathcal{U}$ where $\omega_\mathcal{U}>0$. The terms $I(X_{[k]};Y_{\mathcal{U}},Z|X_{\mathcal{U}})$ vanish because $X_j$ is a constant for $j\notin[u]$.

\section{Cardinality Bounds for Theorem~\ref{thm2}}
\label{appendixD}

Consider the statement of Theorem~\ref{thm2}. Fix the distribution $p(x_{[k]},t_{[a]}|t_m)$ and vary 
$p(t_m)$. For a marginal distribution $q(t_m)$, we require
\begin{align}
    \sum\nolimits_{t_m} q(t_m)\, p(x_{[k]},t_{[a]} \big| t_m)
    = \sum\nolimits_{t_m} p(t_m)\, p(x_{[k]},t_{[a]} \big| t_m), \quad \forall\, x_{[k]},t_{[a]}.
\end{align}
The factorization \eqref{eq:p-factorization-2} ensures it suffices to impose the following condition for every $x_{[k]}$:
\begin{align}
    \sum\nolimits_{t_m} q(t_m)\, p(x_{[k]} \big | t_m)
    = \sum\nolimits_{t_m} p(t_m)\, p(x_{[k]} \big| t_m).
\end{align}
This yields $\prod_{i}|\mathcal{X}_i|$ equations.
The number of equations involving $T_m$ in \eqref{eqn50-2} is $2^{k}-1$. To preserve the values of these expressions under $q(t_m)$ and $p(t_m)$, one must impose $2^{k}-1$ linear equations. Finally, instead of imposing \eqref{eqnDBs1} for every fractional partition $\lambda$, it suffices (by the linearity of the equation in $\lambda$) to impose the constraints only for the vertices of the fractional partition polytope, i.e., vertices formed by $2^{k}-1$ tuples $\{\lambda_{\mathcal B}\}$ for $\mathcal B\in \mathsf{B}$, defined by the $2^{k}-1$ non-negativity constraints $\lambda_{\mathcal B}\geq 0$ and the $k$ equality constraints in \eqref{eqnEE2}. Every vertex corresponds to the intersection of $2^{k}-1$ hyperplanes, so the number of vertices is at most
\begin{equation}
    \binom{2^{k}-1+k}{2^{k}-1}.
\end{equation}
Thus, by imposing $
    \binom{2^{k}-1+k}{2^{k}-1}
$ linear equations on $q(t_m)$, we can ensure that the DB inequalities are satisfied under $q(t_m)$. The total number of linear equations imposed on $q(t_m)$ is 
\begin{equation}
   \prod\nolimits_{i\in[k]} |\mathcal{X}_i| + (2^k-1) + \binom{2^{k}-1+k}{2^{k}-1}.
\end{equation}
Next, we have the inequality constraints $q(t_m)\geq 0$ for all $t_m$. Consider the polytope formed by the equality and inequality constraints, and let $q(t_m)$ be a vertex of this polytope. Since every vertex must lie on $|\mathcal{T}_m|$ hyperplanes (defining the polytope), the vertex must satisfy at least
\begin{equation}
   |\mathcal{T}_m|-\left(\prod\nolimits_{i\in[k]} |\mathcal{X}_i| + (2^k-1) + \binom{2^{k}-1+k}{2^{k}-1}\right)
\end{equation}
inequalities of the form $q(t_m)\geq 0$ with equality. Thus, the number of non-zero entries of $q(t_m)$ will be at most the desired cardinality bound on $T_m$ given in the theorem statement.

\section{Optimality of Gaussian Inputs}
\label{appendixE}

Consider the channel \eqref{eq:Gaussian-channel} and the power constraints \eqref{eq:power-constraints}. The following lemma bounds the maximum weighted sum rate.

\begin{definition}\label{def6}
     Let $\mathcal{P}$ be the set of $p(x_{[k]},t_{[a]})$ factorizing as in \eqref{eq:p-factorization-2} and satisfying the DB constraints \eqref{eqnDBs1} and power constraints \eqref{eq:power-constraints}. Let $\mathcal{P}'$ be the set of $p(x_{[k]},t_{[a]})$ satisfying \eqref{eqnDBs1} and \eqref{eq:power-constraints}, but not necessarily factorizing as in \eqref{eq:p-factorization-2}.
\end{definition}
 
\begin{lemma}
    Let $\beta_{i\mathcal S}$ (for $i, \mathcal{S}\subseteq [k]-\{i\}$) be non-negative real numbers. The outer bound in Theorem \ref{thm2} can be equivalently expressed as follows. 
    Any achievable rate tuple $\{R_{i\mathcal{S}}\}$ satisfies

\begin{align}
    \sum\nolimits_{i,\mathcal{L}} \beta_{i\mathcal{L}} R_{i\mathcal{L}}  \le \min_{\gamma\in\mathcal{G}}\sup_{p(x_{[k]},t_{[a]})\in \mathcal{P}} \sum\nolimits_{m,\mathcal{S}} \gamma_{\mathcal{S},m} I(X_{\mathcal S};Z_m,Y_{{\mathcal S^c}}|X_{\mathcal S^c},T_m)
\end{align}
    for all $\{\beta_{i\mathcal S}\}$, where $\mathcal{P}$ is given by Definition \ref{def6} and
    $\mathcal{G}$ is the set of non-negative weights $\gamma_{\mathcal{S},m}$ for non-empty $\mathcal{S}\subsetneq [k]$, $m\in[a]$ satisfying
\begin{align}
    \beta_{i\mathcal{L}} = \sum_{m,\mathcal{S}:i\in \mathcal{S},\mathcal{L}\cap \mathcal{S}^c\neq \emptyset}\gamma_{\mathcal{S},m}.
\end{align}
\end{lemma}
\begin{proof}
 The proof of Theorem \ref{thm2} shows that taking union over $p(x_{[k]},t_{[a]})$ in $\mathcal{P}'$ yields the same region as taking union over $p(x_{[k]},t_{[a]})$ in $\mathcal{P}$ because all mutual information terms depend only on the marginals $p(x_{[k]},t_{m})$ for $m\in[a]$. 
 From \eqref{eqn50-2}, for any $\gamma_{\mathcal{S},m}\geq 0$ we have
\begin{align}
    \sum_{m,\mathcal{S}} \gamma_{\mathcal{S},m} \sum_{i\in \mathcal{S},\mathcal{L}\cap \mathcal{S}^c\neq \emptyset} R_{i\mathcal{L}} \leq \sum_{m,\mathcal{S}} \gamma_{\mathcal{S},m} I(X_{\mathcal S};Z_m,Y_{{\mathcal S^c}}|X_{\mathcal S^c},T_m).
\end{align}
For any $\gamma\in\mathcal{G}$, we have
\begin{align}
    \beta_{i\mathcal{L}} = \sum_{m,\mathcal{S}:i\in \mathcal{S},\mathcal{L}\cap \mathcal{S}^c\neq \emptyset}\gamma_{\mathcal{S},m}
\end{align}
so we obtain 
\begin{align}
   \sum_{i,\mathcal{L}} \beta_{i\mathcal{L}} R_{i\mathcal{L}} & \leq \sup_{p(x_{[k]},t_{[a]}) \in \mathcal{P}} \min_{\gamma\in\mathcal{G}} \sum_{m,\mathcal{S}}\gamma_{\mathcal{S},m} I(X_{\mathcal S};Z_m,Y_{{\mathcal S^c}}|X_{\mathcal S^c},T_m) 
   \nonumber \\
    & =\sup_{p(x_{[k]},t_{[a]})\in \mathcal{P}'} \min_{\gamma\in\mathcal{G}} \sum_{m,\mathcal{S}}\gamma_{\mathcal{S},m} I(X_{\mathcal S};Z_m,Y_{{\mathcal S^c}}|X_{\mathcal S^c},T_m)
    \nonumber \\
    & = \min_{\gamma\in\mathcal{G}} \sup_{p(x_{[k]},t_{[a]})\in \mathcal{P}'} \sum_{m,\mathcal{S}} \gamma_{\mathcal{S},m} I(X_{\mathcal S};Z_m,Y_{{\mathcal S^c}}|X_{\mathcal S^c},T_m)
    \nonumber \\
    & = \min_{\gamma\in\mathcal{G}} \sup_{p(x_{[k]},t_{[a]}) \in \mathcal{P}} \sum_{m,\mathcal{S}} \gamma_{\mathcal{S},m} I(X_{\mathcal S};Z_m,Y_{{\mathcal S^c}}|X_{\mathcal S^c},T_m)
\end{align}
where the minimax exchange follows from Corollary 2 in \cite{ggny14} and because the set of all tuples $(\tilde{R}_{m,\mathcal{S}})$ satisfying
\begin{align}
\tilde{R}_{m,\mathcal{S}}\leq I(X_{\mathcal S};Z_m,Y_{{\mathcal S^c}}|X_{\mathcal S^c},T_m)\label{tpregion}
\end{align}
over all $p(x_{[k]},t_{[a]})\in\mathcal{P}'$ is a convex region. The latter holds by including a time-sharing variable in the $T_m$'s as follows: take two tuples $(\tilde{R}^{(1)}_{m,\mathcal{S}})$ and $(\tilde{R}^{(2)}_{m,\mathcal{S}})$, and corresponding distributions 
$p_1(x_{[k]},t^{(1)}_{[a]})\in\mathcal{P}'$ and $p_2(x_{[k]},t^{(2)}_{[a]})\in\mathcal{P}'$. Let $Q\in\{1,2\}$ be a uniform random variable, independent of all previously defined random variables, and set $T'_m=(T^{(Q)}_m,Q)$ for all $t\in[a]$. Since all mutual information terms (including those in DB constraints) are conditioned on $T'_m$ for some $m$, every mutual information term will be conditioned on $Q$, and its value will be the average of those under $p_1(x_{[k]},t^{(1)}_{[a]})$ and $p_2(x_{[k]},t^{(2)}_{[a]})$. This will convexify the region based on \eqref{tpregion}.
\end{proof}

\begin{theorem}For any weights $\gamma_{\mathcal{S},m}\geq 0$, the supremum
\begin{align}
    \sup_{p(x_{[k]},t_{[a]})\in \mathcal{P}}\sum_{m,\mathcal{S}}\gamma_{\mathcal{S},m}I(X_{\mathcal S};Z_m,Y_{{\mathcal S^c}}|X_{\mathcal S^c},T_m)
\end{align}
is obtained by a jointly Gaussian distribution where $T_m$ is a $k$-dimensional random vector. Here, the set  $\mathcal{P}$ is defined in Definition \ref{def6}.
\end{theorem}
\begin{proof}
We perturb the objective function\footnote{This idea was first introduced in \cite{gohari2021outer}. For a non-trivial application of this idea, please see \cite{lau2024entropic}.} by adding a small term $\epsilon I(X_{[k]};\tilde Y_{[k]},Z_m|T_m)$. By continuity, it suffices to show the optimality of the Gaussian input distribution for
\begin{align}
    \sup_{p(x_{[k]},t_{[a]})\in \mathcal{P}}\epsilon I(X_{[k]};\tilde Y_{[k]},Z_m|T_m)+\sum_{m,\mathcal{S}}\gamma_{\mathcal{S},m}I(X_{\mathcal S};Z_m,Y_{{\mathcal S^c}}|X_{\mathcal S^c},T_m)\label{Vsup}
\end{align}
for every $\epsilon>0$ where
\begin{equation}
    \tilde{Y}_{i} = X_i + G_i
\end{equation}
for standard Gaussian noise $G_i$ (which are mutually independent of each other, and independent of all previously defined variables). 
Let $p^*(x_{[k]},t_{[a]})$ be a maximizer in \eqref{Vsup}, which exists based on arguments in \cite[Appendix II]{gen14}. The power constraints yield tightness, and the additive Gaussian noise yields the continuity of the various terms with respect to weak convergence.  Alternatively, one can use the approach in \cite{mahvari2023stability}, which does not require the existence of a maximizer.

Take two i.i.d. copies of the maximizer and denote them as $X_{[k]},T_{[a]}$ and $X'_{[k]},T'_{[a]}$ respectively. Thus,
$X_{[k]},T_{[a]},Z_{[a]},Y_{{[k]}},\tilde Y_{{[k]}}$ and $X'_{[k]},T'_{[a]},Z'_{[a]},Y'_{{[k]}},\tilde Y'_{{[k]}}$ are i.i.d.\ copies.
Denote the rotated versions by $(\cdot)_+ = \frac{(\cdot) + (\cdot)'}{\sqrt{2}}$ and let $(\cdot)_- = \frac{(\cdot) - (\cdot)'}{\sqrt{2}}$. The rotation results in the $+$ and $-$ variables
\begin{equation}
   (T_{[a]+},X_{[k]+},Z_{[a]+},Y_{{[k]}+},\tilde Y_{{[k]}+}), \quad
   (T_{[a]-},X_{[k]-},Z_{[a]-},Y_{{[k]}-},\tilde Y_{{[k]}-})
\end{equation}
respectively. Since $p(z_{[a]},y_{{[k]}},\tilde y_{{[k]}}|x_{[k]})$ is an additive Gaussian noise channel, the following Markov chains hold after rotation:
\begin{align}
   (T_{[a]+},T_{[a]-}, X_{[k]-},Z_{[a]-}, Y_{{[k]}-},\tilde Y_{{[k]}-})
   & \mkv X_{[k]+}
   \mkv (Z_{[a]+},Y_{{[k]}+},\tilde Y_{{[k]}+})
   \label{MCc1} \\
   (T_{[a]+},T_{[a]-}, X_{[k]+},Z_{[a]+}, Y_{{[k]}+},\tilde Y_{{[k]}+})
   & \mkv X_{[k]-}
   \mkv (Z_{[a]-},Y_{{[k]}-},\tilde Y_{{[k]}-}).
\label{MCc2}
\end{align}
Guided by the proof of Theorem \ref{thm2}, which uses the past of $Z^{j-1}$ for single-letterization, the idea is to consider the two-letter form of the expressions with the $+$ and $-$ variables, and single-letterize it using the identification
$T_{m+}, T_{m-}$ for the $-$ variables, and $T_{m+}, T_{m-}, Z_{m-}$ for the $+$ variables (interpreting the $-$ variables as the past, and the $+$ variables as the future).

We start from the DB constraints. First, observe that the DB constraint
\begin{align} 
    & I_\lambda(X_{i_1}Y_{{i_1}};X_{i_2}Y_{{i_2}};\cdots;X_{i_u}Y_{{i_u}}|T_m,Z_m)
    \nonumber \\
    & \ge I_\lambda(X_{i_1};X_{i_2};\cdots;X_{i_u}|T_m)
    + \left(1-\sum\nolimits_{\mathcal B\subsetneq \mathcal{U}}\lambda_{\mathcal B}\right) I(X_{[k]};Z_m,Y_{{\mathcal{U}}}|X_{\mathcal{U}},T_m)
\end{align}
can be written as
\begin{align} 
    & \left(1-\sum\nolimits_{\mathcal B\subsetneq \mathcal{U}}\lambda_{\mathcal B}\right) H(X_{[k]}Y_{\mathcal{U}}|T_m,Z_m) - \left(1-\sum\nolimits_{\mathcal B\subsetneq \mathcal U}\lambda_{\mathcal B}\right) H(X_{[k]}|T_m)
    \nonumber \\
    &  + \sum\nolimits_{\mathcal B\subsetneq \mathcal U}\lambda_{\mathcal B} H(X_{{\mathcal B}^c}Y_{{\mathcal B^c}}|T_m,Z_m)- \sum\nolimits_{\mathcal B\subsetneq \mathcal U}\lambda_{\mathcal B} H(X_{{\mathcal B}^c}|T_m)\geq 0
    .
\end{align}
Since $X_{[k]},T_{[a]},Y,Z_{[a]},Y_{{[k]}}$ and $X'_{[k]},T'_{[a]},Y',Z'_{[a]},Y'_{{[k]}}$ are i.i.d.\ copies of the maximizer and satisfy the DB constraints, we obtain the following chain of inequalities:
\begin{align} 
    0&\leq \left(1-\sum\nolimits_{\mathcal B\subsetneq \mathcal{U}}\lambda_{\mathcal B}\right) H(X_{[k]}X'_{[k]}Y_{\mathcal{U}}Y'_{\mathcal{U}}|T_m,T'_m,Z_m,Z'_m) - \left(1-\sum\nolimits_{\mathcal B\subsetneq \mathcal U}\lambda_{\mathcal B}\right) H(X_{[k]}X'_{[k]}|T_m,T'_m)
    \nonumber \\
    &  + \sum\nolimits_{\mathcal B\subsetneq \mathcal U}\lambda_{\mathcal B} H(X_{{\mathcal B}^c}X'_{{\mathcal B}^c}Y_{{\mathcal B^c}}Y'_{{\mathcal B^c}}|T_m,T'_m,Z_m,Z'_m)- \sum\nolimits_{\mathcal B\subsetneq \mathcal U}\lambda_{\mathcal B} H(X_{{\mathcal B}^c}X'_{{\mathcal B}^c}|T_m,T'_m)
    \nonumber \\
    &= \left(1-\sum\nolimits_{\mathcal B\subsetneq \mathcal{U}}\lambda_{\mathcal B}\right) H(X_{[k]+}X_{[k]-}Y_{\mathcal{U}+}Y_{\mathcal{U}-}|T_{m+},T_{m-},Z_{m+},Z_{m-})
    \nonumber \\
    & \quad- \left(1-\sum\nolimits_{\mathcal B\subsetneq \mathcal U}\lambda_{\mathcal B}\right) H(X_{[k]+}X_{[k]-}|T_{m+},T_{m-})
    \nonumber \\
    &  \quad+ \sum\nolimits_{\mathcal B\subsetneq \mathcal U}\lambda_{\mathcal B} H(X_{{\mathcal B}^c+}X_{{\mathcal B}^c-}Y_{{\mathcal B^c}+}Y_{{\mathcal B^c}-}|T_{m+},T_{m-},Z_{m+},Z_{m-})- \sum\nolimits_{\mathcal B\subsetneq \mathcal U}\lambda_{\mathcal B} H(X_{{\mathcal B}^c+}X_{{\mathcal B}^c-}|T_{m+},T_{m-})
    \nonumber \\
    &\overset{(a)}{\leq} \color{blue}\left(1-\sum\nolimits_{\mathcal B\subsetneq \mathcal{U}}\lambda_{\mathcal B}\right) H(X_{[k]-}Y_{\mathcal{U}-}|T_{m+},T_{m-},Z_{m-})- \left(1-\sum\nolimits_{\mathcal B\subsetneq \mathcal U}\lambda_{\mathcal B}\right) H(X_{[k]-}|T_{m+},T_{m-})
    \nonumber \\
    &  \quad\color{blue}+ \sum\nolimits_{\mathcal B\subsetneq \mathcal U}\lambda_{\mathcal B} H(X_{{\mathcal B}^c-}Y_{{\mathcal B^c}-}|T_{m+},T_{m-},Z_{m-})- \sum\nolimits_{\mathcal B\subsetneq \mathcal U}\lambda_{\mathcal B} H(X_{{\mathcal B}^c-}|T_{m+},T_{m-})
    \nonumber \\
    &+ \color{purple}\left(1-\sum\nolimits_{\mathcal B\subsetneq \mathcal{U}}\lambda_{\mathcal B}\right) H(X_{[k]+}Y_{\mathcal{U}+}|T_{m+},T_{m-},Z_{m-},Z_{m+}) - \left(1-\sum\nolimits_{\mathcal B\subsetneq \mathcal U}\lambda_{\mathcal B}\right) H(X_{[k]+}|T_{m+},T_{m-},Z_{m-})
    \nonumber \\
    &  \quad\color{purple}+ \sum\nolimits_{\mathcal B\subsetneq \mathcal U}\lambda_{\mathcal B} H(X_{{\mathcal B}^c+}Y_{{\mathcal B^c}+}|T_{m+},T_{m-},Z_{m-},Z_{m+})- \sum\nolimits_{\mathcal B\subsetneq \mathcal U}\lambda_{\mathcal B} H(X_{{\mathcal B}^c+}|T_{m+},T_{m-},Z_{m-})
\end{align}
where the colored terms single-letterize the DB constraint for the $+$ and $-$ components using the identification
$T_{m+}, T_{m-}$ for the $-$ variables, and $T_{m+}, T_{m-}, Z_{m-}$ for the $+$ variables. Step (a) holds because, after the cancellation of common terms, it is equivalent to
\begin{align}
   & \left(1-\sum\nolimits_{\mathcal B\subsetneq \mathcal{U}} \lambda_{\mathcal B}\right)
   H(X_{[k]-}Y_{\mathcal{U}-} | X_{[k]+},Y_{\mathcal{U}+}, T_{m+},T_{m-},Z_{m+},Z_{m-})
   \nonumber \\
   & \quad- \left(1-\sum\nolimits_{\mathcal B\subsetneq \mathcal U}\lambda_{\mathcal B}\right) H(X_{[k]+}|T_{m+},T_{m-},X_{[k]-})
   \nonumber \\
   &  \quad+ \sum\nolimits_{\mathcal B\subsetneq \mathcal U}\lambda_{\mathcal B} H(X_{{\mathcal B}^c-}Y_{{\mathcal B^c}-} |X_{{\mathcal B}^c+},Y_{{\mathcal B^c}+}, T_{m+},T_{m-}, Z_{m+},Z_{m-})- \sum\nolimits_{\mathcal B \subsetneq \mathcal U} \lambda_{\mathcal B}  H(X_{{\mathcal B}^c+} |X_{{\mathcal B}^c-},T_{m+},T_{m-})
   \nonumber \\
   & \leq \left(1-\sum\nolimits_{\mathcal B\subsetneq \mathcal{U}} \lambda_{\mathcal B}\right) H(X_{[k]-}Y_{\mathcal{U}-} | T_{m+},T_{m-},Z_{m-}) + \sum\nolimits_{\mathcal B\subsetneq \mathcal U} \lambda_{\mathcal B} H(X_{{\mathcal B}^c-}Y_{{\mathcal B^c}-} | T_{m+},T_{m-},Z_{m-})
   \nonumber \\
   & \quad- \left(1-\sum\nolimits_{\mathcal B\subsetneq \mathcal U} \lambda_{\mathcal B}\right) H(X_{[k]+} | T_{m+},T_{m-},Z_{m-})
   - \sum\nolimits_{\mathcal B\subsetneq \mathcal U}\lambda_{\mathcal B} H(X_{{\mathcal B}^c+}|T_{m+},T_{m-},Z_{m-}).
\end{align}
Using \eqref{MCc1} and \eqref{MCc2}, the above is equivalent to
\begin{align}
   & \left(1-\sum\nolimits_{\mathcal B\subsetneq \mathcal{U}} \lambda_{\mathcal B}\right) H(X_{[k]-}Y_{\mathcal{U}-} | X_{[k]+},Y_{\mathcal{U}+},T_{m+},T_{m-},Z_{m-}) \nonumber \\
   & \quad- \left(1-\sum\nolimits_{\mathcal B\subsetneq \mathcal U}\lambda_{\mathcal B}\right) H(X_{[k]+}|T_{m+},T_{m-},X_{[k]-},Z_{m-})
   \nonumber \\
   & \quad + \sum\nolimits_{\mathcal B\subsetneq \mathcal U}\lambda_{\mathcal B} H(X_{{\mathcal B}^c-}Y_{{\mathcal B^c}-}|X_{{\mathcal B}^c+},Y_{{\mathcal B^c}+},T_{m+},T_{m-},Z_{m-})- \sum\nolimits_{\mathcal B\subsetneq \mathcal U}\lambda_{\mathcal B} H(X_{{\mathcal B}^c+}|X_{{\mathcal B}^c-},T_{m+},T_{m-},Z_{m-})
   \nonumber \\
   & \leq \left(1-\sum\nolimits_{\mathcal B\subsetneq \mathcal{U}}\lambda_{\mathcal B}\right) H(X_{[k]-}Y_{\mathcal{U}-} | T_{m+},T_{m-},Z_{m-}) + \sum\nolimits_{\mathcal B\subsetneq \mathcal U} \lambda_{\mathcal B} H(X_{{\mathcal B}^c-}Y_{{\mathcal B^c}-}|T_{m+},T_{m-},Z_{m-})
   \nonumber \\
   & \quad- \left(1-\sum\nolimits_{\mathcal B\subsetneq \mathcal U} \lambda_{\mathcal B}\right) H(X_{[k]+} | T_{m+},T_{m-},Z_{m-})
   - \sum\nolimits_{\mathcal B\subsetneq \mathcal U} \lambda_{\mathcal B} H(X_{{\mathcal B}^c+}|T_{m+},T_{m-},Z_{m-}).
\end{align}
The above can be rewritten as 
\begin{align}
   & \left(1-\sum\nolimits_{\mathcal B\subsetneq \mathcal U} \lambda_{\mathcal B}\right) I(X_{[k]+};X_{[k]-} | T_{m+},T_{m-},Z_{m-}) + \sum\nolimits_{\mathcal B\subsetneq \mathcal U} \lambda_{\mathcal B} I(X_{{\mathcal B}^c+};X_{{\mathcal B}^c-} |T_{m+},T_{m-},Z_{m-})
   \nonumber \\
   & \leq \left(1-\sum\nolimits_{\mathcal B\subsetneq \mathcal{U}} \lambda_{\mathcal B}\right) I(X_{[k]-}Y_{\mathcal{U}-};X_{[k]+}Y_{\mathcal{U}+} | T_{m+},T_{m-},Z_{m-})
   \nonumber \\
   & \quad+ \sum\nolimits_{\mathcal B\subsetneq \mathcal U} \lambda_{\mathcal B} I(X_{{\mathcal B}^c-}Y_{{\mathcal B^c}-} ; X_{{\mathcal B}^c+}Y_{{\mathcal B^c}+} | T_{m+},T_{m-},Z_{m-})
\end{align}
But from \eqref{MCc1} and \eqref{MCc2}, we have
\begin{equation}
    I(X_{[k]-}Y_{\mathcal{U}-} ; X_{[k]+}Y_{\mathcal{U}+} | T_{m+},T_{m-},Z_{m-})
    = I(X_{[k]+};X_{[k]-} | T_{m+},T_{m-},Z_{m-})
\end{equation}
so the inequality follows.

Next, let us consider the objective function. Let $V$ be the supremum in \eqref{Vsup}. We have
\begin{subequations}
\begin{align}
   2V &= \epsilon I(X_{[k]},X'_{[k]} ; \tilde Y_{[k]},\tilde Y'_{[k]},Z_m,Z'_m | T_m,T'_m) + \sum_{m,\mathcal{S}} \gamma_{\mathcal{S},m} I(X_{\mathcal S}X'_{\mathcal S} ; Z_m,Z'_m,Y_{{\mathcal S^c}},Y'_{{\mathcal S^c}} | X_{\mathcal S^c},X'_{\mathcal S^c},T_m,T'_m)
   \nonumber \\
   & = \epsilon I(X_{[k]+},X_{[k]-};\tilde Y_{[k]+},\tilde Y_{[k]-},Z_{m+},Z_{m-}|T_{m+},T_{m-})
   \nonumber \\
   & \quad + \sum_{m,\mathcal{S}} \gamma_{\mathcal{S},m} I(X_{\mathcal S+}X_{\mathcal S-} ; Z_{m+},Z_{m-},Y_{{\mathcal S^c}+},Y_{{\mathcal S^c}-} | X_{\mathcal S^c+},X_{\mathcal S^c-},T_{m+},T_{m-})
   \nonumber \\
   & = \epsilon I(X_{[k]-};\tilde Y_{[k]-},Z_{m-} | T_{m+},T_{m-})
   \nonumber \\
   & \quad + \epsilon I(X_{[k]+};\tilde Y_{[k]+},Z_{m+}|T_{m+},T_{m-},\tilde Y_{[k]-},Z_{m-})
   \nonumber \\
   & \quad + \sum_{m,\mathcal{S}} \gamma_{\mathcal{S},m}  h(Z_{m+},Z_{m-},Y_{{\mathcal S^c}+},Y_{{\mathcal S^c}-} | X_{\mathcal S^c+},X_{\mathcal S^c-},T_{m+},T_{m-})
   \nonumber \\
   & \qquad- \sum_{m,\mathcal{S}} \gamma_{\mathcal{S},m}h(Z_{m+},Z_{m-},Y_{{\mathcal S^c}+},Y_{{\mathcal S^c}-}|X_{[k]+},X_{[k]-},T_{m+},T_{m-})
   \nonumber \\
   & \overset{(a)}{=} \epsilon I(X_{[k]-};\tilde Y_{[k]-},Z_{m-}|T_{m+},T_{m-})
   \nonumber \\
   & \quad + \epsilon I(X_{[k]+};\tilde Y_{[k]+},Z_{m+}|T_{m+},T_{m-},Z_{m-})-\epsilon I(\tilde Y_{[k]-};\tilde Y_{[k]+},Z_{m+}|T_{m+},T_{m-},Z_{m-})
   \nonumber\\
   & \quad \sum_{m,\mathcal{S}} \gamma_{\mathcal{S},m} h(Z_{m+},Z_{m-},Y_{{\mathcal S^c}+},Y_{{\mathcal S^c}-} | X_{\mathcal S^c+},X_{\mathcal S^c-},T_{m+},T_{m-})
   \nonumber \\
   & \qquad - \sum_{m,\mathcal{S}} \gamma_{\mathcal{S},m} h(Z_{m+},Y_{{\mathcal S^c}+}|X_{[k]+},T_{m+},T_{m-},Z_{m-})
   \nonumber \\
   & \qquad-\nonumber \sum_{m,\mathcal{S}} \gamma_{\mathcal{S},m} h(Z_{m-},Y_{{\mathcal S^c}-} | X_{[k]-},T_{m+},T_{m-})
   \nonumber \\
   & = \color{black} \epsilon I(X_{[k]+};\tilde Y_{[k]+},Z_{m+}|T_{m+},T_{m-},Z_{m-})+\sum_{m,\mathcal{S}}\gamma_{\mathcal{S},m}I(X_{\mathcal S+};Z_{m+},Y_{{\mathcal S^c}+}|X_{\mathcal S^c+},T_{m+},T_{m-},Z_{m-})
   \nonumber \\
   & \quad \color{purple} + \epsilon I(X_{[k]-};\tilde Y_{[k]-},Z_{m-} | T_{m+},T_{m-}) + \sum_{m,\mathcal{S}} \gamma_{\mathcal{S},m} I(X_{\mathcal S-};Z_{m-},Y_{{\mathcal S^c}-} | X_{\mathcal S^c-},T_{m+},T_{m-})
   \nonumber \\
   & \quad-\sum_{m,\mathcal{S}} \gamma_{\mathcal{S},m} I(Z_{m-},Y_{{\mathcal S^c}-};X_{\mathcal S^c+} | X_{\mathcal S^c-},T_{m+},T_{m-})
   \label{ggap1}
   \\
   & \quad-
   \sum_{m,\mathcal{S}} \gamma_{\mathcal{S},m} I(Z_{m+},Y_{{\mathcal S^c}+};Y_{{\mathcal S^c}-},X_{\mathcal S^c-} | X_{\mathcal S^c+},T_{m+},T_{m-},Z_{m-})
   \label{ggap2}
   \\
   & \quad -\epsilon I(\tilde Y_{[k]-};\tilde Y_{[k]+},Z_{m+}|T_{m+},T_{m-},Z_{m-})
   \label{ggap3}
\end{align}
\end{subequations}
where step $(a)$ follows from \eqref{MCc1} and \eqref{MCc2}. The colored terms are single-letterizations for the $+$ and $-$ components using the identification $T_{m+}, T_{m-}$ for the $-$ variables, and $T_{m+}, T_{m-}, Z_{m-}$ for the $+$ variables.

Let $Q\in\{+,-\}$ be a uniform time-sharing random variable and set $\hat{T}_m=(T_{m+},T_{m-},Q)$ if $Q=-$ and  $\hat{T}_m=(T_{m+},T_{m-},Z_{m-},Q)$ if $Q=+$. The above argument shows that the gap terms in \eqref{ggap1}, \eqref{ggap2} and \eqref{ggap3} vanish for the maximizer. In particular, since $\epsilon>0$  we deduce
\begin{align}
    I(\tilde Y_{[k]-};\tilde Y_{[k]+},Z_{m+} | T_{m+},T_{m-},Z_{m-})
    & = 0.
\end{align}
Proposition 2 in \cite{geng2014capacity} implies
\begin{align}
   I(X_{[k]-};X_{[k]+}|T_{m+},T_{m-},Z_{m-})
   & =0. \label{eqnDM1}
\end{align}
We also have
\begin{align}
I(Z_{m-};X_{[k]+}|T_{m+},T_{m-},X_{[k]-})&=0.\label{eqnDM2}
\end{align}
Equations \eqref{eqnDM1} and \eqref{eqnDM2} indicate Markov chains in different orders. The Double Markovity lemma \cite[Exercise 16.25]{csk11} (see also \cite[Lemma 6]{gohari2021outer}) shows that
\begin{equation}
    I(X_{[k]+};X_{[k]-},Z_{m-}|T_{m+},T_{m-}) = 0
\end{equation}
because $Z_{m-}$ and $X_{[k]-}$ have no Gacs-Korner common part.
This implies $I(X_{[k]+};X_{[k]-}|T_{m},T'_{m})=0$. By the Skitovic-Darmois characterization of Gaussian distributions, $X_{[k]}$ is jointly Gaussian conditioned on $T_m$, and the covariance matrix of $X_{[k]}$ given $T_m=t_m$ is independent of $t_m$.
This property should hold for \emph{any} maximizer $(X_{[k]}, T_{[a]})$. Let $K_{X_{[k]}}$ and $K_{X_{[k]} \mid T_m}$ denote the unconditional and conditional covariance matrices, respectively.

We next identify a new maximizer $(X_{[k]}, \tilde T_{[a]})$ satisfying
\begin{equation} \itemsep 0pt
   \label{eqn1nnt}
   p(x_{[k]}, \tilde t_{[a]}) = p(x_{[k]}) \cdot \left( \prod_{m \in [a]} p(\tilde t_m \mid x_{[k]}) \right)
\end{equation}
and the following two properties:
\begin{itemize}
  \item $(X_{[k]}, \tilde T_m)$ is a jointly Gaussian random vector for all $m$;
  \item $\tilde T_m$ is a $k$-dimensional random vector.
\end{itemize}
By \eqref{eqn1nnt}, we only need to define the joint distribution of $(X_{[k]}, \tilde T_m)$.
Note that $K_{X_{[k]} \mid T_m} \preceq K_{X_{[k]}}$, and let $\tilde T_m$ be a $k$-dimensional Gaussian vector with covariance matrix
\begin{equation}
   K_{\tilde T_m}
   = K_{X_{[k]}} - K_{X_{[k]} \mid T_m}
\end{equation}
and let $W_m$ be a Gaussian random vector (independent of $\tilde{T}_m$) with covariance matrix
\begin{equation}
   K_{W_m} = K_{X_{[k]} \mid T_m}.
\end{equation}
Define
\begin{equation}
   X_{[k]} = W_m + \tilde{T}_m.
\end{equation}
In this construction, $(X_{[k]}, \tilde T_m)$ is jointly Gaussian. Moreover, $X_{[k]}$ has unconditional covariance
\begin{equation}
K_{W_m} + K_{\tilde T_m} = K_{X_{[k]}},
\end{equation}
and conditional covariance
\begin{equation}
K_{X_{[k]} \mid \tilde T_m}= K_{X_{[k]} \mid T_m}.
\end{equation}
Therefore, this transformation preserves all relevant mutual information terms and yields a maximizer.
\end{proof}

\section{Calculations for the Gaussian Relay Channel}
\label{RelayAppendix}

Consider a Gaussian relay channel with equal power constraints $P$ on $X$ and $X_{\rmr}$:
\begin{subequations}
\begin{align} 
   Y_\mathrm{r} & = g_{12} X+N_\mathrm{r} \\
   Y &= g_{13} X+g_{23} X_\mathrm{r}+N_e \\
   Z_1 & = \alpha X + \beta X_\mathrm{r} + \gamma N_e + \eta N_\mathrm{r} + \zeta N
\end{align}
\end{subequations}
where $N_e, N_\rmr, N$ are  mutually independent standard Gaussian random variables.

We evaluate the bound for
\begin{align}
    K_{X,X_{\mathrm{r}}} & 
    =\begin{bmatrix} P & \rho P \\ \rho P & P \end{bmatrix} \\
    K_{X,X_{\mathrm{r}}|T_1} & =
    \begin{bmatrix} Q_1 & \tilde{\rho} \sqrt{Q_1Q_2} \\ \tilde{\rho} \sqrt{Q_1Q_2} & Q_2 \end{bmatrix} \preceq K_{X,X_{\mathrm{r}}}.
\end{align}
We have
\begin{align}
    & h(Y_{\mathrm{r}},Y,Z_1|X,X_{\mathrm{r}},T_1) = h(N_\rmr, N_e, \gamma N_e+\eta N_\mathrm{r}+\zeta N)=\frac12\log((2\pi e)^3\zeta^2) \\
    & h(Y_{\mathrm{r}},Y,Z_1|X_{\mathrm{r}},T_1) = h(g_{12} X+N_\mathrm{r}, g_{13} X+N_e, \alpha X+\gamma N_e+\eta N_\mathrm{r}+\zeta N|X_\rmr,T_1) \nonumber \\
    & \quad = \frac12\log\left((2\pi e)^3
    \det \begin{pmatrix} g_{12}^2Q_1(1-\tilde{\rho}^2)+1 & g_{12}g_{13}Q_1(1-\tilde{\rho}^2) & g_{12}\alpha Q_1(1-\tilde{\rho}^2)+\eta \\
    g_{12}g_{13}Q_1(1-\tilde{\rho}^2) & g_{13}^2Q_1(1-\tilde{\rho}^2)+1 & g_{13}\alpha Q_1(1-\tilde{\rho}^2)+\gamma \\
    g_{12}\alpha Q_1(1-\tilde{\rho}^2)+\eta & g_{13}\alpha Q_1(1-\tilde{\rho}^2)+\gamma & \alpha^2Q_1(1-\tilde{\rho}^2)+\gamma^2+\eta^2+\zeta^2
    \end{pmatrix} \right) \nonumber \\
    & \quad = \frac12\log\left((2\pi e)^3\left(\zeta^2 + Q_1(1-\tilde{\rho}^2)\left[(g_{12}\eta + g_{13}\gamma - \alpha)^2 + \zeta^2(g_{12}^2 + g_{13}^2)\right]\right)\right)
\end{align}
and therefore
\begin{align}
    I(X;Y,Y_{\mathrm{r}},Z_1|X_{\mathrm{r}},T_1) & = \frac12\log\left(\zeta^2 + Q_1(1-\tilde{\rho}^2)\left[(g_{12}\eta + g_{13}\gamma - \alpha)^2 + \zeta^2(g_{12}^2 + g_{13}^2)\right]\right)-\frac12\log(\zeta^2).
\end{align}
We have
\begin{align}
    & h(Y,Z_1|T_1)=h(g_{13} X+g_{23} X_\mathrm{r}+N_e,\alpha X+\beta X_\mathrm{r}+\gamma N_e+\eta N_\mathrm{r}+\zeta N|T_1) =\frac12\log\left((2\pi e)^2
    \det(M) \right)
\end{align}
where 
$$
   M = \begin{pmatrix} g_{13}^2Q_1+g_{23}^2Q_2+2g_{13}g_{23}\tilde{\rho} \sqrt{Q_1Q_2}+1
   & \alpha g_{13}Q_1+\beta g_{23}Q_2+(\alpha g_{23}+\beta g_{13})\tilde{\rho} \sqrt{Q_1Q_2}+\gamma
   \\
   \alpha g_{13}Q_1+\beta g_{23}Q_2+(\alpha g_{23}+\beta g_{13})\tilde{\rho} \sqrt{Q_1Q_2}+\gamma
   & \alpha^2Q_1+\beta^2Q_2+2\alpha\beta\tilde{\rho} \sqrt{Q_1Q_2}+\gamma^2+\eta^2+\zeta^2
    \end{pmatrix}.
$$
Next, we have
\begin{align}
    &h(Y,Z_1|T_1,X,X_\rmr)=h(N_e,\gamma N_e+\eta N_\mathrm{r}+\zeta N|T_1)=\frac12\log((2\pi e)^2(\eta^2+\zeta^2))
\end{align}
and therefore
\begin{align}
   I(X,X_{\mathrm{r}};Y,Z_1|T_1)
   = \frac12\log \bigg\{
   & \left( g_{13}^2Q_1 + g_{23}^2Q_2 + 2g_{13}g_{23}\tilde{\rho} \sqrt{Q_1Q_2}+1 \right)
   \nonumber \\
   & \cdot \left( \alpha^2Q_1 + \beta^2Q_2 + 2\alpha\beta\tilde{\rho} \sqrt{Q_1Q_2} + \gamma^2 + \eta^2 + \zeta^2 \right)
   \nonumber \\
   & - \left( \alpha g_{13}Q_1+\beta g_{23}Q_2 + (\alpha g_{23}+\beta g_{13})\tilde{\rho} \sqrt{Q_1Q_2}+\gamma \right)^2 \bigg\}
   - \frac12 \log(\eta^2+\zeta^2).
\end{align}

Next, consider the expressions
\begin{align}
   I(X;X_\rmr|T_1) & = -\frac12\log(1-\tilde{\rho}^2) \\
   I(X;X_{\mathrm{r}},Y_{\mathrm{r}} | T_1,Z_1)
   &= I(X;X_{\mathrm{r}},Y_{\mathrm{r}},Z_1 | T_1)-I(X;Z_1|T_1)
   \\
   h(Z_1|T_1) & = \frac{1}{2} \log\left( 2\pi e(\alpha^2Q_1 + \beta^2Q_2 + 2\alpha\beta\tilde{\rho} \sqrt{Q_1Q_2} + \gamma^2+\eta^2+\zeta^2) \right)
   \\
   h(Z_1|X,T_1) & = \frac{1}{2} \log\left( 2\pi e(\beta^2Q_2(1-\tilde{\rho}^2) + \gamma^2 + \eta^2 + \zeta^2) \right).
\end{align}
We compute
\begin{align}
    I(X;Z_1|T_1) & = \frac{1}{2} \log\left( \alpha^2Q_1 + \beta^2Q_2+2\alpha\beta\tilde{\rho} \sqrt{Q_1Q_2}+\gamma^2+\eta^2+\zeta^2 \right)
    \nonumber \\
    & \quad - \frac{1}{2}\log\left( \beta^2Q_2(1-\tilde{\rho}^2)+\gamma^2+\eta^2+\zeta^2 \right).
\end{align}
Finally, we compute 
$I(X;X_{\mathrm{r}},Y_{\mathrm{r}},Z_1|T_1)$
via
\begin{align}
   & h(X_{\mathrm{r}},Y_{\mathrm{r}},Z_1|T_1)
   = h(X_{\mathrm{r}},g_{12} X + N_\mathrm{r},\alpha X + \beta X_\mathrm{r} + \gamma N_e + \eta N_\mathrm{r} + \zeta N|T_1)
   \nonumber \\
   & = \frac{1}{2} \log(2\pi e Q_2) + h\left (g_{12} X + N_\mathrm{r},\alpha X + \gamma N_e + \eta N_\mathrm{r}+\zeta N|T_1,X_\rmr \right)
   \nonumber \\
   & = \frac{1}{2} \log(2\pi e Q_2) + 
   \frac12 \log\left( (2\pi e)^2\det
   \begin{bmatrix}
   g_{12}^2Q_1(1-\tilde{\rho}^2)+1 & g_{12}\alpha Q_1(1-\tilde{\rho}^2)+\eta  \\
   g_{12}\alpha Q_1(1-\tilde{\rho}^2)+\eta & \alpha^2Q_1(1-\tilde{\rho}^2)+\gamma^2+\eta^2+\zeta^2
   \end{bmatrix}\right)
   \nonumber \\
   & = \frac{1}{2}\log((2\pi e)^3)+\frac{1}{2}\log(Q_2)+\frac{1}{2}\log\left(\gamma^2+\zeta^2+Q_1 (1 - \tilde{\rho}^2) \left[ (\alpha - \eta g_{12})^2 + g_{12}^2 (\gamma^2 + \zeta^2) \right]\right)
\end{align}
and
\begin{align}
    h(X_{\mathrm{r}},Y_{\mathrm{r}},Z_1|T_1,X)
    & = h(X_{\mathrm{r}},N_\mathrm{r},\beta X_\mathrm{r} + \gamma N_e + \eta N_\mathrm{r} + \zeta N|X,T_1)
    \nonumber \\
    & = \frac{1}{2} \log\left( (2\pi e)^3 (1-\tilde{\rho}^2) Q_2(\gamma^2+\zeta^2) \right).
\end{align}
Thus, we have
\begin{align}
    I(X;X_{\mathrm{r}},Y_{\mathrm{r}},Z_1|T_1)
    & = \frac{1}{2}\log\left( \gamma^2 + \zeta^2 + Q_1 (1 - \tilde{\rho}^2) \left[ (\alpha - \eta g_{12})^2 + g_{12}^2 (\gamma^2 + \zeta^2) \right] \right)
    \nonumber \\
    & \quad - \frac{1}{2} \log\left( (1-\tilde{\rho}^2) (\gamma^2+\zeta^2) \right).
\end{align}

\section{Noisy Feedback}
\label{noisy-feedback}

For noisy feedback, the bounds \eqref{DBMAC-21}--\eqref{DBMAC-2b} are
\begin{subequations}
\begin{align}
    R_{1} & \leq  \min\left(\, I(X_{1};Y|X_{2},T_1),\,
    I(X_1;Y|X_2,T_2) \,\right)
    \label{DBMAC-21-1} \\
    R_{2} & \leq \min\left(\, I(X_{2};Y|X_{1},T_1),\, I(X_2;Y|X_1,T_2) \,\right)
    \label{DBMAC-22-1} \\
    R_{1}+R_{2} & \leq \min\left(\, I(X_{1},X_{2};Y|T_1), \, I(X_{1},X_{2};Y|T_2) \,\right)
    \label{DBMAC-23-1} \\
    I(X_{1};X_{2}|T_1) & \leq I(X_{1};X_{2}|Y_{1},Y_{2},T_1)
    \label{DBMAC-2a-1} \\
    I(X_1;X_2|T_2) & \leq I(X_1 ; X_2 | Y,T_2) .
    \label{DBMAC-2b-1}
\end{align}
\end{subequations}
The papers \cite{gastpar06,tandon2011dependence} established \eqref{DBMAC-21-1}--\eqref{DBMAC-2a-1} and \cite[Sec.~X]{tandon2011dependence} shows that joint Gaussian $X_1,X_2,T_1$ are optimal. Moreover, if one chooses $p_{T_2|X_1,X_2}=p_{T_1|X_1,X_2}$, the expression \cite[eq. (66)]{tandon2011dependence} shows that \eqref{DBMAC-2a-1} implies \eqref{DBMAC-2b-1}.  Thus, Corollary~\ref{cor-6} does not improve \cite[Theorem~1]{tandon2011dependence} for noisy feedback.

\begin{remark}
The above example gives insight:\ the bound \eqref{DBMAC-2a-1} is stronger than \eqref{DBMAC-2b-1} for finite noise variances, but the opposite is true for infinite noise variances. More precisely, for $\mathsf{Var}(N_1)\rightarrow\infty$ and $\mathsf{Var}(N_2)\rightarrow\infty$, the papers \cite{gastpar06,tandon2011dependence} show one recovers the capacity region without feedback.
However, if we begin with $\mathsf{Var}(N_1)=\mathsf{Var}(N_2)=\infty$, the bound \eqref{DBMAC-2a-1} is vacuous and Corollary~\ref{cor-6} gives the cut-set bound. We thus have a discontinuity at the limit.
\end{remark}

\begin{remark}
The paper \cite{kramer2006dependence} points out that the DB constraint \eqref{DBMAC-2b-1} restricts the correlations, while the cut-set bound does not, but \eqref{DBMAC-2b-1} admits the correlations that optimize the cut-set bound. 
\end{remark}

\begin{remark}
We simulated the sum-rate bound in Theorem \ref{thm2} for
\begin{align}
    Z_1 & = (Y_1,Y_2,\tilde{Z}_1) \\
    \tilde{Z}_1 & = \alpha X_1 + \beta X_2 + \gamma N + \theta N_3
\end{align}
for various parameters $\alpha,\beta,\gamma,\theta$ and noise $N_3$ independent of the channel inputs and other noise. However, we did not encounter examples that improve upon \cite[Theorem~1]{tandon2011dependence}.
\end{remark}


\bibliographystyle{IEEEtran}
\bibliography{mybiblio}

@article{zhang1986new,
  title={New outer bounds to capacity regions of two-way channels},
  author={Zhang, Zhen and Berger, Toby and Schalkwijk, J},
  journal={IEEE Trans. Inf. Theory},
  volume={32},
  number={3},
  pages={383--386},
  year={1986},
  publisher={IEEE}
}

@article{mahvari2023stability,
  title={Stability of {B}ernstein’s theorem and soft doubling for vector {G}aussian channels},
  author={Mahvari, Mohammad Mahdi and Kramer, Gerhard},
  journal={IEEE Trans. Inf. Theory},
  volume={69},
  number={10},
  pages={6231--6250},
  year={2023},
  publisher={IEEE}
}

@ARTICLE{5625626,
  author={Nitinawarat, Sirin and Ye, Chunxuan and Barg, Alexander and Narayan, Prakash and Reznik, Alex},
  journal={IEEE Trans. Inf. Theory}, 
  title={Secret Key Generation for a Pairwise Independent Network Model}, 
  year={2010},
  volume={56},
  number={12},
  pages={6482-6489},
  doi={10.1109/TIT.2010.2081210}}

@article{Maurer1993,
  title={Secret key agreement by public discussion from common information},
  author={Maurer, Ueli M},
  journal={IEEE Trans. Inf. Theory},
  volume={39},
  number={3},
  pages={733--742},
  year={1993},
  publisher={IEEE}
}

@article{AhlswedeCsiszar1993,
  title={Common randomness in information theory and cryptography. {I.} {S}ecret sharing},
  author={Ahlswede, Rudolf and Csisz{\'a}r, Imre},
  journal={IEEE Trans. Inf. Theory},
  volume={39},
  number={4},
  pages={1121--1132},
  year={1993},
  publisher={IEEE}
}

@book{thomas2006elements,
  title={Elements of Information Theory},
  author={Cover, T. and Thomas, J.},
  year={2006},
  publisher={Wiley-Interscience}
}

@ARTICLE{Cover-Leung-IT81,
  author={Cover, T. and Leung, C.},
  journal={IEEE Trans. Inf. Theory}, 
  title={An achievable rate region for the multiple-access channel with feedback}, 
  year={1981},
  volume={27},
  number={3},
  pages={292-298},
  doi={10.1109/TIT.1981.1056357}}

@ARTICLE{7976410,
  author={Gohari, Amin and Anantharam, Venkat},
  journal={IEEE Trans. Inf. Theory}, 
  title={Comments On “Information-Theoretic Key Agreement of Multiple Terminals—Part {I}”}, 
  year={2017},
  volume={63},
  number={8},
  pages={5440-5442},
  doi={10.1109/TIT.2017.2685579}}

@article{zhang2017multi,
  title={Multi-key generation over a cellular model with a helper},
  author={Zhang, Huishuai and Liang, Yingbin and Lai, Lifeng and Shitz, Shlomo Shamai},
  journal={IEEE Trans. Inf. Theory},
  volume={63},
  number={6},
  pages={3804--3822},
  year={2017},
  publisher={IEEE}
}

@inproceedings{zhang2017multiple,
  title={Multiple secret key generation: Information theoretic models and key capacity regions},
  author={Zhang, Huishuai and Liang, Yingbin and Lai, Lifeng and Shamai, Shlomo},
  booktitle={Proc. Inf. Theoretic Secur. Privacy Inf. Syst.},
  pages={333--360},
  year={2017}
}

@INPROCEEDINGS{9366098,
  author={Poostindouz, Alireza and Safavi-Naini, Reihaneh},
  booktitle={Int. Symp. Inf. Theory Applic.}, 
  title={A Channel Model of Transceivers for Multiterminal Secret Key Agreement}, 
  year={2020},
  volume={},
  number={},
  pages={412-416},
  doi={}}

@article{gohari2021outer,
  title={Outer bounds for multiuser settings: The auxiliary receiver approach},
  author={Gohari, Amin and Nair, Chandra},
  journal={IEEE Trans. Inf. Theory},
  volume={68},
  number={2},
  pages={701--736},
  year={2021},
  publisher={IEEE}
}

@article{narayan2016multiterminal,
  title={Multiterminal secrecy by public discussion},
  author={Narayan, Prakash and Tyagi, Himanshu},
  journal={Foundations and Trends{\textregistered} in Communications and Information Theory},
  volume={13},
  number={2-3},
  pages={129--275},
  year={2016},
  publisher={Now Publishers, Inc.}
}

@article{hekstra1989dependence,
  title={Dependence balance bounds for single-output two-way channels},
  author={Hekstra, Andries P and Willems, Frans MJ},
  journal={IEEE Trans. Inf. Theory},
  volume={35},
  number={1},
  pages={44--53},
  year={1989},
  publisher={IEEE}
}

@article{tandon2011dependence,
  title={Dependence balance based outer bounds for {G}aussian networks with cooperation and feedback},
  author={Tandon, Ravi and Ulukus, Sennur},
  journal={IEEE Trans. Inf. Theory},
  volume={57},
  number={7},
  pages={4063--4086},
  year={2011},
  publisher={IEEE}
}

@article{gohari2010information1,
  title={Information-theoretic key agreement of multiple terminals—Part {I}},
  author={Gohari, Amin  and Anantharam, Venkat},
  journal={IEEE Trans. Inf. Theory},
  volume={56},
  number={8},
  pages={3973--3996},
  year={2010},
  publisher={IEEE}
}

@article{gohari2010information2,
  title={Information-theoretic key agreement of multiple terminals—Part {II}: Channel model},
  author={Gohari, Amin and Anantharam, Venkat},
  journal={IEEE Trans. Inf. Theory},
  volume={56},
  number={8},
  pages={3997--4010},
  year={2010},
  publisher={IEEE}
}

@article{ozarow1980source,
  title={On a source-coding problem with two channels and three receivers},
  author={Ozarow, L},
  journal={Bell System Technical Journal},
  volume={59},
  number={10},
  pages={1909--1921},
  year={1980},
  publisher={Wiley Online Library}
}

@article{csiszar2008secrecy,
  title={Secrecy capacities for multiterminal channel models},
  author={Csisz{\'a}r, Imre and Narayan, Prakash},
  journal={IEEE Trans. Inf. Theory},
  volume={54},
  number={6},
  pages={2437--2452},
  year={2008},
  publisher={IEEE}
}

@inproceedings{tyagi2013secret,
  title={Secret key capacity for multipleaccess channel with public feedback},
  author={Tyagi, Himanshu and Watanabe, Shun},
  booktitle={Allerton Conf. Commun., Control, Computing},
  pages={1--7},
  year={2013}
}

@article{csiszar2012secrecy,
  title={Secrecy generation for multiaccess channel models},
  author={Csisz{\'a}r, Imre and Narayan, Prakash},
  journal={IEEE Trans. Inf. Theory},
  volume={59},
  number={1},
  pages={17--31},
  year={2012},
  publisher={IEEE}
}

@inproceedings{kramer2006dependence,
  title={Dependence balance and the {G}aussian multiaccess channel with feedback},
  author={Kramer, Gerhard and Gastpar, Michael},
  booktitle={IEEE Inf. Theory Workshop},
  pages={198--202},
  year={2006}
}

@INPROCEEDINGS{gastpar06,
  author={Gastpar, M. and Kramer, G.},
  booktitle={2006 Int. Zurich Seminar Commun.}, 
  title={On Cooperation Via Noisy Feedback},
  address={Zurich, Switzerland},
  year={2006},
  volume={},
  number={},
  pages={146-149},
  doi={10.1109/IZS.2006.1649101}}

@INPROCEEDINGS{gastpar06b,
  author={Gastpar, Michael and Kramer, Gerhard},
  booktitle={Asilomar Conf. Signals, Systems, Computers}, 
  title={On Noisy Feedback for Interference Channels}, 
  address={Asilomar, CA, USA},
  year={2006},
  volume={},
  number={},
  pages={216-220},
  doi={10.1109/ACSSC.2006.356618}}

@article{chan2015multivariate,
  title={Multivariate mutual information inspired by secret-key agreement},
  author={Chan, Chung and Al-Bashabsheh, Ali and Ebrahimi, Javad B and Kaced, Tarik and Liu, Tie},
  journal={Proc. IEEE},
  volume={103},
  number={10},
  pages={1883--1913},
  year={2015},
  publisher={IEEE}
}

@article{chan2014multiterminal,
  title={Multiterminal secret key agreement},
  author={Chan, Chung and Zheng, Lizhong},
  journal={IEEE Trans. Inf. theory},
  volume={60},
  number={6},
  pages={3379--3412},
  year={2014},
  publisher={IEEE}
}

@article{Ardestanizadeh,
  title={Wiretap channel with secure rate-limited feedback},
  author={Ardestanizadeh, Ehsan and Franceschetti, Massimo and Javidi, Tara and Kim, Young-Han},
  journal={IEEE Trans. Inf. Theory},
  volume={55},
  number={12},
  pages={5353--5361},
  year={2009},
  publisher={IEEE}
}

@article{wagner2008improved,
  title={An improved outer bound for multiterminal source coding},
  author={Wagner, Aaron B and Anantharam, Venkat},
  journal={IEEE Trans. Inf. Theory},
  volume={54},
  number={5},
  pages={1919--1937},
  year={2008},
  publisher={IEEE}
}

@article{liu2009capacity,
	author = {Liu, N. and Goldsmith, A.},
	journal = {IEEE Trans. Info. Theory},
	number = 11,
	pages = {4986--4994},
	publisher = {IEEE},
	title = {Capacity regions and bounds for a class of {Z}-interference channels},
	volume = 55,
	year = 2009}

@inproceedings{lau2024entropic,
  title={An Entropic Inequality in Finite Abelian Groups Analogous to the Unified {B}rascamp-{L}ieb and Entropy Power Inequality},
  author={Lau, Chin Wa and Nair, Chandra},
  booktitle={IEEE Int. Symp. Inf. Theory},
  pages={3588--3593},
  year={2024},
  organization={IEEE}
}

@article{wen2024new,
  title={A New Upper Bound for Distributed Hypothesis Testing Using the Auxiliary Receiver Approach},
  author={Wen, Zhenduo and Gohari, Amin},
  journal={arXiv preprint arXiv:2409.14148},
  year={2024}
}

@article{chen2025differential,
  title={A Differential Equation Approach to the Most-Informative Boolean Function Conjecture},
  author={Chen, Zijie and Gohari, Amin and Nair, Chandra},
  journal={arXiv preprint arXiv:2502.10019},
  year={2025}
}

@ARTICLE{Yu2018, 
author={L. {Yu} and H. {Li} and W. {Li}}, journal={IEEE Trans. Inf. Theory}, title={Distortion Bounds for Source Broadcast Problems}, year={2018}, volume={64}, number={9}, pages={6034-6053},}

@inproceedings{Yuval2024,
    author    = "Yuval Kochman",
    title     = "An Improved Upper Bound for Distributed Hypothesis Testing",
    year = 2024,
    booktitle = "2024 IEEE International Symposium on Information Theory (ISIT 2024)",
}

@ARTICLE{Cover79,
author={T. {Cover} and A. E. {Gamal}},
journal={IEEE Trans. Inf. Theory},
title={Capacity theorems for the relay channel},
year={1979},
volume={25},
number={5},
pages={572-584},
doi={10.1109/TIT.1979.1056084},
ISSN={1557-9654},
month={Sep.},}

@ARTICLE{kramer03,
  author={Kramer, G.},
  journal={IEEE Trans. Inf. Theory}, 
  title={Capacity results for the discrete memoryless network}, 
  year={2003},
  volume={49},
  number={1},
  pages={4-21},
  doi={10.1109/TIT.2002.806135}}

@ARTICLE{Kramer05,
author={G. {Kramer} and M. {Gastpar} and P. {Gupta}},
journal={IEEE Trans. Inf. Theory},
title={Cooperative strategies and capacity theorems for relay networks},
year={2005},
volume={51},
number={9},
pages={3037-3063},
doi={10.1109/TIT.2005.853304},
ISSN={1557-9654},
month={Sep.},}

@ARTICLE{Kramer-IT02,
  author={Kramer, G.},
  journal={IEEE Trans. Inf. Theory}, 
  title={Feedback strategies for white {G}aussian interference networks}, 
  year={2002},
  volume={48},
  number={6},
  pages={1423-1438},
  doi={10.1109/TIT.2002.1003831}}

@ARTICLE{Sula-IT20,
  author={Sula, Erixhen and Gastpar, Michael and Kramer, Gerhard},
  journal={IEEE Trans. Inf. Theory}, 
  title={Sum-Rate Capacity for Symmetric {G}aussian Multiple Access Channels With Feedback}, 
  year={2020},
  volume={66},
  number={5},
  pages={2860-2871},
  doi={10.1109/TIT.2019.2957808}}

@INPROCEEDINGS{kramer-ITW21,
  author={Kramer, Gerhard},
  booktitle={IEEE Inf. Theory Workshop}, 
  title={Feedback Gains for {G}aussian Massive Multiple-Access Channels}, 
  year={2021},
  address={Kanazawa, Japan},
  volume={},
  number={},
  pages={1-3},
  doi={10.1109/ITW48936.2021.9611423}}

@article{gohari2010information,
  title={Information-theoretic key agreement of multiple terminals—Part {I}},
  author={Gohari, A.  and Anantharam, V.},
  journal={IEEE Trans. Inf. Theory},
  volume={56},
  number={8},
  pages={3973--3996},
  year={2010},
  publisher={IEEE}
}

@PhdThesis{king78,
  Title = {Multiple Access Channels with Generalized Feedback},
  Author = {King, R. C.},
  School = {Dept. Elec. Eng., Stanford University, Stanforrd, CA, USA},
  Year = {1978},
  type = {Ph.D. dissertation}
}

@PhdThesis{aref80,
  Title = {Information {F}low in {R}elay {N}etworks},
  Author = {Aref, Mohammad Reza},
  School = {Dept. Elec. Eng., Stanford University, Stanforrd, CA, USA},
  Year = {1980},
  type = {Ph.D. dissertation}
}

@ARTICLE{ozarow-IT84,
  author={Ozarow, L.},
  journal={IEEE Trans. Inf. Theory}, 
  title={The capacity of the white {G}aussian multiple access channel with feedback}, 
  year={1984},
  volume={30},
  number={4},
  pages={623-629},
  doi={10.1109/TIT.1984.1056935}}

@PhdThesis{willems82,
  Title = {Information-theoretical results for the discrete memoryless multiple access channel},
  Author = {Willems, F. M. J.},
  School = {Katholieke Universiteit Leuven, Leuven,
  Belgium},
  Year = {1982},
  type = {Doctor in de
  Wetenschappen Proefschrift}
}

@ARTICLE{willems-IT82,
  author={Willems, F.},
  journal={IEEE Trans. Inf. Theory}, 
  title={The feedback capacity region of a class of discrete memoryless multiple access channels}, 
  year={1982},
  volume={28},
  number={1},
  pages={93-95},
  doi={10.1109/TIT.1982.1056437}}

@PhdThesis{Wigger08,
  Title = {Cooperation on the Multiple-Access Channel},
  Author = {Wigger, M.A.},
  School = {ETH Zurich, Switzerland},
  Year = {2008},
  type = {Doctoral thesis}
}

@ARTICLE{Sendonaris-COMM03,
  author={Sendonaris, A. and Erkip, E. and Aazhang, B.},
  journal={IEEE Trans. Commun.}, 
  title={User cooperation diversity. {P}art {I}. {S}ystem description}, 
  year={2003},
  volume={51},
  number={11},
  pages={1927-1938},
  doi={10.1109/TCOMM.2003.818096}}

@ARTICLE{Sendonaris-COMM03b,
  author={Sendonaris, A. and Erkip, E. and Aazhang, B.},
  journal={IEEE Trans. Commun.}, 
  title={User cooperation diversity. {P}art {II}. {I}mplementation aspects and performance analysis}, 
  year={2003},
  volume={51},
  number={11},
  pages={1939-1948},
  doi={10.1109/TCOMM.2003.819238}}

@Article{anv09,
  Title                    = {Gaussian Interference Networks: Sum Capacity in the Low-Interference Regime and New Outer Bounds on the Capacity Region},
  Author                   = {Annapureddy, V.S. and Veeravalli, V.V.},
  Journal                  = {IEEE Trans. Inf. Theory},
  Year                     = {2009},

  Month                    = {july},
  Number                   = {7},
  Pages                    = {3032--3050},
  Volume                   = {55},

  Doi                      = {10.1109/TIT.2009.2021380},
  ISSN                     = {0018-9448}
}

@article{ekrem2010secrecy,
  title={Secrecy in cooperative relay broadcast channels},
  author={Ekrem, Ersen and Ulukus, Sennur},
  journal={IEEE Transactions on Information Theory},
  volume={57},
  number={1},
  pages={137--155},
  year={2010},
  publisher={IEEE}
}

@article{bidokhti2016capacity,
  title={Capacity bounds for diamond networks with an orthogonal broadcast channel},
  author={Bidokhti, Shirin Saeedi and Kramer, Gerhard},
  journal={IEEE Transactions on Information Theory},
  volume={62},
  number={12},
  pages={7103--7122},
  year={2016},
  publisher={IEEE}
}

@Book{csk11,
  Title                    = {Information theory: Coding theorems for discrete memoryless systems},
  Author                   = {Imre Csiszar and Janos Korner},
  Publisher                = {Cambridge University Press},
  Year                     = {2011},
  Month                    = {1},
  Doi                      = {10.1017/CBO9780511921889},
  ISBN                     = {9780511921889}
}

@Book{elk11,
  Title = {Network Information Theory},
  Author = {{El Gamal}, Abbas and Kim, Young-Han},
  Publisher = {Cambridge University Pres},
  Year                     = {2011}
}

@Book{fano61,
  Title = {The Transmission of Information: A Statistical Theory of Communication},
  Author = {Fano, Robert},
  Publisher = {MIT Press},
  Year = {1961},
  Address = {Cambridge, MA, USA},
  ISBN = {0121984508}
}

@article{abin2025source,
  title={On the Source Model Key Agreement Problem},
  author={Abin, Hamidreza and Gohari, Amin},
  journal={arXiv preprint arXiv:2502.00294},
  year={2025}
}

@Article{gal74,
  Title                    = {Capacity and coding for degraded broadcast channels},
  Author                   = {Gallager, R G},
  Journal                  = {Probl. Peredac. Inform.},
  Year                     = {1974},
  Pages                    = {3--14},
  Volume                   = {10(3)}
}

@article{el2022strengthened,
  title={A strengthened cutset upper bound on the capacity of the relay channel and applications},
  author={El Gamal, Abbas and Gohari, Amin and Nair, Chandra},
  journal={IEEE Trans. Inf. Theory},
  volume={68},
  number={8},
  pages={5013--5043},
  year={2022},
  publisher={IEEE}
}

@article{kosut2025switched,
  title={Switched Feedback for the Multiple-Access Channel},
  author={Kosut, Oliver and Langberg, Michael and Effros, Michelle},
  journal={arXiv preprint arXiv:2501.14064},
  year={2025}
}

@inproceedings{kosut2023perfect,
  title={Perfect vs. independent feedback in the multiple-access channel},
  author={Kosut, Oliver and Effros, Michelle and Langberg, Michael},
  booktitle={IEEE Int. Symp. Inf. Theory},
  address={Taipei, Taiwan},
  pages={1502--1507},
  year={2023}
}

@ARTICLE{carleial1982multiple,
  author={Carleial, A.},
  journal={IEEE Trans. Inf. Theory}, 
  title={Multiple-access channels with different generalized feedback signals}, 
  year={1982},
  volume={28},
  number={6},
  pages={841-850},
  doi={10.1109/TIT.1982.1056587}}

@Article{ggny14,
  Title                    = {On {M}arton's Inner Bound and Its Optimality for Classes of Product Broadcast Channels},
  Author                   = {Y. Geng and A. Gohari and C. Nair and Y. Yu},
  Journal                  = {IEEE Trans. Inf. Theory},
  Year                     = {2014},

  Month                    = {Jan},
  Number                   = {1},
  Pages                    = {22-41},
  Volume                   = {60},

  Doi                      = {10.1109/TIT.2013.2285925},
  ISSN                     = {0018-9448}
}

@article{geng2014capacity,
  title={The capacity region of the two-receiver {G}aussian vector broadcast channel with private and common messages},
  author={Geng, Yanlin and Nair, Chandra},
  journal={IEEE Transactions on Information Theory},
  volume={60},
  number={4},
  pages={2087--2104},
  year={2014},
  publisher={IEEE}
}

@Article{gen14,
  Title                    = {The Capacity Region of the Two-Receiver {G}aussian Vector Broadcast Channel With Private and Common Messages},
  Author                   = {Y. Geng and C. Nair},
  Journal                  = {IEEE Trans. Inf. Theory},
  Year                     = {2014},
  Month                    = {April},
  Number                   = {4},
  Pages                    = {2087-2104},
  Volume                   = {60},
  Doi                      = {10.1109/TIT.2014.2304457},
  ISSN                     = {0018-9448}
}

@ARTICLE{1362897,
  author={Csiszar, I. and Narayan, P.},
  journal={IEEE Trans. Inf. Theory}, 
  title={Secrecy capacities for multiple terminals}, 
  year={2004},
  volume={50},
  number={12},
  pages={3047-3061},
  doi={10.1109/TIT.2004.838380}}

@inproceedings{vippathalla2021secret,
  title={Secret key agreement and secure omniscience of tree-{PIN} source with linear wiretapper},
  author={Vippathalla, Praneeth Kumar and Chan, Chung and Kashyap, Navin and Zhou, Qiaoqiao},
  booktitle={IEEE Int. Symp. Inf. Theory},
  pages={1624--1629},
  year={2021}
}

@ARTICLE{8995629,
  author={Zhou, Qiaoqiao and Chan, Chung},
  journal={IEEE Trans. Inf. Theory}, 
  title={Secret Key Generation for Minimally Connected Hypergraphical Sources}, 
  year={2020},
  volume={66},
  number={7},
  pages={4226-4244},
  doi={10.1109/TIT.2020.2971215}}

@INPROCEEDINGS{8437913,
  author={Chan, Chung and Mukherjee, Manuj and Kashyap, Navin and Zhou, Qiaoqiao},
  booktitle={IEEE Int. Symp. Inf. Theory}, 
  title={Multiterminal Secret Key Agreement at Asymptotically Zero Discussion Rate}, 
  year={2018},
  volume={},
  number={},
  pages={2654-2658},
  doi={10.1109/ISIT.2018.8437913}}

@ARTICLE{Gaarder-IT75,
  author={Gaarder, N. and Wolf, J.},
  journal={IEEE Trans. Inf. Theory}, 
  title={The capacity region of a multiple-access discrete memoryless channel can increase with feedback}, 
  year={1975},
  volume={21},
  number={1},
  pages={100-102},
  doi={10.1109/TIT.1975.1055312}}

@Article{han80,
  Title = {Multiple mutual informations and multiple interactions in frequency data},
  Author = {Han, Te Sun},
  Journal = {Inf. Control},
  Year = {1980},
  Volume = {46},
  Number = {1},
  Pages = {26--45}
}

@Article{kra04,
  Title                    = {Outer bounds on the capacity of {G}aussian interference channels},
  Author                   = {Kramer, G.},
  Journal                  = {Information Theory, IEEE Trans. on},
  Year                     = {2004},

  Month                    = {March},
  Number                   = {3},
  Pages                    = {581--586},
  Volume                   = {50},

  Doi                      = {10.1109/TIT.2004.825249},
  ISSN                     = {0018-9448}
}

@INPROCEEDINGS{Laneman-ISIT04,
  author={Laneman, J.N. and Kraner, G.},
  booktitle={IEEE Int. Symp. Inf. Theory}, 
  title={Window decoding for the multiaccess channel with generalized feedback}, 
  address={Chicago, IL, USA},
  year={2004},
  volume={},
  number={},
  pages={281},
  doi={10.1109/ISIT.2004.1365316}}

@ARTICLE{madiman-barron-IT07,
  author={Madiman, Mokshay and Barron, Andrew},
  journal={IEEE Trans. Inf. Theory}, 
  title={Generalized Entropy Power Inequalities and Monotonicity Properties of Information}, 
  year={2007},
  volume={53},
  number={7},
  pages={2317-2329},
  doi={10.1109/TIT.2007.899484}}

@Article{mok09,
  Title                    = {Capacity Bounds for the {G}aussian Interference Channel},
  Author                   = {Motahari, A.S. and Khandani, A.K.},
  Journal                  = {IEEE Trans. Inf. Theory},
  Year                     = {2009},
  Month                    = {feb.},
  Number                   = {2},
  Pages                    = {620--643},
  Volume                   = {55},

  Doi                      = {10.1109/TIT.2008.2009807},
  ISSN                     = {0018-9448}
}

@Article{skc09,
  Title                    = {A New Outer Bound and the Noisy-Interference Sum-Rate Capacity for {G}aussian Interference Channels},
  Author                   = {Shang, Xiaohu and Kramer, G. and Chen, Biao},
  Journal                  = {IEEE Trans. Inf. Theory},
  Year                     = {2009},

  Month                    = {feb.},
  Number                   = {2},
  Pages                    = {689--699},
  Volume                   = {55},

  Doi                      = {10.1109/TIT.2008.2009793},
  ISSN                     = {0018-9448}
}

@article{Chung-etal-SIAM88,
author = {Chung, F. R. K. and Füredi, Z. and Garey, M. R. and Graham, R. L.},
title = {On the Fractional Covering Number of Hypergraphs},
journal = {SIAM J. Discrete Math.},
volume = {1},
number = {1},
pages = {45-49},
year = {1988},
doi = {10.1137/0401005},
eprint = {https://doi.org/10.1137/0401005}
}

@INPROCEEDINGS{ElGamal-NTC81,
  author={El Gamal, Abbas},
  booktitle={IEEE Nat. Telecomm. Conf.}, 
  address={New Orleans, LA, USA},
  title={On information flow in relay networks}, 
  year={1981},
  volume={2},
  number={},
  pages={D4.1.1-D4.1.4},
  doi={10.1109/ALLERTON.2008.4797559}}

@article{McGill-P54,
author = {McGill, William J.},
title = {Multivariate information transmission},
journal = {Psychometrika},
volume = {19},
number = {2},
pages = {97-116},
year = {1954},
doi = {10.1007/BF02289159},
eprint = {https://doi.org/10.1007/BF02289159},
}
\end{document}